\documentclass[12pt]{article}

\usepackage{amsmath,amsthm,amssymb,euscript,epsf,epsfig}
\usepackage{array}
\usepackage{fancybox}
\usepackage{comment} 
\usepackage{color}

\usepackage{rotating}

\usepackage{standalone}

\usepackage{tikz,pbox}
\usetikzlibrary{shapes,arrows}
\usetikzlibrary{calc,decorations.markings}

\usepackage{cancel}
\usepackage{hyperref}

\makeatletter
\@addtoreset{equation}{section}
\makeatother
\renewcommand{\theequation}{\thesection.\arabic{equation}}

\def\one{{\hbox{ 1\kern-.8mm l}}}
\def\zero{{\hbox{ 0\kern-1.5mm 0}}}

\def\mC{ \mathbb{C}}

\def\s{ \sigma} 
\newcommand{\Dim}{ {\rm Dim } } 
\newcommand{\pairs}{ { \rm pairs} } 
\newcommand{\singlets}{ {\rm singlets} } 
\newcommand{\wchi}{ \widehat{\chi}}

  \def\cC{{\cal C}}
  
 \def\cH{{\cal H}} 
 \def\cK{{\cal K}} \def\cL{{\cal L}}
\def\cM{{\cal M}}  \def\cO{{\cal O}}

 \def\cZ{{\cal Z}}

\def\Sym{ \hbox{Sym} } 
\def\Orb{ {\rm Orb}} 
\def\Aut{ {\rm Aut} } 
\def\Rib{ {\rm Rib}}
\newcommand{\id}{\rm id}

\newtheorem{lemma}{Lemma}
\newtheorem{corollary}{Corollary}

\newtheorem{theorem}{Theorem}
\newtheorem{proposition}{Proposition}

\newcommand{\be}{\begin{equation}}
\newcommand{\ee}{\end{equation}}
\newcommand{\beq}{\begin{equation}}
\newcommand{\eeq}{\end{equation}}
\newcommand{\bea}{\begin{eqnarray}\displaystyle}
\newcommand{\eea}{\end{eqnarray}}

\newcommand{\cred }{\color{red}}
\def\wc{\widehat \chi} 
\newcommand{\PRIMESDIFFS}{ \texttt{PrimesDiffs} }

\newcommand{\bdel}{ {\boldsymbol{\delta}} }

\def\mP{ \mathbb{P}}  
\def\s{ \sigma } 
\def\g { \gamma }

\def\Max{ {\rm{Max}} } 
\def\Min{ {\rm{Min}} }

\newcommand{\mZ}{ \mathbb{Z} } 
\newcommand{\mR}{\mathbb{R}}

\begin{document}

\begin{flushright}
QMUL-PH-20-21
\end{flushright}

\bigskip

\begin{center}

{\Large \bf Quantum mechanics of bipartite ribbon graphs: 

 }
 \medskip 

{ \Large \bf  Integrality, Lattices and Kronecker coefficients. }

\bigskip

Joseph Ben Geloun$^{a,c,*}$
 and Sanjaye Ramgoolam$^{b , d ,\dag}  $

\bigskip

$^a${\em Laboratoire d'Informatique de Paris Nord UMR CNRS 7030} \\
{\em Universit\'e Paris 13, 99, avenue J.-B. Clement,
93430 Villetaneuse, France} \\
\medskip
$^{b}${\em School of Physics and Astronomy} , {\em  Centre for Research in String Theory}\\
{\em Queen Mary University of London, London E1 4NS, United Kingdom }\\
\medskip
$^{c}${\em International Chair in Mathematical Physics
and Applications}\\
{\em ICMPA--UNESCO Chair, 072 B.P. 50  Cotonou, Benin} \\
\medskip
$^{d}${\em  School of Physics and Mandelstam Institute for Theoretical Physics,} \\   
{\em University of Witwatersrand, Wits, 2050, South Africa} \\
\medskip
E-mails:  $^{*}$bengeloun@lipn.univ-paris13.fr,
\quad $^{\dag}$s.ramgoolam@qmul.ac.uk

\begin{abstract}

We define solvable quantum mechanical systems on a Hilbert space spanned by bipartite ribbon graphs with a fixed number of  edges. The Hilbert space is also an associative algebra, where the product is derived from permutation group products. The  existence  and structure of this Hilbert space algebra has a number of consequences.  
The algebra product, which  can be expressed in terms of  integer ribbon graph reconnection coefficients,  is used to define solvable Hamiltonians with eigenvalues expressed in terms of  normalized characters of symmetric group elements and degeneracies given in terms of Kronecker coefficients, which are  tensor product multiplicities of symmetric group representations. The square of the Kronecker coefficient for a triple of Young diagrams is shown to be equal to the dimension of a  sub-lattice in the lattice of ribbon graphs.  This leads to  an answer to the long-standing question of a combinatorial interpretation of  the Kronecker coefficients. 
As avenues for future research, we discuss applications of the ribbon graph quantum mechanics in algorithms for quantum computation. We also describe  a quantum membrane interpretation of these quantum mechanical systems.

\end{abstract}

\end{center}


\newpage 


\section{Introduction}

Permutation centralizer algebras (PCAs) \cite{PCAMultMat} have been found as an underlying structure which organizes the $N$-dependences of   multi-matrix correlators in   super-Yang Mills theories with $U(N)$ gauge symmetry  \cite{BBFH,Brauer,BHR0711,BCD0801,BCD0805,BHR0806,EHS,QuivCalc,YusukeTFT1,YusukeTFT2}. These correlators are of  interest in 
generalizing beyond the half-BPS sector the link between BPS correlators and Young diagrams \cite{CJR2001}  in the AdS/CFT correspondence \cite{malda,gkp,witten}.  

Permutation methods and PCAs also played a role in the enumeration of observables and the  
computation of correlators in Gaussian tensor models \cite{Sanjo,PCAKron}, which have been studied in the 
context of applications of tensor models to random geometries and holography \cite{ADJ91, Gurau1102, Rivasseau:2016wvy,  Guraubook, Witten:2016iux} (see reviews in \cite{Delporte:2018iyf,KlebanovTasi}). An important observation  from \cite{Sanjo,PCAKron} is  
that 3-index tensor observables of degree $n$ in a complex tensor model with $U(N)^{ \times 3}$ symmetry can  be counted using 3-tuples of permutations in $ S_n$, subject to an equivalence relation defined by left and right multiplication by permutations in $S_n$.  A gauge-fixed version  of this formulation was described  where we have pairs of permutations, subject to an equivalence relation defined using simultaneous conjugation of the pair by a permutation in $S_n$. These equivalence classes of permutation pairs are known to count bipartite ribbon graphs with $n$ edges (a textbook reference for this subject is \cite{LandoZvonkin}). The permutation equivalence classes form an associative algebra, denoted $\cK ( n )$, with a symmetric non-degenerate bilinear form \cite{PCAKron}. As a semi-simple algebra, according to the Wedderburn-Artin theorem, $\cK(n)$  is  isomorphic to a direct sum of matrix algebras \cite{GoodWall}. The  explicit isomorphism was  constructed using Clebsch-Gordan coefficients of the symmetric group \cite{PCAMultMat,PCAKron}.   The matrix basis for the algebra 
takes the form of $Q^{ R_1 R_2 R_3  }_{ \tau_1 , \tau_2 } $, where $R_1, R_2, R_3 $ are Young diagrams or partition of $n$ and $\tau_i$, $i=1,2,$ range over  Clebsch-Gordan multiplicities, also known as Kronecker coefficients (the explicit formula is given in \cite{PCAMultMat} and developed in detail in \cite{PCAKron}). Further investigations of tensor models from this algebraic perspective are in \cite{deMelloKoch:2017bvv, Avohou:2019qrl,BenGeloun:2020lfe, Diaz:2017kub,Diaz:2018xzt,Diaz:2018zbg, Itoyama:2017wjb, Itoyama:2019oab, Itoyama:2019uqv, Amburg:2019dnj,Robert2nd}. 
A known connection between bipartite ribbon graphs and Belyi maps  \cite{LandoZvonkin,schneps} gives a topological version of gauge-string duality between tensor models and string theory \cite{Sanjo}, generalizing analogous  correspondences between two-dimensional Yang Mills theory and topological string theory \cite{GRTA,CMR,Horava}.

AdS/CFT holography gives a map between half-BPS states in $U(N)$ Yang-Mills theory at large $N$ and the corresponding  space-time geometries \cite{LLM}. The study of the half-BPS sector 
as a toy model for questions in the black hole information loss problem \cite{IILoss}
raised a question on how restricted sets of $U(N)$  Casimirs can distinguish Young diagrams with a fixed number
 $n$  (equal to the energy of the BPS state) of boxes. This question is related, by Schur-Weyl duality,  to properties of the group algebra of $S_n$ and was studied from this perspective in  \cite{KR1911}. A key role in this investigation was played by central elements $T_k$  in the group algebra $ \mC ( S_n)$ associated with permutations having cycle structure consisting of a single cycle of  length $k$ ( for some $2 \le k \le n $ ) and remaining cycles of length $ 1$. 

In addition to these developments from theoretical physics, the investigations in this paper have been guided by the mathematical problem of determining whether there are combinatorial objects which are counted by Kronecker coefficients. While a combinatorial construction of Littlewood-Richarson coefficients, another representation theoretic multiplicity, associated with triples of Young diagrams is well known, it has been a long-standing question whether there  exists a family of combinatorial objects, for each triple of Young diagrams, such that the  combinatorial objects are enumerated by Kronecker coefficients. This problem was posed in \cite{Murnaghan} and placed in  the context of a number of positivity problems in representation theory in  \cite{StanleyKron} and is discussed in recent papers, e.g. \cite{PPV,Manivel} . This mathematical  question which may appear, at least at first sight to many physicists, to be a  somewhat esoteric  question, has inspired substantial recent activity and progress at the intersection of computational complexity theory, quantum information theory and representation theory. We will not attempt to give a summary of this thriving area of research, but will point the reader to some papers which give a flavour of this field  \cite{PPV,Manivel,MulVar,BCI2011,IkMuWa-VanKron,PakPanova}.

 A way to understand the problem is to compare two known computations in representation theory. The computation of characters $ \chi_R ( \sigma ) $ of a permutation $ \sigma \in S_n$ in a representation associated to Young diagram $R$ with $n$ boxes can be done  by using  the Murnaghan-Nakayama rule \cite{MurnaghanOnReps,Nak}. This  can be phrased in terms of the counting of a certain pattern of labellings of the boxes in $R$ by numbers according to a rule determined by the cycle structure of $ \sigma $ (see for example \cite{StanleyEnumComb2}\cite{Wiki-MN}).  In this construction, it is clear why the outcome is an integer - which is a   somewhat special property of symmetric group characters, a property not shared by generic finite groups. The Kronecker coefficient can be computed using the formula 
\bea 
C ( R_1, R_2, R_3  ) = { 1 \over n! } \sum_{ \sigma \in S_n } \chi_{R_1} ( \sigma ) \chi_{R_2} ( \sigma ) \chi_{R_3}( \sigma  ) 
\eea
In this formula, it is not clear why the sum over all the conjugacy classes in $S_n$ for general $n$  ends up giving an outcome which  is a non-negative integer - although from the representation theory definition as the number of invariants in the tensor product of  $ R_1 \otimes R_2 \otimes R_3 $, it is clear why this is the case.  A combinatorial interpretation should give a new way to make it manifest that  $ C (  R_1, R_2, R_3 )$ is a non-negative integer. 

 The following formula which has played a role in counting tensor model invariants shows  that bipartite ribbon graphs (also called ribbon graphs for short in this paper)  hold some promise of progress on this problem. It is known that the total number of bipartite ribbon graphs with $n$ edges is equal to the sum of squares of Kronecker coefficients  \cite{Sanjo,PCAMultMat,Diaz:2017kub,PCAKron} 
\bea 
|\Rib( n )| = \sum_{  R_1, R_2, R_3  \vdash n } C (  R_1, R_2, R_3 )^2  
\eea
This formula shows  that the sum of squares of Kronecker coefficients does have a combinatorial and geometric interpretation. Bipartite ribbon graphs have an elegant group theoretic characterisation in terms of pairs of permutations with an equivalence under simultaneous conjugation. A natural question is : Is it possible to  refine this link to give an interpretation of a fixed $ C ( R_1, R_2, R_3 )^2 $, and a fixed $ C ( R_1 , R_2 , R_3 )$,  in terms of ribbon graphs? We would like an interpretation which makes the non-negative integer property  of 
the Kronecker coefficients manifest.  And are  there  combinatorial algorithms based on this interpretation for computing Kronecker coefficients?

The algebras $ \cK(n)$, and analogous algebras related to Littlewood-Richardson coefficients, have been studied
in the theoretical physics literature primarily as a tool to understand the structure of the space of gauge invariant observables and their correlators  in matrix/tensor models and in AdS/CFT (see \cite{RamgoComb} for a  short review). In this paper, motivated by the mathematical question of a combinatorial interpretation of Kronecker coefficients and the connections of this question to quantum information and complexity theory, we introduce a new physical perspective on these algebras.  We propose  that studying solvable quantum mechanics models on algebras such as $ \cK(n)$, which  are related to interesting combinatorial objects (in this case bipartite ribbon graphs) having elegant descriptions in terms of symmetric groups (in this case permutation pairs subject to an equivalence generated by conjugation with a permutation), can be a fruitful avenue to explore interesting interfaces between physics, mathematics and computational complexity theory.

Section \ref{QMribb} develops the quantum mechanics on $ \cK( n )$. $ \cK(n)$ is a subspace of $ \mC ( S_n ) \otimes \mC ( S_n) $ which is invariant under conjugation by $ \gamma \otimes \gamma$ for $ \gamma \in S_n$. As a vector space, it has two interesting bases. There is a basis $E_r$ of elements labelled by an index $r$ ranging over equivalence classes of pairs $ ( \sigma_1 , \sigma_2 ) \in S_n \times S_n $, with the equivalence relation 
\bea 
( \sigma_1 , \sigma_2 ) \sim ( \gamma \sigma_1 \gamma^{-1} , \gamma \sigma_2 \gamma^{-1} ) \,,
\eea
defined using $ \gamma \in S_n$. We refer to this basis as the geometric ribbon graph basis. There is another basis labelled by triples of Young 
diagrams $ ( R_1 , R_2 , R_3 )$, where each Young diagram has $n$ boxes, such that the Kronecker coefficient $ C ( R_1 , R_2 , R_3 )$ is non-zero. We refer to this as the Fourier basis for $ \cK(n)$. In section \ref{QMribb1}, we  review (from \cite{Sanjo,PCAMultMat,PCAKron}) the formula \eqref{qbasis} for the Fourier basis elements in terms of matrix elements and Clebsch-Gordan coefficients of $ S_n$.  The Fourier basis also makes the Wedderburn-Artin decomposition of $\cK(n)$ into matrix algebras manifest. We define a natural inner product on $ \cK(n)$ inherited from $ \mC ( S_n) \otimes \mC ( S_n)$ and prove that $\cK(n)$ is a Hilbert space (Proposition \ref{PropcKnHilb}).  We prove that the product structure on $ \cK(n)$ in the geometric ribbon graph basis is given by integers (section \ref{integral}). 

The fact that $ \cK(n)$ is a vector space as well as an algebra (i.e. vector space equipped with an associative product) with a known Wedderburn-Artin decomposition can be exploited to write down interesting solvable Hamiltonians for quantum mechanical systems having $ \cK(n)$ as a Hilbert space. We introduce a set of Hermitian operators $T_k^{(i)}$ on $ \cK(n)$  which are central elements of $ \cK(n)$ and act on $ \cK(n)$ using the product operation in the algebra. The indices take values  $ i \in \{ 1 , 2, 3\} $ and $ k \in \{ 2 , 3, \cdots  \widetilde k_*  \}$. The number $ \widetilde k $ is chosen to obey $ \widetilde k_* \ge k_* ( n )$, where $k_*(n)$ is an integer between $2$ to $n$. 
$k_*(n)$  is defined \cite{KR1911}   as the minimum integer such that the central elements $T_k$ in $ \mC ( S_n)$ with $k$ ranging in $ \{ 2 , 3 , \cdots , k_* (n) \}$ generate the centre.  
The precise definition of the operators $T_k^{(i)}$, which we call reconnection operators, is given in section \ref{sec:CentreKnReconn}. It is shown (Proposition \ref{PropIntTk}) that the matrix elements of these operators in the geometric ribbon graph basis are non-negative integers. 

In section \ref{sec:IntegerMatricesKron}, we introduce  the notion of the Fourier subspace of $ \cK(n)$ associated with a triple of Young diagrams $ ( R_1 , R_2  , R_3 )$. This subspace has dimension $ C ( R_1 , R_2 , R_3 )^2$. 
Proposition \ref{propTaQo} shows that the Fourier basis elements are eigenvectors of the reconnection operators, with eigenvalues given by normalized characters of symmetric groups. Proposition \ref{PropTkFS} shows that the eigenvalue sets of reconnection operators chosen with $ k \in \{ 2  , \cdots , \widetilde k_*  \}$ can be used to distinguish Fourier subspaces associated with distinct triples of Young diagrams.  These results are used  (section \ref{sec:reconn}) 
to construct for each $n$, and each triple $ ( R_1 , R_2 , R_3)$, a rectangular matrix of   integers having a null space which spans the Fourier subspace of the specified triple. Section \ref{sec:SquareIntMat}
constructs Hamiltonians as linear combinations of the reconnection matrices,  which are square (non-negative) integer matrices in the geometric basis and distinguish Fourier subspaces with distinct Young diagram triples. Using   Proposition \ref{propTaQo}, the eigenvalues of these Hamiltonians are expressed as linear combinations of normalized symmetric group characters. The eigenspaces  for distinct eigenvalues are the Fourier subspaces for distinct Young diagram triples.

The realisation of Fourier subspaces in $ \cK(n )$  labelled by Young diagram triples $ (  R_1 , R_2, R_3 )  $ as eigenspaces of integer reconnection matrices is thus one of two important inputs in our discussion. It means that while the formula \eqref{qbasis} for Fourier basis elements uses 
detailed representation theoretic data such as matrix elements of permutations in some chosen basis for symmetric group representations along with Clebsch-Gordan coefficients, there is a new approach to the Fourier subspace of a triple of Young diagrams based on integer reconnection matrices.  Now generic integer matrices do not necessarily have integer or rational eigenvalues (see for example \cite{Estes}). For the reconnection matrices at hand however we know, using symmetric group representation theory (Proposition \ref{propTaQo} along with Lemma \ref{IntegralityOfNormChar}),  that the eigenvalues are integers. These eigenvalues are known to be calculable using combinatorial algorithms, notably the Murnaghan-Nakayama rule.
Thus, we are able to replace the more obvious (but computationally expensive)  computation of the Fourier subspace using direct implementation of the formula \eqref{qbasis} with the calculation of null spaces of integer matrices which  takes two combinatorial inputs: the combinatorics of reconnection matrices and the Murnaghan-Nakayama algorithm. 
 This allows us to express the problem of finding the Fourier subspaces of Young diagram triples as a question about null spaces of  integer matrices. This in turn  allows us to access results from the subject of integer matrices and lattice  algorithms. 

Section \ref{sec:YDHNF} recalls a key result from the integer matrices and lattice algorithms. Any integer matrix, square or rectangular, has a unique Hermite normal form (HNF). There are standard algorithms in computational number theory for finding the HNF (see e.g. \cite{Cohen,Schrijver,MicciancioBA}) and such algorithms are  also accessible in group theoretic software such as SAGE or GAP \cite{GAP}.  A consequence is that, for the Fourier subspaces  associated with Young diagram triples 
defined in section  \ref{sec:IntegerMatricesKron}, there are bases which are integer linear combinations of the geometric ribbon graph basis vectors. For each triple $ (R_1 , R_2, R_3 )$, given a choice of the rectangular matrix (which can be specified using a choice of $ \widetilde k_* $ as in section \ref{sec:reconn})  or square matrix (specified using a Hamiltonian as in  section \ref{sec:SquareIntMat}), any HNF algorithm leads to 
 a  list of linearly independent integer null vectors, which are $ C ( R_1 , R_2 , R_3 )^2 $ in number.  This list of integer null vectors specifies a sub-lattice in the lattice 
 $ \mZ^{ |\Rib(n)|}$  in $ \cK ( n )$ generated by all integer linear combinations of the geometric ribbon graph vectors. This provides (Theorem \ref{theo:C2} and Corollary \ref{CorrConst})  a positive  answer to the questions of a combinatorial interpretation and  construction for the square of the Kronecker coefficient.  

It is natural to ask if a construction of $C ( R_1 , R_2 , R_3 )$ rather than its square can be given along these lines. To this end, we consider an operation on bipartite ribbon graphs, which has previously been studied in the context of Belyi maps \cite{Jones}. In the permutation pair description of ribbon graphs, this  operation amounts to inverting both permutations. In section \ref{app:conjugate}  we study a linear involution $S$ (also called conjugation)  on $ \cK(n)$ defined using this inversion. Comparing the action of the involution on the ribbon graph basis with its action on the Fourier basis elements \eqref{qbasis} leads to the result that the sum of Kronecker coefficients is equal to the number of self-conjugate ribbon graphs. Considering linear operators acting on $ \cK(n)$ constructed from 
the reconnection operators $ T_k^{(i)} $ as well as the conjugation operator $S$ leads to sub-lattices of dimension $ C ( R_1 , R_2 , R_3 ) ( C ( R_1 , R_2 , R_3 ) +1 ) /2 $ ,  $ C ( R_1 , R_2 , R_3 ) ( C ( R_1 , R_2 , R_3 ) - 1 ) /2 $, both of which come equipped with a list of linearly independent integer basis vectors from an HNF construction.
  Choosing an injection from the set of  basis vectors of  the smaller sub-lattice  into the set of basis vectors of the larger sub-lattice   yields a subset of basis vectors of the larger sub-lattice, which equal $ C ( R_1 , R_2 , R_3)$ in number.    This realises $ C( R_1 , R_2 , R_3 )$ as the dimension of a sub-lattice in $\mZ^{ |\Rib(n)|}  $.

In the concluding section we give a summary of our results. While the content of this paper is primarily mathematical, its motivations come from the physics of strings and tensor models. 
The concluding section thus includes a description of future research directions based on the links to physics. There is a more extended discussion setting up the first steps for these future directions in the  arxiv version V2 of this paper \cite{BenGeloun:2020yau}. The appendices give some detailed steps in the proofs and examples of results from the computation of Fourier basis vectors using reconnection operators. The last appendix gives key parts  of the GAP code used.

\section{Quantum Mechanics of ribbon graphs: commuting Hamiltonians from centres of algebras $\cK(n)$ } 
\label{QMribb}

In this section we set up the quantum mechanics of bipartite ribbon graphs using their description in terms of  permutation groups. We introduce  the space of states, two bases for the space (a geometric basis and a Fourier basis),  an inner product and Hermitian operators  on the state space which have eigenvalues expressible in terms of normalized symmetric group characters.

\subsection{Review of previous results on the algebra $\cK(n)$ of bipartite ribbon graphs} 
\label{QMribb1}

We give an overview of the description of  bipartite ribbon graphs in terms
of symmetric groups. A useful textbook reference is \cite{LandoZvonkin} which gives references to 
the original mathematical literature. We will also be making extensive use of  formulae from 
the representation theory of symmetric groups. A mathematical physics reference is \cite{Hamermesh}. 
The key formulae are summarised  the appendices of \cite{PCAKron}.

\subsubsection{Counting bipartite ribbon graphs.}
\label{sec:Counting} 
A  bipartite ribbon graph, also called a hypermap,  is a graph embedded on a    two-dimensional  surface   with black and  white vertices, such that edges connect black to white vertices and cutting the surface along the edges leaves a disjoint union of regions homeomorphic to open discs.  Bipartite ribbon graphs, denoted ribbon graphs for short in this paper,  with $n$ edges can be  described using permutations of $ \{ 1, 2, \cdots , n \}$ forming the symmetric group $S_n$. Label the edges with integers $ \{ 1, 2,  \cdots , n\} $. Reading  the edges around the black vertices following a chosen orientation on the surface gives the cycles of a permutation $ \tau_1$,  while the white vertices similarly give a permutation $ \tau_2 $.  Relabelling the edges, $ i \rightarrow \mu ( i ) $ using    $\mu  \in S_n $, amounts to conjugating the pair $ ( \tau_1 , \tau_2 ) \rightarrow ( \mu \tau_1  \mu^{-1} , \mu \tau_2 \mu^{-1} ) $. Distinct ribbon graphs are thus equivalence classes  of  pairs $(\tau_1, \tau_2) \in S_n \times S_n$  under the equivalence relation 
\bea\label{adj}
  ( \tau_1 , \tau_2  ) \sim  ( \tau_1' , \tau_2' ) \,  ~~~ \hbox{
iff} ~~~   \exists \mu \in S_n \,,  ~~~ ( \tau_1' , \tau_2' ) =  ( \mu \tau_1 \mu^{-1} , \mu \tau_2 \mu^{-1} ) 
\eea 
The set of permutation pairs within a fixed equivalence class forms an orbit for the action of $S_n$ on $ S_n \times S_n$ given in \eqref{adj}.  We define $ \Rib(n)$ to be the set of equivalence classes, or the set of orbits. There are  commands in group theoretic software GAP \cite{GAP}   that directly 
generate these orbits for any $n$, see \verb|RibbSetFunction(n)| appendix \ref{app:gap}.  
As an example, consider the case $n=3$. These are the 11 ribbon graphs  shown in  Figure \ref{fig:ribb}. The label appearing below each ribbon graph
is an index running from $1$  to $11$. The sole non-planar (genus one)  ribbon graph
is the equivalence class containing the pair $[(123), (123)]$.

 \begin{figure}[h]
 \begin{center}
     \begin{minipage}[h]{.8\textwidth}\centering
\includegraphics[angle=0, width=12cm, height=8cm]{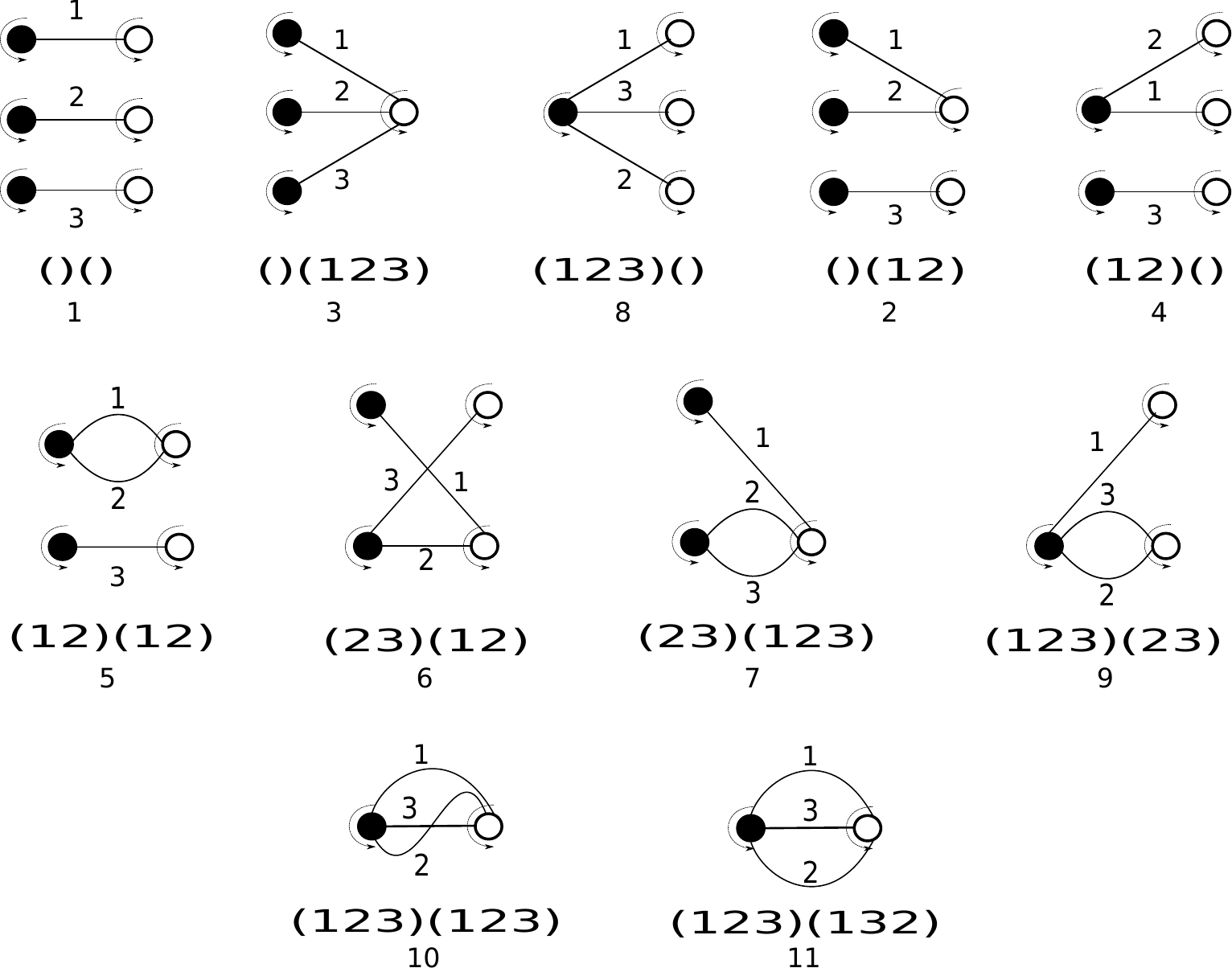}
\vspace{0.3cm}
\caption{ {\small Bipartite ribbon graphs  with  $n=3$ edges}} 
\label{fig:ribb}
\end{minipage}
\end{center}
\end{figure}

The counting of  ribbon graphs is also related to 
the counting of bipartite  graphs with $n$ trivalent vertices with three incoming colored edges  and $n$  trivalent vertices  with three outgoing colored edges \cite{Sanjo,PCAKron}. This counting problem corresponds to equivalence classes of triples $ ( \s_1 , \s_2 , \s_3 ) \sim ( \mu_1 \s_1 \mu_2 , \mu_1 \s_2 \mu_2 , \mu_1 \s_3 \mu_3 ) $ for $ \s_i \in S_n $ and $ \mu_1 ,  \mu_2 \in S_n $.  
 In turn this counting also gives the number of linearly independent  degree $n$ polynomial functions of  tensor variables $ \Phi_{ i_1 , i_2  ,  i_3 } $ and $ \bar \Phi^{ \bar i_1 , \bar i_2 , \bar i_3  }$ invariant the action of $U(N)^{ \times 3} $, for $ n \le N  $ and  with $ (i_1, i_2 , i_3 )  $ transforming as the fundamental of  the unitary group and $ ( \bar i_1 , \bar i_2 . \bar i_3 )  ,  $ transforming  in the  anti-fundamental.  Our focus in this paper will be on ribbon graphs, and we will discuss tensor model observables further in the outlook section \ref{ccl}.

\subsubsection{The permutation centralizer algebra (PCA)  $\cK(n)$ and its geometric basis. } \label{sec:geombasis} 
 
Introducing  the group algebra $ \mC ( S_n ) $, 
consider the elements of $ \mC ( S_n ) \otimes_{ \mC }  \mC ( S_n )$, written more simply 
$ \mC ( S_n ) \otimes \mC ( S_n ) $,   obtained by starting with a tensor product  $ \sigma_1 \otimes \sigma_2 $ 
and summing all their diagonal conjugates as 
\be \label{invgam}
\sigma_1 \otimes \sigma_2 \rightarrow \sum_{ \gamma \in S_n } \gamma \sigma_1 \gamma^{-1} \otimes 
\gamma \sigma_2 \gamma^{-1} 
\ee
Two  pairs  $ ( \sigma_1 , \sigma_2 ) $  and $ ( \sigma_1' ,  \sigma_2' ) $ related by the equivalence \eqref{adj} produce the same sum.  
Now, consider the $\mC$-vector subspace $\cK(n) \subset  \mC ( S_n ) \otimes \mC ( S_n )$ 
spanned by all $\sum_{ \gamma \in S_n } \gamma \sigma_1 \gamma^{-1} \otimes  \gamma \sigma_2 \gamma^{-1} $, 
$\s_1$ and $\s_2 \in S_n$: 
\be
\cK(n) = {\rm Span}_{\mC}\Big\{    
\sum_{ \gamma \in S_n } \gamma \sigma_1 \gamma^{-1} \otimes  \gamma \sigma_2 \gamma^{-1} , \; \s_1, \s_2 \in S_n
\Big\}
\label{graphbasis}
\ee
The dimension of  $\cK(n) $ is equal to the number of ribbon graphs with $n$ edges, i.e. $|\Rib(n)|$. In \cite{PCAKron},  it is  shown that $\cK(n) $ is an associative algebra, with the product being inherited from $ \mC ( S_n) \otimes \mC ( S_n)$.  
$ \cK(n)$ is a  permutation centralizer algebra (PCA) - a subspace of an algebra with basis given by permutations forming a group (here permutation pairs $ ( \s_1 , \s_2 ) $ forming the group $S_n \times S_n$), which commutes with a subgroup of the permutations, here $( \gamma , \gamma )$ forming the diagonal subgroup $ S_n \subset S_n \times S_n  $.   $\cK(n) $ is also semi-simple: it has a  non-degenerate  symmetric bilinear pairing 
given  by 
\bea\label{Wpairing}  
\bdel_2 : \mC(S_n)^{\otimes 2}\times  \mC(S_n)^{\otimes 2} \to \mC 
\eea
where 
\be\label{bdelta}
\bdel_2 (  \otimes_{i=1}^2 \s_i ; \otimes_{i=1}^2 \s'_i ) = 
\prod_{i=1}^2 \delta (\s_i\s'^{-1}_i) 
\ee
which extends to linear combinations with complex coefficients. 
Semi-simplicity implies that, by the Wedderburn-Artin theorem \cite{GoodWall,Ram},   $\cK(n) $ admits a decomposition in simple matrix algebras.  This decomposition is made manifest using what we denote as the Fourier basis,  which we discuss shortly in section \ref{Fourier}.

Start with a ribbon graph with label $r \in   \{ 1 , \cdots , |\Rib(n)| \}$.  As discussed in section \ref{sec:Counting}, the set of ribbon graphs is in 1-1 correspondence with orbits 
of the action of $S_n $ on $ S_n \times S_n$, by conjugation as given in \eqref{adj}. 
Pick a pair of permutations  $ ( \tau_1^{(r)}  , \tau_2^{(r)}  )$ among the permutation pairs 
representing the ribbon graph $r$.  The  orbit $\Orb(r)$  is the set of elements 
in $ S_n \times S_n $ which can be written as 
$ ( \mu \tau_1^{(r)}  \mu^{-1}   , \mu \tau_2^{(r)}  \mu^{-1}  ) $ for some $ \mu \in S_n $. 
Hence, consider the basis element in $\cK(n)$  associated with $r$ as: 
\bea 
\label{classr}
E_r=
{ 1 \over n! }  \sum_{ \mu \in S_n } \mu \tau_1^{(r)}  \mu^{-1}  \otimes \mu \tau_2^{(r)}  \mu^{-1} 
\eea
Let $ \Aut ( \tau_1^{(r)}  , \tau_2^{(r)}  ) $ be the subgroup of $S_n$ which leaves fixed the pair
$( \tau_1^{(r)}  , \tau_2^{(r)}  ) $. 
The order of this group is $ | \Aut ( \tau_1^{(r)}  , \tau_2^{(r)}  )| $ and is independent of 
the choice of representative, so we can write this as $ | \Aut (r) |$. 
The orbit-stabilizer theorem (see for example \cite{Cameron})  gives an isomorphism between  $ \Orb (r) $  and the coset 
$S_n / \Aut ( \tau_1^{(r)}  , \tau_2^{(r)}  )$.  Let $a$ be a label for the distinct 
permutation pairs in $ \Orb (r)$: 
\bea 
E_r &=& { 1 \over n! }  \sum_{ \mu \in S_n } \mu \tau_1^{(r)}  \mu^{-1}  \otimes \mu \tau_2^{(r)}  \mu^{-1}  \cr 
&  =& { | \Aut (r ) | \over n!  } 
 \sum_{ a \in \Orb ( r  ) } \tau^{(r)}_1 (a ) \otimes \tau^{(r)}_2 ( a)   \cr 
 & =& { 1 \over | \Orb ( r ) | }   \sum_{ a \in \Orb ( r ) } \tau^{(r)}_1 (a ) \otimes \tau^{(r)}_2 (a) 
 \label{Err}
\eea
Both  these expressions for $E_r$ will be useful. We will refer to the $E_r$ as the geometric or ribbon graph basis vectors for $ \cK (n)$. 

The pairing \eqref{Wpairing} evaluated on this basis is 
\bea\label{orthobasis}
&& \bdel_2 (  E_r, E_s ) = \frac{1}{n!}  \sum_{\g}
\delta(  \sigma_1^{(r)}   \gamma ( \sigma_1^{(s)})^{-1}\gamma^{-1} )
\delta(  \sigma_2^{(r)}   \gamma  ( \sigma_2^{(s)})^{-1} \gamma^{ -1} )
\cr
&&
 = 
\frac{1}{n!} |\Aut( r ) | \;      \delta_{sr} =  { 1 \over | \Orb (r) | } \delta_{ rs} \cr 
&& 
\eea
The basis vectors associated with distinct orbits $ r \ne s $ are orthogonal.

\subsubsection{A Fourier basis for $\cK(n)$ }
\label{Fourier} 

The number of bipartite ribbon graphs with $n$ edges,  which is the dimension of $ \cK ( n )$, can be given as a sum of partitions of $n$  \cite{Sanjo,PCAKron} or as a sum over triples $ R_1 , R_2 , R_3 $  of irreducible  representations (irreps)  of $S_n$ : 
\bea 
|\Rib ( n ) | = \Dim ( \cK ( n ) ) = \sum_{ R_1 , R_2 , R_3 \vdash n } C ( R_1 , R_2 , R_3 )^2 = \sum_{ p \vdash n }\Sym ( p ) 
\eea
$R_1,R_2,R_3$ are partitions of $n$ (denoted by $R_i \vdash n$) which correspond to  Young diagrams with $n$ boxes.  We will denote their dimension as $d(R_i)$.

Describing $ p$ in terms of a set of numbers $p_i$ giving the  multiplicity of parts $i$ in the partition,
\bea 
\Sym (p) = \prod_i i^{ p_i} p_i!
\eea
 The form of this sum of squares  is explained by the Wedderburn-Artin decomposition of $ \cK ( n )$:  
an explicit basis,  which we refer to as the Fourier basis and which exhibits the decomposition, can be constructed using Clebsch-Gordan coefficients and matrix elements of permutation groups \cite{PCAMultMat,PCAKron}. 
This basis takes the form 
\be\label{qbasis}
Q^{R_1,R_2,R_3}_{\tau_1,\tau_2} = 
\kappa_{R_1,R_2}
\sum_{\s_1, \s_2 \in S_n}
\sum_{i_1,i_2,i_3, j_1,j_2}
C^{R_1,R_2; R_3 , \tau_1  }_{ i_1 , i_2 ; i_3 } C^{R_1,R_2; R_3, \tau_2  }_{ j_1 , j_2 ; i_3 } 
 D^{ R_1 }_{ i_1 j_1} ( \sigma_1  ) D^{R_2}_{ i_2 j_2 } ( \sigma_2 )  \,  \sigma_1 \otimes \sigma_2  
\ee
 $D^R_{ ij} ( \sigma ) $ are the matrix elements of 
the linear operator $D^R(\s)$ in an orthonormal basis for the irrep $R$. The indices 
$ \tau_1 , \tau_2 $ run over an orthonormal basis for the  multiplicity space of $R_3$ appearing in the tensor 
decomposition    of $ R_1 \otimes R_2$. This multiplicity is equal to the Kronecker coefficient $C ( R_1 , R_2 , R_3 )$ which is also the multiplicity  of the trivial representation in the tensor product decomposition of 
 $R_1 \otimes R_2 \otimes R_3$.  $\kappa_{R_1,R_2} = \frac{d(R_1)d(R_2)}{(n!)^2}$ is a normalization factor, where 
$d(R_i)$ is the dimension of the irrep $R_i$.  $C^{R_1,R_2; R_3 , \tau_1  }_{ i_1 , i_2 ;i_3 } $ are Clebsch-Gordan coefficients of the representations of $S_n$ (see the appendices 
of \cite{PCAKron} for the properties needed to prove that this expression gives a Wedderburn-Artin basis for $ \cK(n)$). 

These elements $ Q^{R_1,R_2,R_3}_{\tau_1,\tau_2} \in \mC ( S_n) \otimes \mC ( S_n)$ are invariant under diagonal conjugation 
\be\label{qinv}
(\gamma\otimes \gamma)\cdot Q^{R_1,R_2,R_3}_{\tau_1,\tau_2} \cdot
(\gamma^{-1}\otimes \gamma^{-1}) = Q^{R_1,R_2,R_3}_{\tau_1,\tau_2}  \, , 
\ee
and therefore belong to $\cK(n)$. 
It was verified \cite{PCAKron} that they define  the Wedderburn-Artin matrix bases of $\cK(n)$: 
\be
Q^{R_1,R_2,R_3}_{\tau_1,\tau_2}Q^{R_1',R_2',R_3'}_{\tau_2',\tau_3} 
=  \delta_{R_1R_1'} \delta_{R_2R_2'}  \delta_{R_3R_3'}\delta_{\tau_2 \tau_2'}  Q^{R_1,R_2,R_3}_{\tau_1,\tau_3}  \, . 
\label{qqmatrix0}
\ee
The normalization $\kappa_{R_1,R_2} $ is chosen to ensure that the RHS has the standard form for 
multiplication of elementary matrices, for each block labelled by triples $ ( R_1 , R_2 , R_3 )$ with non-vanishing
Kronecker coefficient $ C ( R_1 , R_2 , R_3)$. 
Noting that $C(R_1,R_2,R_3)$ is at most 1 for $n\leq 4$, 
then the matrices $Q^{R_1,R_2,R_3}_{\tau_1,\tau_2}$ are $1\times 1$
hence are commuting for $n\le 4 $. 
The set  $\{Q_{\tau_1,\tau_2}^{R_1,R_2,R_3}\} \in \cK (n)$ are  orthogonal with respect to the bilinear pairing  $\bdel_2$  
 \be
\bdel_2 (Q_{\tau_1,\tau'_1}^{R_1,R_2,R_3}; Q_{\tau_2,\tau'_2}^{R_1',R_2',R_3'})
 = \kappa_{R_1,R_2}  d(R_3)\,  \delta_{R_1R_1'} \delta_{R_2R_2'}
\delta_{R_3R_3'}\delta_{\tau_1\tau_2} \delta_{\tau'_1\tau'_2}     \, . 
\label{orhoqq0}
\ee
The cardinality of this set of orthogonal elements in $ \cK(n)$ is 
\bea 
\sum_{R_1,R_2,R_3} \sum_{\tau_1,\tau_2} 1 = 
\sum_{R_1,R_2,R_3} C(R_1,R_2,R_3)^2 = |\Rib(n)|
\eea
which allows us to confirm
that these elements  $\{Q_{\tau_1,\tau_2}^{R_1,R_2,R_3}\}$ form  an orthogonal basis of $\cK(n)$. 

The sets  $\{Q^{R_1,R_2,R_3}_{\tau_1, \tau_2}\}$ and
$\{E_r\}$ define orthogonal bases of $\cK(n)$.  We refer to the set 
$\{ E_r \}$  as the geometric or ribbon graph basis  and  to the set $\{ Q^{R_1,R_2,R_3}_{\tau_1, \tau_2}\}$
as the representation theoretic or Fourier basis of $\cK(n)$.  The change of basis from the Fourier basis to the geometric basis is made explicit in the   appendix \ref{app:geombasis}. 
 The existence of these  two bases and their interplay is an important resource exploited  in this paper. 
 
\subsection{Inner product and Hilbert space}
\label{sec:InnProdHilb} 

\subsubsection{$\cK ( n)$ as a Hilbert space}

\begin{proposition}
\label{PropcKnHilb} 
 The algebra $ \cK ( n)$ is a Hilbert space with the ribbon graph vectors 
$E_r$ forming an orthogonal basis; the vectors $\sqrt{ | \Orb(r) | } E_r \equiv e_r  $ form an orthonormal basis. 
\end{proposition} 
\begin{proof}
Define the inner product $g$ on $  \mC ( S_n)  \otimes \mC (  S_n)   $, using the basis 
of permutation pairs and extend it by linearity. For two pairs $ \alpha = ( \alpha_1  , \alpha_2 ) , \beta = ( \beta_1  ,  \beta_2 )  $ in $ S_n \times S_n $  we define 
\bea 
g ( \alpha , \beta ) = g ( \alpha_1 \otimes \alpha_2 , \beta_1 \otimes \beta_2 ) 
= \delta ( \alpha_1^{-1 }  \beta_1) \delta ( \alpha_2^{-1 }  \beta_2 ) 
\eea 
where $\delta$ is the delta function on $S_n$  ($\delta(\s)= 1$ if and only if $\s= \id$, otherwise
$\delta(\s)= 0$). This  extends by linearity  to a sesquilinear form on $ \mC ( S_n ) \otimes \mC ( S_n )$ as 
\bea\label{def:innerprod}
g ( \sum_i a_i   \alpha_{1i}  \otimes  \alpha_{2i} ,    \sum_j  b_j  \beta_{1j} \otimes  \beta_{2j}) 
= \sum_{i,j} \bar  a_i   b_j  \; 
\delta (\alpha_{1i} ^{-1 }\beta_{1j}  ) \delta (   \alpha_{2i} ^{-1} \beta_{2j} ) 
\eea
where $a_i,b_i \in \mC$ and 
where the bar means complex conjugation. 

We can  show that  $g$ satisfies conjugation property
$g (    \alpha ,   \beta ) = \overline{g ( \beta, \alpha )}$ and is positive definite. 
It therefore gives an inner product
on  $ \mC ( S_n ) \otimes \mC ( S_n )$.

We compute the inner product of two ribbon graph basis vectors 
\bea 
g ( E_r , E_s ) &=&
 { 1 \over | \Orb (r) |}  { 1 \over | \Orb (s) |} 
  \sum_{ a \in \Orb (r ) }
   \sum_{ b \in \Orb (s )  }
   g (  \s_1^{(r)} (a) \otimes \s_2^{(r)} (a)     ,  \s_1 ^{(s)} (b) \otimes \s_2^{(s)} (b)  )
     \crcr
&   =  &  { 1 \over | \Orb (r) |}  { 1 \over | \Orb (s) |} 
  \sum_{ a \in \Orb (r ) }
   \sum_{ b \in \Orb (s )  } \delta_{a,b} \delta_{ r, s } 
\eea
where the only way that 
$\delta((\s_1^{(r)} (a))^{-1} \s_1 ^{(s)} (b)) \, \delta((\s_2^{(r)} (a))^{-1} \s_2 ^{(s)} (b))=1$
for two orbit elements $a$ and $b$ is when $a=b$. 
As a couple $( \s^{(r)} _1 (a) , \s^{(r)} _2 (a)  )$ can only appear in a unique orbit, 
we therefore have 
\bea    \label{gErEs}
g ( E_r , E_s ) = { 1 \over | \Orb (s) | } \delta_{ r , s } \, . 
\eea
Then  the set of $\{E_r\}$ for $  r \in \{ 1 , \cdots , |\Rib(n)| \}$ defines an orthogonal basis of $\cK(n)$  which becomes a Hilbert space of states
 with inner product $g$.  We can define orthonormal bases for $\cK(n)$ as states of the form: 
\bea 
e_r = \sqrt { |\Orb (r) | } E_r \,. 
\eea
This  completes the proof of the proposition.  

\end{proof}

\subsubsection{Involution on $ \cK(n)$ from permutation inversion} 
\label{sec:Involution} 

We also define the linear conjugation operator $S: \mC(S_n) \to \mC(S_n)$ that maps a linear combination $A = \sum_{i} c_i \s_i \in \mC(S_n)$ to 
\bea\label{inverse}
S(A) := \sum_{i} c_i \s_i^{-1} 
\eea
Extend this operation to $\mC ( S_n )  \otimes \mC ( S_n) $ by inverting the permutation in each tensor factor: $S(\s_{1}  \otimes \s_{2} ) = \s_{1}^{-1}  \otimes \s_{2} ^{-1} $ and using linearity 
\bea
S(\sum_{i  } a_i  \; \s_{1i} \otimes \s_{2i} ) = \sum_{i  } a_i  \s_{1i}^{-1}  \otimes \s_{2i} ^{-1} 
\eea
$S$ is an involution:  $S^2 = \id$, and obeys $ S ( A B ) = S ( B ) S ( A )$.

The conjugation $S$ gives a well defined involution from the set of equivalence classes forming $ \cK ( n ) $ to itself. To see this note that if $( \sigma_1 , \sigma_2) \in \Orb  ( r )$ maps under inversion to a pair  
$ ( \sigma_1^{-1} , \sigma_2^{-1} )  \in \Orb (s) $ (where $s$ may or may not be equal to $r$),
 then for any $\mu \in S_n$ 
\bea 
S ( \mu \sigma_1 \mu^{-1} , \mu \sigma_2 \mu^{-1} ) = ( \mu \sigma_1^{-1}  \mu^{-1} , \mu \sigma_2^{-1} \mu^{-1} ) \in \Orb (s)  
\eea
If for a given $r$, $S $ maps the pairs $ \Orb (r) $ back to $ \Orb (r)$ we have 
\bea 
S ( E_r ) = E_r 
\eea
Such ribbon graphs will be called self-conjugate. If $ \Orb (r)$  is mapped to $ \Orb(s)$ with $ s \ne r$, then 
\bea 
&& S ( E_r ) = E_s \cr 
&& S  ( E_s ) = E_r 
\eea
and we call $(E_r,E_s)$ a conjugate pair. In section  \ref{app:conjugate} we compute the action of  $S$ on the Fourier basis.  The  interplay between these two actions is used to show that the total number of self-conjugate ribbon  graphs  with $n$ edges 
is equal to the sum of Kronecker coefficients $ C ( R_1 , R_2 , R_3 )$ for $ R_1 , R_2 , R_3 \vdash n $. The operator $S$ is also useful in proving the hermiticity of the reconnection operators $ T_k^{(i)}$  (Proposition \ref{propTkHerm})  
which are used to construct  Hamiltonians on $ \cK(n)$ in section \ref{sec:SquareIntMat}.

\subsection{The integrality structure of the product on $\cK(n)$} 
\label{integral}

The product in the  algebra $\cK(n)$ gives an expansion  of the product $E_r E_s$ of two geometric basis vectors
\bea 
E_r E_s = \sum_{ t =1 }^{ |\Rib(n) | }  C_{rs}^{t}E_t
\eea 
We will express the structure constants 
\bea
C_{ rs}^t =  {\rm Coeff}  ( E_t  , E_r E_s )
\eea
in terms of non-negative   integers. Recall  from section \ref{sec:geombasis}  the two expressions for $E_r$ 
\bea\label{keyEqs}  
E_r &=& { 1 \over n! }  \sum_{ \mu \in S_n } \mu \sigma_1^{(r)}  \mu^{-1}  \otimes \mu \sigma_2^{(r)}  \mu^{-1}  \cr 
 & =& { 1 \over | \Orb ( r ) | }   \sum_{ a \in \Orb ( r ) } \sigma^{(r)}_1 (a ) \otimes \sigma^{(r)}_2 (a) 
\eea 
For notational convenience, we define 
\bea  \sigma^{ (r) } = \sigma_1^{(r)} \otimes \sigma_2^{ (r)}  \;, \qquad \quad 
 \mu \sigma^{(r)} \mu^{-1} = \mu \sigma_1^{(r)} \mu^{-1} \otimes \mu \sigma^{(r)}_2 \mu^{-1}  \, . 
\eea
We express \eqref{Err} in a simpler form as
\bea 
E_r   = { 1 \over n! } \sum_{ \mu \in S_n} \mu \sigma^{(r)} \mu^{-1} 
                  = {  1 \over | \Orb(r) | } \sum_{ a \in \Orb (r) } \sigma^{(r)} (a) \,. 
\eea
Consider the product of two elements of $ \cK ( n ) $ associated with orbits $r$ and $s$: 
\bea\label{takeit}  
 E_r E_s 
& = & { 1 \over  | \Orb(r) | } { 1 \over | \Orb (s) | } 
\sum_{ a \in \Orb (r) } \sum_{ b \in \Orb (s) } \sigma^{(r)} (a) \sigma^{(s)} ( b )  \cr 
& = & { 1 \over (n!)^2 } \sum_{ \mu_1 , \mu_2 \in S_n } \mu_1 \sigma^{(r)} \mu_1^{-1} \mu_2 \sigma^{(s)} \mu_2^{-1}  \, . 
\eea
We can write $ \nu = \mu_1^{ -1} \mu_2 $ and then solve for $ \mu_1 = \mu_2 \nu^{-1} $ to write 
\bea 
 E_r E_s   &=&  
{ 1 \over n! } { 1 \over | \Orb (r) | } \sum_{ \mu_2 \in S_n} \sum_{ a \in \Orb ( r ) } 
 \mu_2 \sigma^{(r)} (a) \sigma^{(s)}  \mu_2^{-1}   \cr 
 & =&  { 1 \over | \Orb (r) | } \sum_{ a \in \Orb (r) } 
{  1 \over | \Orb (\sigma^{(r)} (a) \sigma^{(s)} ) |  }\sum_{ b \in \Orb (\sigma^{(r)} (a) \sigma^{(s)} ) } \sigma ( b ) \cr 
 & =&  \sum_{ t } { 1 \over | \Orb (r) | } \sum_{ b \in \Orb (t) } { \sigma^{ (t)} { (b) }\over | \Orb (t) | } 
  \sum_{ a  \in \Orb (r) } \delta (  \Orb ( t )  , \Orb ( \sigma^{(r)} (a)  \sigma^{ (s)} ) ) \cr 
  & =&   \sum_{ t } { 1 \over | \Orb (r) | }E_t  \sum_{ a  \in \Orb (r) } \delta (  \Orb ( t )  , \Orb ( \sigma^{(r)} (a)  \sigma^{ (s)} ) )  \,, 
\eea
where $\delta(\Orb(s), \Orb(t))$ is the Kronecker symbol $\delta_{st}$
for the labels $s$ and $t$. We have thus expressed the product of $E_r E_s $ in terms of the non-negative 
 integer 
\bea
&& \sum_{ a  \in \Orb (r) } \delta (  \Orb ( t )  , \Orb ( \sigma^{(r)} (a)  \sigma^{ (s)} ) )  \cr 
&& = \hbox{ Number of times the multiplication of elements from orbit $r$ } \cr 
&&  \hbox{ with a fixed element in orbit $s$ to the right 
produces an element in orbit $t$ }  \cr 
&& 
\eea
If we solve for $\mu_2$ instead as $ \mu_2 = \mu_1 \nu $, then we get 
\bea 
  E_r E_s & =&  
\sum_{ t } { 1 \over | \Orb (s) | }   \sum_{ b \in \Orb (t)} { \sigma^{ (t)}  { (b)} \over | \Orb (t) | } 
  \sum_{ a \in \Orb (s)  } \delta (  \Orb ( t )  , \Orb ( \sigma^{(r)}   \sigma^{ (s)} (a)  ) )   \cr 
  & = & \sum_{ t } { 1 \over | \Orb (s) | } E_t \sum_{ a \in \Orb (s)  } \delta (  \Orb ( t )  , \Orb ( \sigma^{(r)}   \sigma^{ (s)} (a)  ))
\eea
Here we have expressed the same product in terms of the non-negative integers 
\bea 
&& \sum_{ a \in \Orb (s)  } \delta (  \Orb ( t )  , \Orb ( \sigma^{(r)}   \sigma^{ (s)} (a)  )) \cr 
&& = \hbox{ Number of times the multiplication of elements from orbit $s$ } \cr 
&&  \hbox{ with a fixed element in orbit $r$ from the left  
produces an element in orbit $t$ } \cr 
&& 
\eea
Equivalently, we can express these by saying that the coefficient ${\rm Coeff} ( E_t , E_r E_s  )$  of $E_t$ in  the expansion of $E_rE_s$ is given by 
\bea\label{ErEsN}  
 {\rm Coeff}  ( E_t  , E_r E_s ) &=& { 1 \over | \Orb (s) | }  \sum_{ a \in \Orb (s)  } \delta (  \Orb ( t )  , \Orb ( \sigma^{(r)}   \sigma^{ (s)} (a)  ) ) \cr 
& =&  { 1 \over | \Orb (r) | }\sum_{ a  \in \Orb (r) } \delta (  \Orb ( t )  , \Orb ( \sigma^{(r)} (a)  \sigma^{ (s)} ) ) 
\eea 
If we just keep $ \mu_1 , \mu_2 $  in \eqref{takeit} we can write  (treating  $r$ and $s$ more symmetrically) 
\bea 
{\rm Coeff}  ( E_t ,  E_r E_s  ) = 
{ 1 \over | \Orb (s) | ~ | \Orb (r) | } \sum_{ a   \in \Orb (r) } \sum_{ b \in \Orb (s) } 
 \delta (  \Orb ( t )  , \Orb ( \sigma^{(r)} (a)  \sigma^{ (s)} (b)  ) )
\eea
Recalling that \eqref{gErEs} holds, we therefore have 
\bea 
 {\rm Coeff}  ( E_t , E_r E_s  ) =  { g( E_r E_s  , E_t  ) | \Orb (t) | }  
\eea
Hence 
\bea 
&& \bdel_2 ( E_r E_s S(E_t) ) =  g ( E_r E_s  , E_t  )  \cr 
&&  = { 1 \over | \Orb (s) | ~ | \Orb (r) | ~ | \Orb (t) |  } \sum_{ a   \in \Orb (r) } \sum_{ b \in \Orb (s) } 
 \delta (  \Orb ( t )  , \Orb ( \sigma^{(r)} (a)  \sigma^{ (s)} (b)  ) )
\eea
These formulae show that the  structure constants of the algebra $ \cK ( n )$ are expressed in terms of 
non-negative integers obtained from the combinatorial  multiplications of elements in the orbits, which provide 
the geometrical ribbon graph basis vectors $E_r$ of $ \cK (n)$.  In general, the product is not commutative $ E_r E_s \ne E_s E_r $. In the next section, we will exploit this integral structure, for the particular cases where $E_r$ are chosen to be central elements in $ \cK (n)$, associated with  permutations having a  cycle of length $k$ and remaining cycles of length $1$. These central elements will be used to construct Hamiltonians and the eigenproblems of these Hamiltonians will become questions about non-negative integer matrices.

\subsection{The centre of  $ \cK ( n ) $ and reconnection operators  $ T_k^{(i)} $ }
\label{sec:CentreKnReconn} 

In this section, we will review some properties of the centre of $ \cK(n)$ and introduce central  
elements  $T_k^{(i)} \in \cK(n)$ labelled by $ k \in \{ 2, 3, \cdots  , n \} $ and $i \in \{ 1 , 2, 3 \} $. 
These central elements act on $ \cK(n)$ by multiplication. Since $ \cK(n)$ (when equipped with the inner product specified)  is also the Hilbert  space of our quantum mechanical systems, the elements  $ T_k^{(i)} $ also define linear operators when they act on $ \cK(n)$ by multiplication. We are taking advantage of a state-operator correspondence which is possible when the Hilbert space of a quantum mechanical space is also an algebra. We prove that these linear operators are Hermitian  with respect to the inner product defined in section 
\ref{sec:InnProdHilb}.   We refer to the $T_k^{(i)}$ as reconnection operators, since they act, as we will see shortly,  on the permutations defining elements of $\cK(n)$ by multiplication of permutations.  In the diagrammatic description of  tensor model observables associated with $\cK(n)$ \cite{Sanjo} this operation involves reconnecting the index lines of the tensor variables $\Phi^{ \otimes n } $ with those of $\bar \Phi^{\otimes n }$.  
In terms of ribbon graphs, the action of $T^{(i)}_k$ amounts to splitting and joining of vertices: such operators have also been discussed in \cite{Itoyama:2017wjb,Itoyama:2019oab}.

Let us first recall some properties of the centre $ \cZ ( \mC ( S_n ) ) $ of $ \mC ( S_n)$. 
The centre is defined as the sub-algebra of elements which commute with all $ \mC ( S_n)$. $ \cZ ( \mC ( S_n ) ) $ is a commutative algebra of dimension $p(n)$,  the number of partitions of $n$.  The conjugacy classes of $S_n$ are specified by cycle structures of permutations which define partitions of $n$. The sum of elements in a conjugacy class is a central element in the group algebra. A linear basis for the centre is given by these class sums. 
For any integer $k$, such that $2 \le k \le n $, let $ \mathcal{C}_k $ to be the conjugacy class of 
 permutations $\s \in S_n$ made of a single cycle of length $k$ and remaining cycles of length $1$.  As an example, for $ n=3$, $ k =2$, the conjugacy class $ \cC_2$ is the set of permutations $ \{ ( 1,2) ( 3) , (2,3) ( 1) , (1,3) (2) \} $.  Define $T_k$ as the sum 
\be
\label{tk}
T_k = \sum_{ \s \in \cC_{ k } }  \s 
\ee    
$|T_k|$ will refer to as the number of terms
in that sum, equivalently the number of terms in $ \cC_k$, which is ${  n! \over k ( n-k)! } $. 
 For any $n$, the set $ \{ T_2 , T_3 , \cdots , T_n \} $ generates the centre \cite{KR1911}, i.e. 
   linear combinations of products of these $T_k$  span $\cZ ( \mC (S_n))$. 
In fact   there is no need to consider the entire set to generate   $ \cZ ( \mC ( S_n ) ) $. Indeed, there exists $k_* (n) \le n $, such that the subset $\{ T_2 , \cdots , T_{ k_*(n)} \} $  spans the center \cite{KR1911}. 
This is related to the fact that the ordered list of 
of normalized characters $ ( \widehat \chi_{ R } ( T_2  ) , \widehat \chi_{ R }( T_3 ) , \cdots , \widehat \chi_{ R } (T_{k_*(n)}  )  ) $ uniquely identifies the Young diagram $R$. These normalized characters are defined as 
\bea
\widehat \chi_{ R } ( T_k  ) =  { \chi_R ( T_k ) \over d(R) }
\eea 
where $d(R)$ is the dimension of the irrep $R$.  The sequence $k_*(n)$   was explicitly computed  \cite{KR1911}, with the help of  character formulae in \cite{Lasalle,CGS}, to be 
\bea \label{kstar}
k_* ( n )   &  =  & 2 \;  \hbox{ for }   n \in \{ 2,3, 4,5,7 \} \cr 
k_* ( n )   & = &  3  \; \hbox{ for } n \in \{ 6 , 8 , 9 \cdots , 14 \} \cr 
k_* (n) & = &  4\; \hbox{ for } n \in \{ 15 , 16 , \cdots , 23, 25 ,26  \} \cr 
k_{ *} ( n ) &  = &  5 \;\hbox{ for } n \in \{ 24 ,  27 ,   \cdots , 41 \} \cr 
k_* (n)  & = &  6\; \hbox{ for } n \in \{ 42 , \cdots , 78,79  , 81 \} 
\eea

At any $n \ge 2$,  we will define elements in   $ \mC ( S_n) \otimes \mC ( S_n)$  
\bea\label{tki}
T^{(1)}_k &=& T_k \otimes 1   = \sum_{ \s \in \cC_k } \sigma \otimes 1 \,,  \crcr
T^{(2)}_k &=& 1 \otimes T_k    = \sum_{ \s \in \cC_k } 1 \otimes \sigma \,, \crcr
T^{(3)}_k &=& \sum_{ \sigma \in \cC_k }  \sigma \otimes  \sigma  \;. 
\eea
These commute with permutations $ \gamma \otimes \gamma$ and are thus in the sub-algebra
 $ \cK(n) \subset \mC ( S_n) \otimes \mC ( S_n) $.  Using the correspondence between permutation pairs and ribbon  graphs described in section \ref{sec:Counting}, it is straightforward  to describe the ribbon graphs corresponding to $T_k^{(i)} $.  In the case $n=3$, $T_2^{(1)} , T_2^{(2)} , T_2^{(3)} $ correspond to the ribbon graphs labelled $4,2,5 $ respectively in  Figure \ref{fig:ribb}. The elements  $ T_3^{(1)} , T_3^{(2)} , T_3^{(3)} $ correspond to the graphs  labelled $8,3,10$ respectively.  For general $n$, $T_k^{(1)} $ corresponds to a   ribbon graph  with genus zero, having one black vertex of valency $k$, $n-k$ black vertices of valency one, and $n$ white vertices of valency one.  For $T_k^{(2)}$, we have a ribbon graph with one white vertex of valency $k$, $n-k$ vertices of valency one, and $n$ black vertices of valency one. For $T_k^{(3)}$ we have a graph with genus $ \lfloor {k-1 \over 2 } \rfloor $ with a black $k$-valent vertex connected to a white $k$-valent vertex, along with $n-k$ one-valent black and white vertices. These graphs are disconnected for $ n > k$.

The $T_k^{(i)}$'s  act as linear operators on $\mC(S_n)\otimes \mC ( S_n )$ by left 
multiplication. Right multiplication gives the same operators because these are central operators in 
$\mC ( S_n ) \otimes \mC ( S_n)$. They are also central in $ \cK (n)$, since this is a sub-algebra of  $\mC ( S_n ) \otimes \mC ( S_n)$. 
Let $(\cM^{ (i)}_k )_r^s $ be the matrix elements of $T_k^{(i)}$
\bea\label{Ter}
T_k^{(i)} E_s = \sum_{ s } (\cM^{ (i)}_k )_s^t  E_t 
\eea
in the geometric basis. 
\begin{proposition} 
\label{PropIntTk}
The matrix elements $ (\cM^{ (i)}_k )_r^s $ are non-negative integers. 
\end{proposition}
\begin{proof}
The $T_k^{(i)} $ are proportional to instances of the geometric basis vectors
 $  E_r$ \eqref{keyEqs}  obtained by summing over  diagonal conjugations of
permutations of the form $ \sigma \otimes 1 , 1 \otimes \sigma  , \sigma \otimes \sigma $, where $ \sigma $ is
a cyclic permutation of a subset of $k$ numbers from $\{ 1,2, \cdots , n \}$. Using the correspondence (section \ref{sec:Counting})  between ribbon graphs and permutations, they each correspond to a ribbon graph. Each $T_k^{(i)} $ corresponds to a ribbon graph, with some  label  $r$ which we will call $r(k,i)$. The proportionality constant is given as
\bea
T_k^{(i)} = | \Orb ( r (k,i) ) | ~E_{ r (k,i)} 
\eea
since the $T_k^{(i)}$ are equal to a sum of elements in an orbit generated by the diagonal conjugations, while $E_r$ are defined to be such sums normalized by the orbit size. The formula \eqref{ErEsN} for the algebra product in the geometric basis then implies that 
\bea\label{eqnMst}  
&& (\cM^{ (i)}_k )_s^t = \hbox{ Number of times the multiplication of elements in the sum $T_{k}^{(i)} $ }\cr 
&&  \hbox{ with a fixed element in orbit $s$ to the right  produces an element in orbit $t$.} \cr 
&& 
\eea 
\end{proof}

\begin{proposition}\label{propTkHerm}
$T_k^{(i)} $ are Hermitian operators on $\cK(n)$ in the inner product defined by \eqref{def:innerprod} :  
\bea\label{hermsym}  
 g ( E_s , T_k^{ (i)} E_r )= g ( T_k^{ (i)} E_s , E_r )  \, . 
\eea
\end{proposition}  
\begin{proof}
Using \eqref{gErEs} and \eqref{Ter} we evaluate 
\bea\label{gesr}  
g ( E_s , T_k^{ (i)} E_r ) = ( \cM^{(i)}_{k} )_{ r}^s  { 1 \over | \Orb (s) | } 
\eea
By renaming $ r , s $, we have 
\bea 
g  ( E_r  , T_k^{ (i)} E_s ) = ( \cM^{(i)}_{k} )_{ s}^r   { 1 \over | \Orb (r) | } 
   \label{geres}
\eea
Using definition of $g$ in \eqref{def:innerprod} and of  $\bdel_2$ in \eqref{bdelta}, 
the following is true
\bea 
g(\alpha, \beta) = \bdel_2( \overline{\alpha},\beta) \,. 
\eea
Using this relation between the inner product and the delta function, 
we have 
\bea\label{startSeq} 
 g ( E_r  , T_k^{ (i)} E_s )    =  \bdel_2  ( \overline{E_r}, ( T_k^{(i)} E_s) ) 
\eea
From the definition of $S$ \eqref{inverse}, note that $\bdel_2(S(a),id)= \bdel_2(a,id)$, 
and $$
\bdel_2(a,b)= \bdel_2(aS(b),id) = \bdel_2(S(b)a,id).
$$
Moreover, $S(T_k)= T_k$ and $ S ( AB ) =  S ( B ) S ( A )$  imply 
\begin{equation} 
\bdel_2  ( \overline{E_r}, T_k^{(i)} E_s ) = \bdel_2  ( S(T_k^{(i)} E_s) \overline{E_r},id  )  
= \bdel_2(  S(E_s) T_k^{(i)} \overline{E_r}, id)  =\bdel_2(  T_k^{(i)} \overline{E_r},  E_s)  
  = g( \overline{T_k^{(i)} \overline{E_r}},  E_s) 
\end{equation} 
 Next observe that $\alpha = T_k^{(1)}E_r$ and $\beta = E_s$
have all real coefficients  in the geometric ribbon graph basis, 
\bea 
 g( \overline{T_k^{(i)} \overline{E_r}},  E_s) = g(T_k^{(i)}E_r, E_s)
\eea  
This sequence of steps  starting from \eqref{startSeq} shows that 
\bea 
 g ( E_r  , T_k^{ (i)} E_s )  = g(T_k^{(i)}E_r, E_s)
\eea
This proves the proposition. 
\end{proof} 

\noindent
{\bf Remark} \\
The matrix elements of $T_k^{(i)}$ in the non-orthogonal basis 
$E_r$ are not symmetric. Indeed from \eqref{gesr}, \eqref{geres}, and \eqref{hermsym} we have 
\bea \label{morb}
 ( \cM_{ k}^{ (i)} )_r^s  = | \Orb (s) |   ( \cM_k^{(i)} )_s^r  { 1 \over | \Orb (r) | } 
\eea
If we consider instead the matrix elements of $ T_k^{(i)} $ on the orthonormal basis vectors 
 $e_r = \sqrt{(| \Orb(r) | }E_r$, 
\bea
T_k^{(i)} e_r =   \sqrt{ | \Orb(r) | } ( \cM_{ k}^{(i)} )_r^s  { 1 \over \sqrt { | \Orb(s) | } }  e_s 
\eea
These matrix elements are symmetric under exchange of $r$ and $s$.

\noindent 
{\bf Remark} \\
The operators $T_k^{(i)} $, as $i$ ranges over $\{ 1,2, 3 \}$ and $k$ ranges over some subset of $\{ 2,3, \cdots , n  \}$ form a set of commuting Hermitian operators on $\cK(n)$. The commutativity follows from the fact that they are central elements of $ \cK(n)$, the hermiticity from Proposition \ref{propTkHerm}.   Considering such sets of operators as Hamiltonians defining a time evolution of states in $\cK ( n ) $ we have time-dependent ribbon graph states of 
the form 
\bea 
E_r (t) = e^{ - i t T_{ k}^{(i)} } E_r 
\eea
In section \ref{primes} we will construct Hamiltonians which are  particular linear combinations of these operators, and  have the property that their eigenvalue degeneracies are Kronecker coefficients. In order to build up to this, we will now consider the action of the $T_k^{(i)} $ operators on the Fourier basis of $ \cK(n)$.

\section{ Integer matrices and Kronecker coefficients} 
\label{sec:IntegerMatricesKron}

In  this section we will consider the action of the reconnection operators $T_{k}^{(i)} $ introduced in section 
\ref{sec:CentreKnReconn} on the elements   $ Q^{R_1, R_2, R_3}_{ \tau_1 , \tau_2 } $ of the  Fourier basis set for $\cK(n)$ described in section \ref{Fourier}. The subspace  of $\cK(n)$ spanned by the Fourier basis 
 elements for a fixed ordered Young diagram triple $(R_1 , R_2  , R_3 )$ has dimension equal to $ C ( R_1 , R_2 , R_3 )^2$, the square of the Kronecker coefficient for the triple. We will refer to such a subspace as the Fourier subspace of $ \cK(n)$ associated with the triple $ ( R_1 ,  R_2 , R_3 )$.  We will show (section \ref{sec:reconn})   that the Fourier subspace for a triple form an  eigenspace  of the reconnection operators. The eigenvalues are normalized characters of the symmetric group, with rational values known from symmetric group representation theory. 
 
 We work with a set of reconnection operators chosen such that their eigenvalues uniquely specify the Young diagram triple.  This means that, although the Fourier basis was initially defined using matrix elements and Clebsch-Gordan coefficients for symmetric groups (equation \eqref{qbasis}) , we can use the reconnection operators and the eigenvalues as input, and compute  the Fourier subspace  for a fixed triple of Young diagrams directly as the null space of an integer matrix built from the reconnection operators and the eigenvalues specifying a Young diagram triple. This gives a  computational approach for the Fourier subspace associated with a Young diagram triple without using the detailed representation theory input of matrix elements and Clebsch-Gordan coefficients: we are only using the coarser input of character formulae (having known combinatorial algorithms for their computation) alongside the combinatorially defined reconnection operators. The first algorithm based on this approach  (section \ref{sec:reconn}) amounts to finding the null vectors of a rectangular  integer matrix. In section \ref{primes} we give the construction of quantum mechanical Hamiltonians which are integer  linear combinations of the reconnection operators  $T_k^{(i)}$ and have eigenvalues that  uniquely specify a Young diagram triple. This leads to an algorithm which obtains the Fourier subspace for a specified triple as the null space of a square integer matrix. 

\subsection{Fourier subspace of a Young diagram  triple as eigenspace of reconnection operators }
\label{sec:FSEig}

 \begin{proposition}
   \label{propTaQo}
 For all $k \in \{ 2, 3, \cdots  n \}$,  $\{ R_i \vdash n : i \in \{  1,2,3\}  \} $, $\tau_1, \tau_2  \in  [\![1, C(R_1,R_2,R_3) ]\!]$,  the Fourier basis elements 
  $ Q^{R_1, R_2, R_3}_{\tau_1 , \tau_2}$ are   eigenvectors of $T_k^{(i)} $: 
\bea
&&  T_k^{(1)} Q^{R_1, R_2, R_3}_{ \tau_1 , \tau_2 }  = (\sum_{\s \in \cC_k} \s \otimes 1 ) Q^{R_1, R_2, R_3}_{ \tau_1 , \tau_2 }   = { \chi_{R_1} ( T_k ) \over d(R_1) }   Q^{R_1, R_2, R_3}_{ \tau_1 , \tau_2 } \,,
\label{t1Qo} \\
&&  T_k^{(2)}   Q^{R_1, R_2, R_3}_{ \tau_1 , \tau_2 } 
=  (\sum_{\s \in \cC_k } 1 \otimes \s ) Q^{R_1, R_2, R_3}_{ \tau_1 , \tau_2 }   = 
 { \chi_{R_2} ( T_k ) \over d(R_2) } Q^{R_1, R_2, R_3}_{ \tau_1 , \tau_2 } \,,
\label{t2Qo}  \\
&&  T_k^{(3)} Q^{R_1, R_2, R_3}_{ \tau_1 , \tau_2 } =  (\sum_{\s \in \cC_k } \s \otimes \s ) Q^{R_1, R_2, R_3}_{ \tau_1 , \tau_2 } = { \chi_{R_3} ( T_k ) \over d(R_3) }    Q^{R_1, R_2, R_3}_{ \tau_1 , \tau_2 } \,. 
\label{t3Qo} 
\eea
 \end{proposition}
The proof is given in appendix \ref{app:geombasis}.  
Note that the eigenvalues do not depend on  the multiplicity
indices $\tau_1$ and $ \tau_2$, but only on the Young diagram labels $ ( R_1 , R_2 , R_3 )$. 
The proof of Proposition \ref{propTaQo} relies on representation theoretic arguments.

\begin{proposition} 
\label{PropTkFS} 
For any $\widetilde k_*  \in \{ k_*(n) , k_*(n) +1 , \cdots , n \} $ the list of eigenvalues of the 
reconnection operators 
$ \{  T^{(1)}_{2} , T^{(1)}_{ 3} , \cdots , T^{(1)}_{ \widetilde k_* } ;  T^{(2)}_{2} , T^{(2)}_{ 3} , \cdots , T^{(2)}_{ \widetilde k_* }  ; T^{(3)}_{2} , T^{(3)}_{ 3} , \cdots , T^{(3)}_{ \widetilde k_*}   \}$  
uniquely  determines the Young diagram triples $(R_1 , R_2 , R_3 )$.
\end{proposition} 

\begin{proof}
It was shown in \cite{KR1911} that the normalized characters $\{ {\chi_R (T_2) \over d(R)} , { \chi_{R} (T_3 ) \over d(R) } , \cdots , { \chi_R ( T_n ) \over d(R) }   \} $ form  ordered lists of numbers which distinguish Young diagrams $R$ with $n$ boxes.  For all $ n >2$ it was shown that the shorter list  $\{ {\chi_R (T_2) \over d(R)} , { \chi_{R} (T_3 ) \over d(R) } , \cdots , { \chi_R ( T_{n-1} ) \over d(R) }   \}$ distinguishes Young diagrams. This result follows from the fact that the central elements $ \{ T_2 , \cdots , T_{n-1} \} $ generate the centre of $\mC ( S_n)$ for $n >2$. It was found that there generically exist $k_*(n) <n-1$ such that the shorter lists 
 $\{ {\chi_R (T_2) \over d(R)} , { \chi_{R} (T_3 ) \over d(R) } , \cdots , { \chi_R ( T_{ k_*(n)} ) \over d(R) }   \}$ distinguish Young diagrams. The values of $k_*(n)$ computed for all $n$ up to $79$ are given in \eqref{kstar}.
 While the general $k_*(n)$ are not currently known, any $  n \ge \widetilde k_* \ge k_*(n) $ gives a longer list which distinguishes Young diagrams.  A triple of lists of length $\widetilde k_* $ distinguishes a triple of Young diagrams. 
 Using Proposition \ref{propTaQo}, these are the sets of eigenvalues of the operators $ \{  T^{(1)}_{2} , T^{(1)}_{ 3} , \cdots , T^{(1)}_{ \widetilde k_* } ;  T^{(2)}_{2} , T^{(2)}_{ 3} , \cdots , T^{(2)}_{ \widetilde k_* }  ; T^{(3)}_{2} , T^{(3)}_{ 3} ,$ $  \cdots ,  T^{(3)}_{ \widetilde k_*}   \}$.  
\end{proof} 

\begin{lemma}
\label{IntegralityOfNormChar}
The  sum of all permutations $ \sigma $ in the conjugacy  class $ C_{ p } $ in $ S_n$ for partition 
$p $ are central elements in $ \cZ ( \mC ( S_n)) $. The irreducible normalized 
 characters of these central elements are integers : 
\bea 
{ \chi^R ( T_{ p } ) \over d ( R )  }  \in \mZ 
\eea
\end{lemma}

The proof combines a known number theoretic fact about the normalized characters of  a finite group 
 being algebraic integers \cite{Simon}, along with the rationality of characters of irreducible representations of $ S_n$
 which follows from the Murnaghan-Nakayama Lemma. 
\begin{proof} 
The elements $T_p$ as $p$ runs over the classes form a basis for $ \cZ ( \mC ( S_n) ) $.  The structure constants 
of the multiplication are defined by 
\bea
T_p T_q = \sum_{ r } C_{ pq}^r T_r 
\eea
These structure constants $C_{ pq}^r $ are integers. The normalized characters ${ \chi^R ( T_p ) \over d( R )    }$    are eigenvalues of the matrix defined by $C_{pq}^r $ for fixed $p$. The eigenvalues of an integer matrix are algebraic integers (see e.g. Proposition III.4.3  \cite{Simon}). In the case of symmetric groups, we know that $ \chi^R ( \sigma )$ for 
$ \sigma \in C_{ p  } $ is an integer by using the Murnaghan-Nakayama Lemma. It follows that the normalized characters $ { \chi^R ( T_p) \over d( R ) }  $ are rational. A rational number which is also an algebraic integer is necessarily an integer. This means that the normalized characters are integers, for any conjugacy class $p$. 
\end{proof}

In particular for the partitions of the form $ [ k , 1^{ n-k} ] $ the normalized characters 
\bea 
{ \chi^R ( T_{ k  } ) \over d ( R ) }  \in \mZ 
\eea 
Using the formulae for the normalized characters in \cite{Lasalle,CGS}
 the normalized characters for $T_2$ (as well as $T_3,T_4,T_5,T_6$) are evidently integers 
for any Young diagram.  For higher higher $T_k$, the integrality is not evident from the formulae, but hold from the above argument.

\subsection{Fourier subspace of triple  as null-space of rectangular  integer matrices}
\label{sec:reconn}

It is useful to recall from equation \eqref{eqnMst} and Proposition \ref{PropIntTk}
reproduced here for convenience (with a slight change in index labels ) : 
\bea 
T_k^{(i)} E_r = \sum_{ s } (\cM^{ (i)}_k )_r^s  E_s
\eea 
with 
\bea 
&& (\cM^{ (i)}_k )_r^s = \hbox{ Number of times the multiplication of elements in the sum $T_{k}^{(i)} $ }\cr 
&&  \hbox{ with a fixed element in orbit $r$ to the right  produces an element in orbit $s$.} \cr 
&& 
\eea 
Using the definition of $T_k^{(i)}$, this means that 
\bea 
&& (\cM^{(1)}_k )^{ s }_{ r }   = \sum_{ \gamma \in \cC_k } 
 \delta ( \Orb ( s )  , \Orb ( \gamma \tau_1^{(r)}  \otimes \tau_2^{(r)}  )  ) \cr 
&& 
 (\cM^{(2)}_k )^{ s }_{ r }   = \sum_{ \gamma \in \cC_k } 
 \delta ( \Orb ( s )  , \Orb (  \tau_1^{(r)}  \otimes \gamma \tau_2^{(r)}  )  ) \cr
 &&
  (\cM^{(3)}_k )^{ s }_{ r }   = \sum_{ \gamma \in \cC_k } 
 \delta ( \Orb ( s )  , \Orb ( \gamma \tau_1^{(r)}  \otimes \gamma \tau_2^{(r)}  )  ) 
\eea

The integer matrices $\cM^{(i)}_k$ of $ T_k^{(i)} $ are  constructed  in Code2 of the 
 appendix \ref{app:gap}, see the function \verb|ArrayTi(n, kmax)|.   Note also that we have the following relations
\bea
\sum_{s}  (\cM^{(i)}_k )^{ s }_{ r }=  |T_k|  = { n! \over k (n-k)! }   \; , 
\eea
where $|T_k|$ is the number of terms in  the sum $T_k$, see the  discussion 
after \eqref{tk}. Thus each column of $ (\cM^{(i)}_k )^{ s }_{ r }$ is a list of non-negative integers adding up  to 
$|T_k|$. 

\subsubsection{Stacking $T_k^{(i)}$ matrices and common eigenspace}

Using Proposition \ref{PropTkFS}, the Fourier subspace for a given triple $ ( R_1 , R_2 , R_3 )$ 
is uniquely specified as  common eigenspace of the operators $T^{(i)}_k$, for 
$k \in \{ 2, \dots, \tilde k_*(n) \} $ and $i\in \{  1,2,3 \} $; with $ k_*(n)  \le \widetilde k_* \le n  $, 
with specified eigenvalues for these reconnection operators, which are known from symmetric group representation theory. These eigenvalues are normalized characters which can be combinatorially computed in at least two known ways. 
The numerator $ \chi_R   ( T_k)$ is given by $\chi_R   ( T_k) = |T_k| \chi_R ( \sigma ) $ for $ \sigma \in \cC_k$. The character $ \chi_R ( \sigma )$ can be computed with the combinatorial Murnaghan-Nakayama rule  \cite{MurnaghanOnReps} \cite{Nak}. The dimension $d(R)$ is obtained from the hook formula for dimensions. Another combinatorial formula gives $ { k \chi_R ( T_k ) \over d (R) } $ \cite{StanleyConjecture}\cite{FerayProof}.

The vectors in the  Fourier subspace for a triple $(R_1, R_2 , R_3)$ solve the following matrix equation  
\bea\label{stack}
 \left[\begin{array}{c}\
\cM^{(1)}_2  - { \chi_{R_1} ( T_2 ) \over d({R_1}) }  \\
\vdots \\
\cM^{(1)}_{\widetilde k_*} -  { \chi_{R_1} ( T_{\widetilde k} ) \over d({R_1}) } \\
\cM^{(2)}_2  -  { \chi_{R_2} ( T_2 ) \over d({R_2}) } \\
\vdots \\
\cM^{(2)}_{\widetilde k_* } - { \chi_{R_2} ( T_{\widetilde k_*} ) \over d({R_2}) }  \\
\cM^{(3)}_2   -  { \chi_{R_3} ( T_2 ) \over d({R_3}) }  \\
\vdots \\
\cM^{(3)}_{\widetilde k_* } - { \chi_{R_3} ( T_{\widetilde k_* } ) \over d({R_3}) }  
\end{array}\right] 
 \cdot v 
  = {\bold 0 } 
\eea
This rectangular array gives the matrix elements of a linear operator mapping  $ \cK(n)$ 
to $3 (\widetilde k_* -1 ) $ copies of  $ \cK(n) $,  using the geometric basis of ribbon graph vectors for $ \cK(n)$. From  Lemma \ref{IntegralityOfNormChar}, the normalized characters are integers.  Renaming as 
$\cL_{ R_1 , R_2 , R_3 } $ the integer matrix in  \eqref{stack} we have 
\bea\label{RecInt} 
\cL_{ R_1 , R_2 , R_3 } \cdot v = 0 
\eea

 We then  have, for each triple of Young diagrams, the problem of finding the null space of an integer matrix. Null spaces of integer matrices have integer null vector bases. These  can be interpreted in terms of lattices and can be constructed using integral algorithms. We will discuss the integrality properties of the null vectors and the associated interpretation in terms of lattices 
  further in section \ref{sec:YDHNF}.

\subsubsection{Computational implementation and examples}

The construction of this rectangular matrix using reconnections  on 
ribbon graph equivalence classes  along with normalized symmetric group characters is 
implemented using the software \cite{GAP}. This is described in appendix \ref{app:gap}. 
For the case $n=3$, $ k_*(n) = 2$ and we choose $ \widetilde k_* = k_* (3)  = 2$.  The three reconnection operators $  T_2^{(1)} , T_2^{(2)} , T_2^{(3)} $ suffice to distinguish the Young diagram triples. The matrices for each  operator are given in appendix \ref{appsub:nullspn3}. In this case the Kronecker coefficients $ C ( R_1 , R_2 , R_3 )$ are either $0$ or $1$. An  integer null basis vector for the Fourier subspace associated with each triple having non-vanishing Kronecker coefficient is given in appendix \ref{appsub:nullspn3}.  The vectors shown give the coefficients of the vectors $E_r $ for the index  $ r \in \{ 1, 2, \cdots , 11 \}$ and the graph associated with each index is 
shown in Figure \ref{fig:ribb}.

\subsection{Fourier subspace of triple as null space of square integer matrices} 
\label{sec:SquareIntMat}

We now show that Fourier subspace for a triple of Young diagrams can be obtained as the null space of a square matrix.  Rather than stacking the matrices for reconnection operators $T_k^{(i)}$ in a rectangular matrix, we will take linear combinations of these linear operators with integer coefficients. These linear combinations define Hamiltonians in the quantum mechanics of ribbon graphs. The coefficients are chosen with some care, using a procedure we explain,  to ensure 
that the eigenvalues of the Hamiltonian for distinct Young diagram triples are distinct.

\paragraph{Distinguishing Young diagram triples with quantum mechanical Hamiltonians}
\label{primes}

We consider Hermitian  Hamiltonians of the form  
\be \label{hamiltonianQ}
\cH = \sum_{ i =1}^{ 3 } \sum_{ k =2 }^{ k_* (n) } a_{  i , k  } T_k^{ (i)} 
\ee
with coefficients $ a_{  i , k  } $ which we will discuss shortly.  For simplicity we have taken $ \widetilde k_* = k_* (n)$ which is the minimum needed for the  list of normalized characters to distinguish Young diagrams with $n$ boxes. Using Proposition \ref{propTaQo} we have  
\be
\cH Q^{ R_1 , R_2 , R_3 }_{ \tau_1 , \tau_2 } = \left ( \sum_{  i . k  } a_{  i , k }  { \chi_{ R_i } ( T_k ) \over d(R_i ) }  \right )Q^{ R_1 , R_2 , R_3 }_{ \tau_1 , \tau_2 }
\ee
The Fourier subspace of $ \cK(n)$ for the ordered triple $ (R_1 , R_2 , R_3)$, which has dimension $ C ( R_1 , R_2 , R_3 )^2 $, is an eigenspace of these Hamiltonians.  
We will show that the coefficients $a_{ i,k}$ can be chosen as integers which ensure that the eigenvalues, which we denote as  $ \omega_{ R_1 , R_2 , R_3 } $ 
\be
 \omega_{ R_1 , R_2 , R_3 } =  \sum_{ i, k  } a_{  i , k  }  { \chi_{ R_i } ( T_k ) \over d(R_i )}  
\ee
distinguish the ordered  triples $ ( R_1 , R_2 , R_3  ) $.  In other words, we can choose integers $ a_{ i,k}$ to have 
\bea\label{tdH}  
&& \hbox{ \bf triple-distinguishing Hamiltonians with the property :   } \cr 
&& \hbox{ If }   ( R_1 , R_2 , R_3 ) \ne ( R_1' , R_2', R_3') \hbox{ then }   \omega_{ R_1 , R_2 , R_3 } \ne \omega_{ R_1' , R_2' , R_3'}
  \eea 
Characterising the general choice of integers $ a_{ i , k } $ which defines a Hamiltonian \eqref{hamiltonianQ} with the property \eqref{tdH} is an interesting problem: here we will  only show that these Hamiltonians exist, using a particular construction. 
As an operator on $ \cK(n)$, using the geometric ribbon graph basis,  the matrix elements of $\cH $ are integers 
\bea
\hbox { Coeff } ( E_s , \cH E_r ) = \sum_{ i =1}^{ 3 } \sum_{ k =2 }^{ k_* (n) } a_{  i , k  } ( \cM_k^{ (i)} )_r^s 
\eea
using Proposition \ref{PropIntTk} and \eqref{eqnMst}. 
For convenience, we will use the notation 
$\wchi_R ( T_k )$ for the normalized characters ${ \chi_R ( T_k) \over d(R) } $
which are integers by Lemma \ref{IntegralityOfNormChar}. 
The eigenvalues of the Hamiltonians are 
\bea 
 \omega_{ R_1 , R_2 , R_3 } =  \sum_{ i=1}^3 \sum_{ k=2  }^{ k_* }  a_{  i , k  }  \wchi_{ R_i } ( T_k ) 
   =   \sum_{ k  =2  }^{ k_* (n) }   a_{ 1, k }   \wc_{ R_1} ( T_k  )   + a_{ 2,k} 
   \wc_{ R_2} ( T_k  )    + a_{ 3,k}    \wc_{ R_3} ( T_k  ) 
\eea
For each triple of Young diagrams $ ( R_1 , R_2  , R_3 )$ the operator 
\bea 
(  \cH - \omega_{ R_1 , R_2 , R_3 } ) \equiv \cH_{ R_1 , R_2 , R_3 } 
\eea
is, in the geometric ribbon graph basis, a square integer matrix. 
 The Fourier subspace of $ \cK( n )$ associated with the triple $ ( R_1 , R_2 , R_3 )$ is the space spanned by the null vectors $v$ of this operator 
\bea\label{HsubRInt}  
\cH_{ R_1 , R_2 , R_3 } \cdot v = 0 
\eea
We now turn to the demonstration that such triple-distinguishing Hamiltonians constructed from integer $a_{ i , k } $ indeed exist in general.

\vskip.2cm

\noindent
{\bf Triple distinguishing property for low $n$} \\
Consider the problem of establishing the property in \eqref{tdH} for the cases $ n = 2,3,4,5,7$ where $ k_*(n) = 2 $. 
In this case, simplify the notation 
\bea 
a_{ 1, 2} = a_1 ~;~ a_{ 2,2 } = a_2  ~;~  a_{ 3,2} = a_3 
\eea
Introduce a label $q$ which indexes partitions of $n$, so $R_q$ is a Young diagram with $n$ boxes. 
With the simplified notation 
\bea 
 X_q = \wc_{ R_q } ( T_2 ) 
\eea
the eigenvalues are  
\bea\label{evalslown}  
\omega_{ R_{ q_1} , R_{ q_2} , R_{ q_3} }  =
(  a_1  X_{ q_1}  + a_2 X_{ q_2}  + a_3 X_{ q_3 }  ) 
\eea
These $X_q$ are integers using Lemma \ref{IntegralityOfNormChar}.

\ 

\noindent{\bf Conditions on $a_i $} \\ 
The problem of finding $a_i$ to have Hamiltonians with the triple-distinguishing property is now the problem of 
finding $ a_1 . a_2 , a_3 $ such that for any distinct triples $ ( R_{ q_1} , R_{ q_2} , R_{ q_3} ) \ne  ( R_{ q_1'}  , R_{ q_2'} , R_{ q_3'} ) $ the eigenvalues in \eqref{evalslown}  $ \omega_{ R_{ q_1} , R_{ q_2} , R_{ q_3} } \ne \omega_{ R_{ q_1'} , R_{ q_2'} , R_{ q_3'} }$. 
In other words,  the problem is to 
\bea\label{Lindep} 
&& \hbox{\bf Find integers $a_1 , a_2 ,  a_3 $ such that  } \cr 
&& a_{ 1} ( X_{ q_1} - X_{ q_1'} ) + a_2 ( X_{ q_2} - X_{ q_2'} ) + a_3 ( X_{ q_3} - X_{ q_3'} ) = 0 \cr 
&& \hbox{is only satisfied when $ X_{ q_1} = X_{ q_1'} , X_{ q_2} = X_{ q_2'} , X_{ q_3} = X_{ q_3'} $} 
\eea
Equivalently, the problem is to find integers $(a_1 , a_2 , a_3 )$ such that the sum 
\bea 
a_{ 1} ( X_{ q_1} - X_{ q_1'} ) + a_2 ( X_{ q_2} - X_{ q_2'} ) + a_3 ( X_{ q_3} - X_{ q_3'} )
\eea
is never zero for any choice of distinct triples 
$ ( q_1 , q_2 , q_3 )$ and  $(q_1' , q_2' , q_3') $. Note that two triples are considered distinct if they differ in any of the 3 slots. E.g. if $ q_1 \ne q_1'$ then $ ( q_1 , q_2 , q_3 ) $ and $ ( q_1' , q_2 , q_3 )$ are distinct triples.

Given the condition on the $a_i$, none of the $a_i$ can be zero. Suppose the contrary, e.g.  $a_1 = 0$. Then we can take $ X_{ q_1} \ne X_{ q_1'} $ but 
$ X_{ q_2} = X_{ q_2'} , X_{ q_3} = X_{ q_3'}$, and get a solution to (\ref{Lindep}). 
Let us look for a solution where $ a_1 =1$. 

The possible differences $ X_{ q } - X_{ q'} $ at fixed $n$ form a finite set of values. 
For example at $ n=3$, they can be $ ( 0 , 3 , -3 , 6 , -6 )$. The list of non-zero $ X_{ q } - X_{ q'} $ has prime factors $ 2,3$. Take a prime $p_1$ which  is not one of these prime factors. E.g. in this $ n=3$ case, take $ p_1 = 5$. Let $a_2 = p_1$.  This ensures that, when $ X_{ q_3} = X_{ q_3'}$, we cannot solve (\ref{Lindep}) : the second term has a prime factor $p_1$ while the first does not, so they cannot add up to zero.  
 For $a_3$ we pick another prime $p_2 $ which is not $ p_1$ and does not appear among the prime factors of $ X_q - X_{q'}$ for any $ q , q'$.  This ensures that the condition \eqref{Lindep} on the $a_i$ is satisfied for all triples where $ X_{q_2} = X_{ q_2'}$. 

Now we consider the generic case where
\bea\label{generic}  
X_{ q_1 } \ne X_{ q_1'} ,   X_{ q_2 } \ne X_{ q_2'} ,  X_{ q_3 } \ne X_{ q_3'}  \, . 
\eea 
To continue satisfying \eqref{Lindep} we can choose  $a_3$ large enough that the last term cannot be cancelled by the sum of first two terms. So pick  $p_2$ such that
\bea 
p_2 \Min_{ q, q'} | X_q - X_{ q'}  |  >  \Max (   |   p_1( X_{q_2} - X_{ q_2'} )   + 
 (  X_{q_1} - X_{ q_1'}  ) | )  
\eea 
Using the inequality, 
\bea
&&  p_1 \Max_{ q_2 , q_2'}   |  X_{q_2} - X_{ q_2'}  | +  
  \Max_{ q_1 , q_1'}   | X_{q_1} - X_{ q_1'}  | \ge\Max (    |   p_1 ( X_{q_2} - X_{ q_2'} )   + 
 (  X_{q_1} - X_{ q_1'}  ) | )  \cr 
 && 
\eea
we can write a computationally simpler condition 
\bea\label{csd}  
p_2 \Min_{ q, q'} | X_q - X_{ q'}  | > ( p_1 +1) \Max_{q , q'}  | X_{q} - X_{ q'} |
\eea
Since we are in the case \eqref{generic},  $\Min_{ q, q'} | X_q - X_{ q'}  | > 0$. 

In the example of $n=3$, choosing  $ p_1 = 5  $  as explained above, pick $p_2 = 13 $, which satisfies 
\eqref{csd}  because $ 13 * 3 > 5 * 6 +  6 $. 
We conclude that the  choice $ ( a_1 , a_2 , a_3 ) = ( 1, 5 , 13)$  at $ n=3$, satisfies the condition \eqref{Lindep}   and  the eigenvalues of $ \cH = T_2^{(1)} + 5 T_2^{(2)} + 13 T_2^{(3)} $ distinguish the triples $ ( R_1  , R_2 , R_3 )$ which label the $Q$-basis, and the degeneracies of the eigenspaces are precisely the squares of Kronecker coefficients. 

\vskip.2cm 

\noindent 
{\bf Triple distinguishing property for general $n$ } \\ 
We now have eigenvalues of $ \cH$ equal to 
\bea\label{genCOND}  
\omega_{ R_{ q_1} , R_{ q_2} , R_{ q_3} } = 
\sum_{ k =2 }^{ k_* }  (  a_{ 1 , k }  X_{ q_1 ,  k  } + a_{ 2 , k } X_{ q_2  , k } + a_{ 3 , k } X_{ q_3  , k }   )  
\eea
with  
\bea 
 X_{ q_i  , k } = \wc_{ R_{q_i } } ( T_k )  
\eea
We choose  $a_{i,k} $ to have Hamiltonians with the triple-distinguishing property, i.e. 
  $ \omega_{ R_{ q_1} , R_{ q_2} , R_{ q_3} }$ $  \ne \omega_{ R_{ q_1'} , R_{ q_2'} , R_{ q_3'} }$ 
for $ ( R_{q_1} ,R_{ q_2} , R_{ q_3} ) \ne ( R_{ q_1'} , R_{ q_2'} , R_{ q_3'} ) $. This the problem  
\bea\label{NCS} 
&& \hbox{\bf Find integers $a_{ 1 , k } , a_{ 2,k} , a_{ 3 , k }   $ such that  } \cr 
&& \sum_{ k=2 }^{k_*}  a_{ 1, k } ( X_{ q_1 , k } - X_{ q_1' , k } ) + 
 a_{ 2 , k } ( X_{ q_2  , k } - X_{ q_2' , k } ) + a_{ 3 ,k } ( X_{ q_3  , k } - X_{ q_3' , k }    )  = 0 
\cr 
&& \hbox{is only satisfied when $ X_{ q_1} = X_{ q_1'} , X_{ q_2} = X_{ q_2'} , X_{ q_3} = X_{ q_3'} $} 
\eea
The previous strategy for low $n$ extends here. Suppose $ q_1 \ne q_1'$ but $q_2 = q_2' , q_3 = q_3'$. 
In this case, we need to make sure 
that the $a_{ 1,k}$ are chosen such that for any pair $ q_1 , q_1'$ 
\bea\label{NCS1}  
 \sum_{ k=2 }^{k_*}  a_{ 1, k } ( X_{ q_1 , k } - X_{ q_1' , k } )  \ne 0 
\eea
One scheme for producing such a collection of $ a_{ 1, k } $ is to use prime decompositions again. Consider the differences $ X_{ q_1 , k  } - X_{ q_1' , k } $
as $ q_1, q_1'$ range over distinct pairs.  Consider the set of prime factors, denoted $ \texttt{PrimesDiffs} ( n  , k ) $  appearing in the integer differences  $ X_{ q_1 , k  } - X_{ q_1' , k } $ as $ q_1, q_1'$ range over distinct pairs.  Choose $ a_{ 1, 2} =1$ and $ a_{ 1,3} = p_1$, with $ p_1 \notin \PRIMESDIFFS (n , 2 )$. 
Then $ a_{ 1,4} = p_2 $ is a bigger prime chosen such that $ p_2 \notin \{ p_1 \} \cup \PRIMESDIFFS ( n , 2 ) \cup \PRIMESDIFFS ( n , 3 ) $ and 
\bea 
\Max_{ q, q'}   | ( X_{ q, 2 } - X_{ q' , 2 } ) |  + 
p_1  \Max_{ q, q'}  |  ( X_{ q, 3 } - X_{ q' , 3 } ) |  < p_2 \Min_{ q, q'}  |  ( X_{ q, 4 } - X_{ q' , 4 } ) | 
\eea
By iterating this procedure, we select  increasing primes $ p_1 , p_2 , \cdots , p_{ k_* -2} $ to ensure (\ref{NCS1}). 

Back to considering \eqref{NCS}: the case $ q_1 = q_1', q_2 \ne q_2' , q_3 = q_3'$ requires 
\bea\label{NCS2}   
 \sum_{ k=2 }^{k_*}   a_{ 2, k } ( X_{ q_2 , k } - X_{ q_2' , k } )  \ne 0
\eea
The case $ q_1 = q_1' , q_2 =q_2', q_3 \ne q_3'$ requires
\bea\label{NCS2}  
 \sum_{ k=2 }^{k_*} a_{ 3, k } ( X_{ q_3 , k } - X_{ q_3' , k } )  \ne 0
\eea
We also need to ensure that the condition  \eqref{NCS} holds when two of the $q$'s are distinct and when all three are distinct.  We can pick 
\bea \label{aikgen}
&&
( a_{ 1 , 2}  , a_{ 1, 3 } ,  \cdots , a_{ 1, k_* } ) = ( 1 , p_1 , p_2 , \cdots ,  p_{ k_* - 2} ) \cr  
&&
( a_{ 2,1} , a_{ 2,2 } \cdots , a_{ 2 , k_* -2} ) =  p_{ k_* -1} ( 1 , p_1 , \cdots , p_{ k_* -2} )  \cr 
&&
( a_{ 3,1} , a_{ 3,2} , \cdots , a_{ 3 , k_* -2}  ) = p_{ k_* } ( 1 , p_1 , \cdots , p_{ k_* -2} )  
\eea
The primes are chosen such that $p_{ 1 } < p_2 < \cdots <  p_{ k_* -2} < p_{ k_* -1} <p_{ k_* } $, 
with  $ p_{ k_* -1} $ such that 
\bea 
&&
 \Max_{ q, q'} \Big(   \Big| \sum_{ k=2 }^{k_*}  p_{  k  - 2 }  ( X_{  q , k } - X_{ q' , k } )  \Big| \Big)  
 \le  \crcr
  &&
  \Max_{ q, q'} \Big(   \sum_{ k=2 }^{k_*}  p_{  k  - 2 } |  X_{  q , k } - X_{ q' , k } | \Big)  <   
 p_{ k_* - 1 }\,   \Min_{ q, q'} \Big(  \Big|\sum_{ k=2 }^{k_*}   p_{ k - 2  }   (  X_{  q , k } - X_{ q' , k } )   \Big| \Big)
\eea
where we extend the sequence $p_{l}$ to $p_0 =1$, and also such that $ p_{ k_* } $ obeys 
\bea 
 (1+   p_{ k_* -1} )  \Max_{ q, q'} \Big( \sum_{ k=2 }^{k_*}   p_{ k - 2  }  | X_{  q , k } - X_{ q' , k } | \Big)   < p_{ k_* }   \Min_{ q, q'} \Big( \Big| \sum_{ k=2 }^{k_*}    p_{ k - 2  }  (  X_{  q , k } - X_{ q' , k } )  \Big|  \Big)  
\eea
The $\Min$ on the RHS is non-zero since this condition is coming from the vase $ q_1 \ne q_1' , q_2 \ne q_2' , q_3 \ne q_3'$. 
With these integer choices of $a_{ 1, k } , a_{ 2, k } , a_{ 3 , k } $ we can ensure 
that $ \cH $ has eigenvalues which distinguish the  triples $ ( R_1 , R_2 , R_3 ) $ in the Fourier 
basis elements $ Q^{ R_1 , R_2 , R_3 }_{\tau_1 , \tau_2 }  $. The dimensions of the distinct eigenspaces are 
$ C ( R_1 , R_2 , R_3 )^2$.

\paragraph{Examples}
 
In fact, for $k=2$, $ \Max_{ q , q'} | X_{q,2} - X_{q',2}  | $ is known  and equals $ 2\cdot | \wc_{ [1^n] } ( T_2 ) |= n(n-1)$ and  $\Min_{ q, q'} | X_{q,2 }- X_{ q',2}  |$ 
cannot be lower than $1$ since we know that $ X_{ q,2} , X_{ q',2}$ are integers. Using this lower bound 
\bea
p_2  \Min_{ q , q'} | X_{q,2} - X_{q',2}  |  > p_2 > (p_1 +1) n(n-1)
\eea
Thus picking the minimal prime $p_2$ larger than $(p_1 +1) n(n-1)$
would solve the inequality in \eqref{csd}. When $\Min_{ q, q'} | X_{q,2} - X_{q',2}  | >1$
then the above is a still sufficient condition but does 
not lead to the smallest  $p_2$. After some illustrations, 
we will discuss  sufficient conditions that leads to 
other solutions of the problem. 

\ 

\noindent{\bf Case $n=5$.}
Here $k_* =2$ and we have 
\be
 X_{  q , 2 } - X_{ q' , 2 } 
 \in \{  -20, -15, -12, -10, -8, -7, -5, -4, -3, -2, 0, 2,3,4,5,7,8, 10,12,15, 20 \}
\ee
with $ \wc_{ [1^n] } ( T_2 )  = -10$
so that  $  \Max_{ q, q'} | X_{  q , 2 } - X_{ q' , 2 } | = 20$, 
and  $  \Min_{ q, q'} | X_{  q , 2 } - X_{ q' , 2 } | = 2>0$. 
The set of prime divisors of the above set
is $\{2,3,5,7\}$. We choose $p_1= 11$ and therefore the inequality \eqref{csd} 
becomes 
\bea
2 p_2  > 12 \times 20 = 240 
\eea
Hence we choose $p_2 = 127$, and the triple $(p_0, p_1, p_2) = (1,11,127)$  solves \eqref{Lindep}. 
The Hamiltomnian 
 $ \cH = T_2^{(1)} + 11 T_2^{(2)} + 127 T_2^{(3)} $
is  an integer matrix in the geometric ribbon graph basis, with the property that distinct Young diagram triples are associated with distinct eigenvalues, and the eigenvalue degeneracies are given by $ C ( R_1 , R_2 , R_3)^2$.

\

\noindent{\bf Case $n=7$.} Again $k_* = 2$ should be the max of $k$. 
We have 
\bea
&&
 X_{  q , 2 } - X_{ q' , 2 } 
 \in \{  
-42, -35, -30, -28, -27, -24, -23, -22, -21, -20, -18, -17, -16, 
\crcr
&&
-15, -14, 
  -13, -12, -11, -10, -9, -8, -7, -6, -5, -4, -3, -2, -1, 0, 1, 2, 3, 
   \crcr
  && 4, 5, 
  6, 7, 8, 9, 10, 11, 12, 13, 14, 15, 16, 17, 18, 20, 21, 22, 23, 24, 27, 28, 
  30, 35, 42 \} \,, 
  \eea
with $ \wc_{ [1^n] } ( T_2 )  = -21$,
 $  \Max_{ q, q'} | X_{  q , 2 } - X_{ q' , 2 } | = 42$, 
and  $  \Min_{ q, q'} | X_{  q , 2 } - X_{ q' , 2 } | = 1$. 
The set of prime divisors is $\{2,3,5,7,23,17,13,11\}$. 
Choose $p_1 = 19$, and then we seek 
\bea
p_2 > 20 * 42 = 840 \,. 
\eea
We fix $p_2 =853$ and $(a_1 , a_2 , a_3 ) = ( p_0, p_1, p_2) = (1,19,853)$  is one correct triple 
 solves \eqref{Lindep}. Thus $ \cH = T_2^{(1)} + 19 T_2^{(3)} + 853 T_2^{(3)} $ is  an integer matrix in the geometric ribbon graph basis, with the property that distinct Young diagram triples are associated with distinct eigenvalues, and the eigenvalue degeneracies are given by $ C ( R_1 , R_2 , R_3)^2$.

\

\noindent{\bf Case $n=6$.}
In this case  $k_* = 3$ and 
\bea
&&
 X_{  q , 2 } - X_{ q' , 2 } 
 \in \{  
 -30, -24, -20, -18, -15, -14, -12, -10, -9, -8, -6, -5, -4, -3, \crcr
 && 
 -2, 0, 2, 
      3, 4, 5, 6, 8, 9, 10, 12, 14, 15, 18, 20, 24, 30 \} \cr\cr
 && 
 \PRIMESDIFFS( 6  , 2  ) = \{2, 3,5, 7\} \,, \quad  \Min_{q,q'} | X_{  q , 2 } - X_{ q' , 2 }|  = 2 \;; 
    \cr\cr 
&&
 X_{  q , 3 } - X_{ q' , 3 } 
 \in \{   -48, -45, -40, -36, -24, -21, -16, -12, -9, -8, -5, 
 -4, -3, 0,  3, \crcr
&& 4, 5, 
      8, 9, 12, 16, 21, 24, 36, 40, 45, 48 \} 
      \cr\cr
 && 
 \PRIMESDIFFS( 6  , 3  ) = \{ 2,3,5,7 \} 
  \,, \qquad  \Min_{q,q'} | X_{  q , 3 } - X_{ q' , 3 }|  =  3 \, . 
 \eea 
 We follow the procedure and require $p_1 \notin \PRIMESDIFFS( 6  , 2  ) $,
hence, for instance $p_1 = 11$. Then we seek $p_2$ that obeys  
 \bea
 &&
 \Max_{q,q'} \Big(  | X_{  q , 2 } - X_{ q' , 2 }  | + p_1 |X_{  q , 3} - X_{ q' , 3}| \Big)  
 < p_2 \Min_{q,q'}  \Big|  ( X_{  q , 2 } - X_{ q' , 2 } ) + p_1 (X_{  q , 3} - X_{ q' , 3}) \Big| 
 \crcr
 &&
 \Max_{q,q'} \Big(  | X_{  q , 2 } - X_{ q' , 2 }  | + p_1 |X_{  q , 3} - X_{ q' , 3}| \Big)   
  =30 + p_1*48 = 30 + 11*48 =  558 \,, \crcr
&& 
  \Min_{q,q'}  \Big|  ( X_{  q , 2 } - X_{ q' , 2 } ) + p_1 (X_{  q , 3} - X_{ q' , 3})  \Big| 
   = 2 \,. 
\eea
Thus we seek  a prime $p_2$ such that
\bea
558 < 2 p_2 \,,  \qquad \quad  p_2 > 279\,. 
\eea
We then use $p_2 = 289$. It remains to determine $p_3$ satisfying 
\bea
&&
(1+p_2) \times 558 < 2 p_3 \crcr
&&
(1 + 289)\times 558  =161820   < 2p_3 \,,  \qquad \quad  p_3  > 80910. 
\label{p3}
\eea
that gives $p_3 =80917$. Thus the quadruple
is $ ( 1 , p_1 , p_2 , p_3) = (1, 11, 289, 80917)$ with $ a_{1,2} = 1 , a_{ 1,3}= 11 , a_{2,2} = 289  , a_{ 2,3} = 289 \times 11 = 3179 , 
a_{ 3,  2 } = 80917 , a_{ 3,3} = 80917 \times 11 = 890087 $   solves the condition \eqref{NCS}. As a result 
the Hamiltonian $ \cH = \sum_{ i ,k } a_{i,k} T_k^{(i)} $ with these coefficients has integer matrix elements in the ribbon graph basis, has distinct eigenvalues for distinct triples of Young diagrams, and eigenvalue degeneracies given by $ C ( R_1 , R_2 , R_3)^2 $.  

 \ 
 
\noindent{\bf Sufficient conditions.}
At smallest order of $k_*=2,3$,  there are quick sufficient conditions that solve the  problem, for all $n=2,3,4,\dots, 14$. We are confident that similar identities 
holds for higher order in $k_*$.   Note that the solutions $p_k$'s provided below need not be the smallest possible but we arrive at easy programming equalities. 

Consider first $k_* = 2$, pick $p_1 $ as the first prime number 
above $n(n-1)$ (as $\Max_{ q , q'} | X_{q,2} - X_{q',2}  | =   n(n-1)$). This already guarantees that it does not belong to the set of prime
divisors of the set $\{ X_q - X_{ q'}\} $.  Then choose the prime $p_2 > (p_1+1)*n(n-1)$ then $( a_1 , a_2 , a_3 ) = (1,p_1,p_2)$ solves the condition \eqref{Lindep}.

Addressing $k_* =3$, we can replace  $\Min_{q,q'}| \cdot |  $ by the lower bound $1$ and the $\Max_{q,q'}| \cdot | =2 
\hat\chi_{[1^n]}(T_k)$.
Let us illustrate this idea at $n=6$. Already, $p_1$ has been fixed to be the smallest prime $p_1 > n(n-1)$. 
We consider   
$ \Max_{ q } X_{q,3} =    \wc_{ [1^n] } ( T_3 ) = \frac{n(n-1)(n-2)}{3}$, 
and thus  
 $ \Max_{ q , q'} | X_{q,3} - X_{q',3}  | \le  2 \frac{n(n-1)(n-2)}{3}$. 
We choose 
\bea
 p_2 &>& \Max_{q,q'}( | X_{q,2} - X_{q',2}  |  ) + p_1 \Max (| X_{q,3} - X_{q',3}  |)\crcr
  &>&  n(n-1) + p_1  \frac23 n(n-1) (n-2)   \,. 
\eea
Thus we choose $p_2$ to be the next prime after $\frac{n}{3}(n-1)(3+2p_1(n-2))$. 
Last $p_3$ should obey the bound 
\bea
p_3 & >&  
(1+p_2)\Max_{q,q'}( | X_{q,2} - X_{q',2}  | + p_1 | X_{q,3} - X_{q',3}  |) 
 \cr\cr
&>&
(1+p_2)( n(n-1) + p_1  \frac23 n(n-1) (n-2) ) 
\eea
Therefore picking $p_3$ as the next prime larger than $
 \frac{n}{3}(n-1)  (1+p_2)(3 + 2p_1(n-2) )$ will solve the issue. 
 
At $n=6$,  $  \Max_{ q, q'} | X_{  q , 2 } - X_{ q' , 2 } | = 30 = 2  \wc_{ [1^n] } ( T_2 ) $, We can pick  $p_1 =31 $. Then $  \Max_{ q, q'} | X_{  q , 3 } - X_{ q' , 3 } | = 48<80 =  2\wc_{ [1^n] } ( T_3)$. 
This will fixe $p_2$ and $p_3$.  
Then an alternative quadruple that solves the problem is given 
by $(1,31,2521, 6330223)$ (to be compared with the previous
quadruple  $(1,11,289,80917)$  in equation \eqref{p3}).

\section{Kronecker coefficients and  ribbon graph sub-lattices  }
\label{sec:YDHNF}

In section \ref{sec:IntegerMatricesKron}, we constructions of an  integer matrix for each  ordered triple of Young diagrams 
$ (R_1 , R_2 , R_3 )$ with $n$ boxes, with the property that their null space gives a basis for the Fourier subspace of $ \cK ( n )$ associated with that triple. This subspace has 
dimension  equal to  the square of the Kronecker coefficient : $ C ( R_1 , R_2 , R_3 )^2 $. These matrices are constructed from central elements $ T_k^{(i)} $ 
(introduced in section \ref{sec:CentreKnReconn}) of the algebra $ \cK(n)$ of bipartite ribbon graphs with $n$ edges, where $ i \in \{ 1,2,3 \} $ and $ k \in \{ 2, 3 , \cdots , \widetilde k_* \}$. The parameter $ \widetilde k_* \in \{ k_*(n) , k_*(n)+1 ,  \cdots , n \}$, where $ k_* (n)$ is the minimal integer such that the  central elements $ \{ T_2 , \cdots  , T_{k_*(n)} \}$ generate the centre of $ \mC (S_n)$ and it has been computed for $n$ up to $79$ \cite{KR1911}. We have two constructions for each Young diagram triple, one producing a  rectangular matrix $  \cL_{ R_1 , R_2 , R_3 } $  \eqref{RecInt} and another producing a  square matrix $  \cH_{ R_1 , R_2 , R_3 } $ \eqref{HsubRInt}.  In each case, we are solving the linear equation 
\bea\label{Xv}  
X \cdot v = 0 
\eea
where $ X = \cL_{ R_1, R_2 , R_3 } $ or $ X = \cH_{ R_1 , R_2 , R_3 }$.

The null spaces  of integer matrices have bases given as integer vectors. This follows from the theory of Hermite normal forms and has an interpretation in terms of sub-lattices. In  the present application we have a lattice 
\bea 
\mZ^{ |\Rib(n) | } \subset \mR^{ |\Rib(n) | } 
\eea
which is interpreted as the space of integer linear combinations of the geometric ribbon graph basis vectors $E_r$ of the ribbon graph algebra $ \cK(n)$. We will refer to $\mZ^{ |\Rib(n) | } $ as the lattice of ribbon graphs. 
 In this section, we will explain the key facts about integer matrices and sub-lattices that we will need and state the first  main result of this paper, Theorem \ref{theo:C2}.  This  is the construction of $ C(R_1,R_2,R_3)^2 $ as the dimension of a sub-lattices of  the lattice of ribbon graphs.

A classic problem asks for a combinatorial construction of the Kronecker coefficient 
associated with every triple of Young diagrams \cite{Murnaghan,StanleyKron}. Recent progress on this problem from a number of directions  and its connections to computational complexity is summarised in \cite{PPV}. Theorem \ref{theo:C2} gives a { \it combinatorial interpretation}  of the square of the  Kronecker coefficients. 
The theory of Hermite normal forms for integer matrices also offers {\it combinatorial algorithms}  for finding 
the null spaces (Corollary \ref{CorrConst}).   It is  also interesting to ask if there is a  purely combinatorial proof  - without using representation theory - of the formula for Kronecker coefficients in terms of characters, which can be viewed as 
combinatorial objects, for example, by the Murnaghan-Nakayama algorithm. Our proof of Theorem \ref{theo:C2} relies in an important way on representation theory, e.g. in the derivation of Proposition \ref{propTaQo} which enters the proof, and is therefore not purely combinatorial. 
In section \ref{sec:CombConsDisc} we 
 discuss how the question of a  purely combinatorial proof of the Theorem \ref{theo:C2}
  raises interesting questions on integer matrices. 

In section \ref{app:conjugate}, we consider the $ S=\pm 1 $ eigenspaces of the conjugation operator defined in 
\ref{sec:Involution}. This leads to the definition of sub-lattices  of the lattice of ribbon  graphs with dimensions
$ C (R_1 , R_2 , R_3 ) ( C ( R_1 , R_2 , R_3 ) +1 ) /2 $ and $C (R_1 , R_2 , R_3 ) ( C ( R_1 , R_2 , R_3 ) -1 ) /2 $, constructed as null spaces of integer matrices.  
The difference of these dimensions is $ C ( R_1 , R_2 , R_3 )$ which can therefore 
 be constructed by choosing a map from a  basis set for the null vectors, determined  for example 
 by a Hermite normal form algorithm, for the smaller sub-lattice to a basis set for the larger sub-lattice (Theorem \ref{addLatt}).

\subsection{Null-vectors of integer  matrices and lattices} 
\label{nulllattices}

The null space of the integer matrix   $X$ defined by \eqref{Xv}   is the span of a set of null vectors which can be chosen to be integer vectors, i.e. integral linear combinations of the $E_r$. A key result from the theory of integer matrices and lattices is that any integer matrix $A$ (square or rectangular) has a unique Hermite normal form (HNF) (as explained in textbooks such as \cite{Cohen}\cite{Schrijver} or online notes such as \cite{MicciancioBA}).  These can be computed using mathematical software such as GAP, SAGE or Mathematica. Thus $ A $ has a  decomposition $A= U h$: $U $ is a unimodular matrix, i.e. an integer  matrix of determinant $ \pm 1$. In the following discussion we will use $ A = X^T$.  $h$ is an integer matrix with  the following properties :

\begin{itemize} 

\item $  h $ is upper triangular (that is, $h_{ij} = 0 ~\hbox{for}~  i > j$), and any rows of zeros are located below any other row.

\item  The leading coefficient (the first non-zero entry from the left, also called the pivot) of a non-zero row is always strictly to the right of the leading coefficient of the row above it; moreover, it is positive.

\item 
The elements below pivots are zero and elements above pivots are non-negative and strictly smaller than the pivot.

\end{itemize} 

The construction of $h$ proceeds by applying a sequence of operations involving one of the following 
in each step. 

\begin{itemize} 

\item Swop two rows. 

\item Multiply a row by $-1$. 

\item Add an integer multiple of a row to another row of $A$.

\end{itemize} 

Each of these operations corresponds to left multiplication of $A$ by a unimodular matrix $U$: $ A \rightarrow U A $. Every integer   matrix $A$ can be brought into HNF by a sequence of these elementary integer row operations. 
Suppose we want to find the vectors $v$ which obey $ X v =0$. Equivalently $v^T X^T =0$. 
We apply  the elementary integer row operations to bring $ X^T $ to HNF. 
This means $ X^T = U h $. The number $N$  of lower rows of zeroes in $h$ is the dimension of the null space of $X$. The integer null vectors can be read off from the lower  $N$ rows of $U$. The non-zero rows of $ h$ 
give an integer basis for the image of $ X$. 

As a simple example to illustrate these properties, take 
\bea 
X = \begin{pmatrix} 1 & 1 \cr 2 & 2 \end{pmatrix}  \,, 
\qquad \quad  
X^T = \begin{pmatrix} 1 & 2 \cr 1 & 2 \end{pmatrix}  \, , 
\eea
with $X^T$ being the transpose of $X$. 
By applying the row operation of replacing the second row $R_2$ by $R_2 - R_1$ we get the HNF 
\bea 
X^T \rightarrow h=\begin{pmatrix} 1 & 2 \cr 0 & 0 \end{pmatrix} 
\eea
The unimodular matrix which implements this row operation is 
\bea 
U = \begin{pmatrix} 1 & 0 \cr 1 & -1 \end{pmatrix} 
\eea
i.e 
\bea 
U X^T = \begin{pmatrix} 1 & 2 \cr 0 & 0 \end{pmatrix}  = h 
\eea
The lower row of $ U$, when transposed, gives the null vector for the action of $X : v \rightarrow Xv $
\bea 
X \begin{pmatrix} 1 \cr -1  \end{pmatrix}  = \begin{pmatrix} 0 \cr 0   \end{pmatrix}  
\eea
The non-vanishing row of $ h $, transposes to the column vector which  gives the image of $X$ since 
\bea 
X \begin{pmatrix} x_1  \cr x_2  \end{pmatrix}  = ( x_1 + x_2)\begin{pmatrix} 1 \cr 2 \end{pmatrix} 
\eea

To see that the connection between the lower rows  of the unimodular matrix $U $ in the decomposition 
$ U X^T =  h $ corresponds to null vectors, observe that 
\bea 
U  X^T  =  h 
\eea
can be written as 
\bea 
\sum_{ k } U_{ ik}  X^T_{ kj } = h_{ i j } 
\eea
The vanishing rows of $h$ correspond to values of $i$ such that $h_{ ij} =0 $ for all $j$. 
Fixing one of these $i$ we have vectors $ U_{ ik} $ as $k$ varies, with the property : 
\bea 
\sum_{ k }  X_{ jk} U_{ ik} = 0 
\eea
Note that we could have equivalently   worked with elementary column operations on $ X$ rather than elementary row operations on $X^T$. 

By definition the  matrix $U$ has integer entries, so this construction gives 
an integer basis for the null space of $X^T$. 
The null vectors of $X$, found as integer linear combinations of $E_r$, define a sub-lattice of the lattice  $\mZ^{|\Rib(n)|}$. The dimension of the sub-lattice is 
the square $ ( C ( R_1 , R_2 , R_3 ))^2 $ of the Kronecker coefficient. The square of the Kronecker coefficient thus has the combinatorial interpretation as the dimension  of a sub-lattice of the lattice of ribbon graphs. We have thus arrived at the first main theorem of this paper.

\begin{theorem}\label{theo:C2}
For every triple of Young diagrams $(R_1 , R_2 , R_3 ) $ with $n$ boxes, the lattice 
\bea 
\mZ^{ | \Rib(n) | }
\eea
of integer linear combinations of the geometric basis vectors $E_r$ of $ \cK ( n ) $ contains a sub-lattice 
of dimension $ ( C ( R_1 , R_2 , R_3 ))^2 $ spanned by a basis of  integer null vectors of the operator 
$ X$, which is $ \cL_{ R_1, R_2 , R_3 } $ from \eqref{RecInt}  in the rectangular matrix  construction 
or $ \cH_{ R_1 , R_2 , R_3 } $  from \eqref{HsubRInt} in the square matrix construction. 
\end{theorem}

Solving for the null vectors of $X$ using the HNF shows  that there is sub-lattice 
in the lattice of ribbon graphs whose dimension is $ (C ( R_1 , R_2 , R_3 ))^2 $. 
This gives a combinatorial interpretation for  the square of the Kronecker coefficient. 

But the theory of lattices is even more powerful. The columns of $X$ (equivalently the rows of $X^T$)  are a set of vectors in the lattice of ribbon graphs. The space spanned by the integer linear combinations of 
these vectors  is a  sub-lattice of dimension $|\Rib(n)|  - C ( R_1 , R_2 , R_3 )^2 $
  (that is the column rank of $X$ or the row rank $X^T$). 
The process of arriving at the HNF through row operations on $X^T$ amounts to simplifying the description of this sub-lattice until it is given as the integer linear combinations of a linearly independent set of integer vectors which sit in the rows of $h$. This process also defines a unimodular matrix which encodes the integer null vectors of $X$. Each step in  the process of row operations acts on the set of lattice vectors in $X$, and can thus be viewed as constructive  combinatorial steps. 

\vskip.2cm 

\begin{corollary}\label{CorrConst} 
There is a constructive procedure for the sub-lattice in Theorem \ref{theo:C2}  consisting of integer row operations on the list of integer rows of $X^T$, which produce the HNF of $X^T$. 
\end{corollary}
\begin{proof}
The  treatment of rows of $X^T$ to put it in a HNF form $X^T = U h$ is a combinatorial construction consisting of a discrete sequence of  integer row operations (swop, multiplying by $-1$ and integer linear combinations
of rows).  The outcome $h$ of the HNF construction gives a basis for the sub-lattice of dimension $ |\Rib (n) | - C ( R_1 , R_2 , R_3 )^2 $ spanned by integer linear combinations of the  rows of $X^T$. The outcome $U$ is built in successive steps by matrices implementing the integer elementary row operations on $X^T$. At the start of an algorithm for the HNF of  $X^T$, the rows give a generically over-complete basis for the lattice generated by these rows. $X^T$ is modified step by step until the last step produces $h$.  At each step the intermediate matrix has a list of lattice vectors. At the end of an algorithm for the HNF, there is a sequence of rows of zeros in $h$ and the corresponding rows of $U$ record the integer null vectors, the number of which is $C( R_1 , R_2 , R_3)^2$. 
The construction of   $U$ associated with a given $X$  is thus a sequence of combinatorial operations on lattice vectors in $ \mZ^{ |\Rib(n)|}   $. 

\end{proof}

The key fact from the theory of integer matrices and lattices we have used is the existence and uniqueness of the HNF.  In the above we have focused on the fact that integral algorithms exist which produce from $X$, the null vectors and  the HNF. We have not focused on the computational complexity of 
the problem.  We make some initial remarks in this direction. 
The LLL algorithm \cite{LLL}  is known to calculate HNF's in a time that is polynomial in the 
size of the matrix. Our matrices are very large - grow as the number of ribbon graphs. We know from \cite{Sanjo} for a partition $\sum_i  ip_i = n$
\bea
|\Rib(n)| = \sum_{p \vdash n}  \prod_{i=1}^n i^{p_i} (p_i!)
\eea
The asymptotics of this number is known. For instance \cite{OEIS}, 
$ |\Rib(n)|  \sim n! * (1 + 2/n^2 + 5/n^3 + 23/n^4 + 106/n^5 + 537/n^6 + 3143/n^7 + 20485/n^8 + 143747/n^9 + 1078660/n^{10})$, for the coefficients see 
A279819. Thus the data size of our problem 
already grows like $\cO(n!)$ (assuming that $k_*(n)\ll n$ as $n \to \infty$). It seems that combined with a problem of time complexity, our problem entails an super-exponential complexity in (memory) space. 
Hence, although the time complexity of HNF could be polynomial 
in the data size, it would remain $\cO(n!)$. A more thorough discussion of the complexity of the algorithm for construction the sub-lattice in Theorem \ref{theo:C2} taking into account the group theoretic characteristics of the integer matrix $X$  is left for the future.

\subsection{Combinatorial interpretations, algorithms and proofs}
\label{sec:CombConsDisc}

An interesting  question for a combinatorial construction of Kronecker coefficients posed in \cite{StanleyKron} is whether it gives a new proof of the fact that  these coefficients are non-negative integers. It is of course obvious from representation theory that $ C ( R_1 , R_2 , R_3 )$ is non-negative-integer - it is the number of times $R_3$ appears in the tensor product decomposition $ R_1 \otimes R_2$ when viewed as a representation using  the diagonal action of permutations. The character formula 
\bea\label{CMN}  
C ( R_1  , R_2 , R_3 ) = { 1 \over n! } \sum_{ \sigma \in S_n } \chi_{ R_1 } ( \sigma ) \chi_{ R_2 } ( \sigma ) \chi_{ R_3 } ( \sigma ) 
\eea 
- although it can be derived from representation theory - can also be viewed as a purely combinatorial formula, where the characters are given for example by the Murnaghan-Nakayama combinatorial rule. 
From the purely combinatorial  point of view, the non-negative integer property is not manifest.

The sub-lattice interpretation of Kronecker coefficients  (Theorem \ref{theo:C2}) makes it manifest that 
they are non-negative integers. Algorithms for computing the HNFs are combinatorial operations on lists of lattice vectors. Our proof of the interpretation and of the validity of the algorithms has relied on an important input from 
representation theory (Proposition \ref{propTaQo}). An interesting question is whether lattices of ribbon graphs offer an avenue to provide a purely combinatorial understanding, with no representation theory input, for the non-negativity of $ C ( R_1 , R_2 , R_3 )$, defined by the formula \eqref{CMN} in terms of characters computable by the combinatorial Murnaghan-Nakayama rule. To give some context to this question, consider the equality 
of $n!$ with the sum of  squares of the dimensions of irreducible representations of $S_n$, which can be derived 
using representation theory. This can also be derived purely combinatorially using the Robinson-Schensted correspondence which gives a bijection between permutations  in $S_n$  and pairs of standard Young tableaux having the same shape and $n$ boxes (see for example a textbook reference \cite{FultonYoung}).

This raises some questions on the non-negative integer matrices $T_k^{(i)} $, the rectangular integer matrices in section \ref{sec:reconn} and the square matrices of Hamiltonian matrix elements in section \ref{sec:SquareIntMat}. 
The first step would be to derive formulae for the eigenvalues $T_k^{(i)} $
recovering the Murnaghan-Nakayama combinatorics of 
 these eigenvalues directly from these matrices built using the reconnection matrices $ T_k^{(i)}$. 
 The second step would be to show that the eigenvalue degeneracies are given by \eqref{CMN}, viewed as an  expression for the degeneracies in terms of the eigenvalues. Any integer matrix is known to have a Smith normal form (SNF)  which can be calculated by algorithms generalizing to those for HNFs \cite{Cohen}. In the SNF for $X$, we have $ X = U D V $, where the matrix $D$ is a diagonal matrix of singular values. The relation between these singular values and the eigenvalues of an integer matrix $X$ has been studied in \cite{Lorenzini}. Singular values in the SNF are accessible to integer-matrix algorithms while  eigenvalues enter the link between the integer matrices at hand and the Kronecker coefficients. Better understanding this link could potentially help towards a purely combinatorial proof of the interpretation and algorithms for Kronecker coefficients based on Theorem \ref{theo:C2}.

\subsection{Conjugation action and additional sub-lattices }
\label{app:conjugate}

In  section \eqref{nulllattices}, the HNF of integer matrices
 to a $C(R_1 , R_2 , R_3)^2$ dimensional sublattice of ribbon graphs
and has provided a refinement of the counting of all ribbon graphs. 
We now describe integer matrices which will lead us to a  sublattice interpretation of $C(R_1 , R_2 , R_3)$. 

In section \ref{sec:Involution} we defined a conjugation operator $S$ \eqref{inverse} which 
satisfies $S^2 = \id$.   The conjugation either maps a ribbon graph to itself $S(E_r)= E_r$, or distinct pairs $ E_s \ne E_t $ are related by $ S ( E_s ) = E_t ; S  ( E_t ) = E_s $.
We refer to the former as self-conjugate ribbon graphs and the the latter as conjugate pairs.

In order to illustrate the action of $S$ consider $n=3$.  Inversion of the permutation pairs 
representing a ribbon graph leaves the pair unchanged unless one of the permutations has a cycle of length $3$. For this $n=3$ case, all ribbon graph vectors $E_r$  are self-conjugate : inversion maps any representative pair or permutations to another pair within the same orbit. We list the orbits at $ n=3$ which involve a cycle of length $3$, to illustrate this property
\bea\label{Sconj3}
3 : &&  \quad \; \; [ [(),(1,2,3)]  \;,\;  [(),(1,3,2)] ] \crcr
7 :  && \quad \; \; [ [ (2,3), (1,2,3) ], [ (2,3), (1,3,2) ], [ (1,2), (1,2,3) ]  \crcr
&& \quad  \quad\quad [ (1,2), (1,3,2) ], [ (1,3), (1,2,3) ], [ (1,3), (1,3,2) ] ], \crcr
8 : && \quad \; \;  [ [ (1,2,3), () ], [ (1,3,2), () ] ], \crcr
9 : &&\quad \; \;      [ [ (1,2,3), (2,3) ], [ (1,2,3), (1,2) ], [ (1,2,3), (1,3) ], \crcr
&&  \quad  \quad\quad 
          [ (1,3,2), (2,3) ], [ (1,3,2), (1,2) ], [ (1,3,2), (1,3) ] ], \crcr
 10 : && \quad \; \;        [ [ (1,2,3), (1,2,3) ], [ (1,3,2), (1,3,2) ] ], \crcr
 11 : &&  \quad \; \;           [ [ (1,2,3), (1,3,2) ], [ (1,3,2), (1,2,3) ] ] ,
\eea
where the first column contains the labels (i.e 3,7,8, etc.)  of ribbon graphs of Figure \ref{fig:ribb}. As we will see shortly, this self-conjugation property can be understood using the action of $S$ on the Fourier basis of $ \cK ( n )$.

The following statement holds: 
\begin{proposition}
\label{propSQ}  
 Under the conjugation action, we have 
\bea
S ( Q^{ R_1 , R_2 , R_3 }_{\tau_1,\tau_2 } ) = Q^{ R_1 , R_2 , R_3 }_{ \tau_2,\tau_1 }
\eea
\end{proposition}
\begin{proof} Consider $Q^{ R_1 , R_2 , R_3 }_{\tau_1,\tau_2 }$
given by \eqref{qbasis}, then 
\bea 
&&
S ( Q^{R_1 , R_2 , R_3 }_{\tau_1,\tau_2} ) 
 = 
\kappa_{R_1 , R_2 }
\sum_{\s_1, \s_2 \in S_n}
\sum_{i_l, j_l}
C^{R_1 , R_2 ;  R_3 , \tau_1  }_{ i_1 , i_2 ; i_3 } C^{R_1 , R_2 ; R_3 , \tau_2  }_{ j_1 , j_2 ; i_3 } 
 D^{ R_1 }_{ i_1 j_1} ( \sigma_1  ) D^{R_2 }_{ i_2 j_2 } ( \sigma_2 )  \,  \sigma_1^{-1} \otimes \sigma_2 ^{-1} \crcr
 &&  =  \kappa_{R_1 , R_2 }
\sum_{\s_1, \s_2 \in S_n}
\sum_{i_l, j_l}
C^{R_1 , R_2 ;  R_3 , \tau_1  }_{ i_1 , i_2 ; i_3 } C^{R_1 , R_2 ; R_3 , \tau_2  }_{ j_1 , j_2 ; i_3 } 
 D^{ R_1  }_{ i_1 j_1} ( \sigma_1 ^{-1} ) D^{R_2}_{ i_2 j_2 } ( \sigma_2 ^{-1} )  \,  \sigma_1 \otimes \sigma_2 \crcr
 &&  = \kappa_{R_1 , R_2 }
\sum_{\s_1, \s_2 \in S_n}
\sum_{i_l, j_l}
C^{R_1 , R_2 ;  R_3 , \tau_1  }_{ i_1 , i_2 ; i_3 } C^{R_1 , R_2 ; R_3 , \tau_2  }_{ j_1 , j_2 ; i_3 } 
 D^{ R_1  }_{  j_1 i_1 } ( \sigma_1  ) D^{R_2}_{ j_2 i_2  } ( \sigma_2  )  \,  \sigma_1 \otimes \sigma_2 \crcr  
 && = \kappa_{R_1 , R_2 }
\sum_{\s_1, \s_2 \in S_n}
\sum_{i_l, j_l}
C^{R_1 , R_2 ;  R_3 , \tau_1  }_{ j_1 , _j2 ; i_3 } C^{R_1 , R_2 ; R_3 , \tau_2  }_{ i_1 , i_2 ; i_3 } 
 D^{ R_1  }_{  i_1 j_1  } ( \sigma_1  ) D^{R_2}_{  i_2 j_2   } ( \sigma_2  )  \,  \sigma_1 \otimes \sigma_2 \crcr  
 && = Q^{R_1 , R_2 , R_3 }_{\tau_2,\tau_1} 
\eea
We used the  fact that 
$D^R_{ ij} ( \sigma^{-1} )  = D^R_{ ji } ( \sigma ) $  and a relabelling of indices $ i_1, i_2 \leftrightarrow j_1 , j_2 $. 

\end{proof}

\noindent 
{\bf Remark} Proposition \ref{propSQ} implies that at $ n=3$, where $ C ( R_1 , R_2 , R_3 )  $ is ether $1$ or $0$, the only possible eigenvalue of $ S $ is $1$.  Considering then the action of $S$  on the geometric basis, we deduce that all the ribbon graphs must be self-conjugate. This is indeed confirmed by \eqref{Sconj3}.

 On the geometrical ribbon graph basis for $ \cK ( n )$, the action of $S$ can leave a ribbon basis element $E_r$ invariant, or it can pair the ribbon with another ribbon. Let us denote by $E_r^{(s)} $ the self-conjugate ribbons,  which stay invariant under conjugation. 
The  non-self conjugate pairs are $ ( E_r^{ (n)} , E_r^{ (\bar n) } )$. The $+1$ eigenspace of $S$ in $ \cK ( n )$ is spanned by $\{ E_r^{ (s)} , ( E_r^{(n)}  + E_r^{ (\bar n )}  )\}$. 
The $-1$ eigenspace of $S$ is spanned by $\{(E_r^{(n)}  -  E_r^{ (\bar n )}  )\}$. 
Let us denote the  vector space of ribbon graphs, which is the underlying vector space of the algebra $ \cK ( n ) $ 
by $ V^{ \Rib(n) } $. $ V^{ \Rib(n) } $ has a decomposition according to the eigenvalues of $S$ 
\bea 
V^{ \Rib(n) } = V^{ \Rib(n) }_{ S=1} \oplus V^{ \Rib(n) }_{ S=-1} 
\eea 
The $S=1$-eigenspace is the direct sum 
\bea 
V^{ \Rib(n) }_{ S=1} = V^{ \Rib ( n ) }_{ \pairs^+ }  \oplus V_{ \singlets }  
\eea
where $ V^{ \Rib ( n ) }_{ \pairs^+ }  $ is spanned by $\{ ( E_r^{(n)}  + E_r^{ (\bar n )}  ) \}$ and
$ V_{ \singlets }  $ by $\{E_r^{ (s)} \}$, 
whereas the $ S =(-1) $-eigenspace is 
\bea 
V^{ \Rib(n) }_{ S=-1} =  V^{ \Rib ( n ) }_{ \pairs^- } 
\eea
Using the Wedderburn-Artin decomposition of $ \cK ( n ) $ we also have 
\bea\label{VR1R2R3}  
V^{ \Rib(n) } = \bigoplus_{R_1 , R_2 , R_3} V^{ \Rib (n) :\; R_1 , R_2 , R_3 }  
\eea
where $ V^{ \Rib (n):\; R_1 , R_2 , R_3}   $ has dimension $ C ( R_1 , R_2 , R_3  )^2  $ and is spanned by 
the $Q^{ R_1 , R_2 , R_3  }_{ \tau_1 , \tau_2 }$ for all $\tau_1$ and $\tau_2$. 
The projection of $V^{ \Rib (n ) } $ to a fixed $ R_1 , R_2 , R_3 $ commutes with the operator $S$. 
This is evident from  Proposition \ref{propSQ}. Using this proposition, it is also obvious that the $ S=1$ subspace of $ V^{ \Rib (n) :\; R_1 , R_2 , R_3 }  $ is given by 
\bea\label{Splus1}  
 V^{ \Rib (n) :\;  R_1 , R_2 , R_3  }_{ S =1 } &=& {\rm Span } \{ Q^{ R_1 , R_2 , R_3 }_ { \tau , \tau } : 1 \le \tau \le C ( R_1 , R_2 , R_3 )   \}  \\
& \oplus&
 {\rm Span } \{ Q^{ R_1 , R_2 , R_3 }_ { \tau_1 , \tau_2  } + Q^{ R_1 , R_2 , R_3 }_ { \tau_2 , \tau_1  }  : 1 \le \tau_1 < \tau_2 \le C ( R_1 , R_2 , R_3 ) \}  
 \nonumber
\eea
and its $ S =-1$ subspace is 
\be\label{Smin1} 
V^{ \Rib (n) :\;R_1 , R_2 , R_3 } _{ S =-1 }=  {\rm Span } \{ Q^{ R_1 , R_2 , R_3 }_ { \tau_1 , \tau_2  } -  Q^{ R_1 , R_2 , R_3}_ { \tau_2 , \tau_1  }  : 1 \le \tau_1 < \tau_2 \le C ( R_1 , R_2 , R_3 ) \}
\ee
Then $ V^{ \Rib (n) :\; R_1 , R_2 , R_3 }=  V^{ \Rib (n) :\; R_1 , R_2 , R_3 } _{ S =1 }\oplus  V^{ \Rib (n) :\;R_1 , R_2 , R_3  } _{ S =-1 }$. 
Combining this with \eqref{VR1R2R3} we then have 
\be
V^{ \Rib  ( n ) } = \bigoplus_{ R_1 , R_2 , R_3 }  \left (  V^{ \Rib ( n ) : R_1 , R_2 , R_3 }_{ S =1 } \oplus V^{ \Rib ( n ) : R_1 , R_2 , R_3 }_{ S =-1 }  \right ) 
\ee
From \eqref{Smin1} we  deduce  that 
\bea 
\Dim  \left ( V^{ \Rib ( n ) : R_1 , R_2 , R_3 }_{ S =-1 }  \right )  
&=&  { C ( R_1 , R_2 , R_3) ( C ( R_1 , R_2 , R_3 ) -1 ) \over 2 }  \crcr
& = & \Dim  \left ( P^{ R_1 , R_2 , R_3 } V^{ \Rib (n ) }_{ \pairs^-} \right ) 
\eea
  with $P^{ R_1 , R_2 , R_3 }$ the projector onto $V^{ \Rib ( n ) : R_1 , R_2 , R_3 }$. Similarly from \eqref{Splus1}  we have  
\bea 
 \Dim \left ( V^{ \Rib ( n ) : R_1 , R_2 , R_3 }_{ S =+1 } \right )  & = & { C ( R_1 , R_2 , R_3 ) ( C ( R_1 , R_2 , R_3 )  + 1 ) \over 2 } \cr 
&  = &   \Dim  \left ( P^{ R_1 , R_2 , R_3 } V^{ \Rib (n ) }_{ \pairs^+}  \right )  + \Dim 
\left (  P^{ R_1 , R_2 , R_3 } V^{ \Rib (n ) }_{ \singlets} \right )   \cr 
&& 
\eea
Note that we do not have separate expressions for the two terms in the sum above 
in terms of Kronecker coefficients, since we do not expect the $P^{ R_1 , R_2 , R_3 }$ to  commute with the projection of $V^{ \Rib(n) }_{ S=1}$ into  the separate summands 
 $V^{ \Rib(n) }_{ \singlets }$ and $V^{ \Rib(n) }_{ \pairs^+ }$.

If we do the sum over $ R_1 , R_2 , R_3 $, we have 
\bea
 \Dim \left ( V^{ \Rib ( n ) }_{ S =+1 } \right ) 
 &=& 
\sum_{ R_1 , R_2 , R_3  } { C (R_1 , R_2 , R_3 ) ( C ( R_1 , R_2 , R_3 )  + 1 ) \over 2 } \crcr
&=& \Dim  \left ( V^{ \Rib ( n ) }_{ \pairs^+ }  \right )  + \Dim \left ( V_{ \singlets } \right )
\eea
and 
\bea 
\Dim \left ( V^{ \Rib ( n ) }_{ S = -1  } \right )
= \sum_{ R_1 , R_2 , R_3  } { C ( R_1 , R_2 , R_3 ) ( C (R_1 , R_2 , R_3 )  - 1 ) \over 2 } 
= \Dim  \left ( V^{ \Rib ( n ) }_{ \pairs^- }  \right )
\eea
Since 
\bea 
\Dim  \left ( V^{ \Rib ( n ) }_{ \pairs^+ }  \right ) = \Dim  \left ( V^{ \Rib ( n ) }_{ \pairs^- }  \right )
\eea
we have 
\bea\label{sumpairs+} 
\Dim  \left ( V^{ \Rib ( n ) }_{ \pairs^+ }  \right )
& =& \sum_{ R_1 , R_2 , R_3 }  { C ( R_1 , R_2 , R_3 ) ( C ( R_1 , R_2 , R_3 ) -1 ) \over 2 }  
\eea
 \bea\label{sumsinglets}  
\Dim  \left ( V^{ \Rib ( n ) }_{ \singlets }  \right ) 
&=& \sum_{ R_1 , R_2 , R_3 }   C ( R_1 , R_2 , R_3 )  
\eea 
While the sum over triples of Young diagrams with $n$ boxes of the square of Kronecker coefficients gives 
the number of ribbon graphs with $n$ edges, the sum of the Kronecker coefficients gives the number of singlet ribbon graphs. 

The sequence of sums of the Kroneckers for $ n= 1, \cdots, 10$ is 
\bea
1, 4, 11, 43, 149, 621, 2507, 11174, 49972, 237630
\eea
that coincide with the number of self-conjugate ribbon graphs for $ n =1, \cdots, 6$
\bea
1, 4, 11, 43, 149, 621. 
\eea
For $n \ge 7$ our current GAP  program for enumerating the self-conjugate ribbons is no longer very  efficient, but by the derivation we have given of \eqref{sumsinglets} these two sequences will agree.

The  projection from $ \cK (n)$  to $  V^{ \Rib (n ):R_1 , R_2 , R_3 }_{ S = 1  } $
can be done by using 
the $T_{ k}^{ (i)} $ for $ i \in \{ 1, 2, 3 \} ; k \in \{ 2 , 3 , \cdots , \widetilde k_* \}$ 
to build a rectangular matrix as in \eqref{RecInt} which projects to $ R_1  , R_2 , R_3 $ and further stacking  the matrix $ S - 1 $. This gives an integer  matrix of size $ ( 3 ( \widetilde k_*  - 1) + 1 ) |\Rib (n) | \times | \Rib (n) |   $ with null space spanning $  V^{ \Rib (n ):R_1 , R_2 , R_3 }_{ S = 1  } $.  We can also use the Hamiltonian square matrix construction \eqref{HsubRInt} along with the $ S-1$ matrix to build 
an integer  matrix of size $ 2 | \Rib (n) | \times | \Rib (n) | $ which projects to  
$  V^{ \Rib (n ):R_1 , R_2 , R_3 }_{ S = 1  } $. By replacing $ ( S - 1)  $ with $ ( S+1) $ in these constructions we can obtain the subspace 
$  V^{ \Rib (n ): R_1 , R_2 , R_3 }_{ S = -1  } $  of $ \cK(n)$ as null spaces of integer rectangular or square matrices.

 As in section \ref{nulllattices}  the HNF construction of $  V^{ \Rib (n ):R_1 , R_2 , R_3 }_{ S = \pm 1  } $ determines sub-lattices of $ \mZ^{ | \Rib (n) |  } $. 
  Thus, on one hand, we have lattice constructions for 
 \bea 
{ C ( R_1 , R_2 , R_3 ) ( C ( R_1 , R_2 , R_3 )  -1) \over 2 } 
 \eea
as the dimension of  $ V^{ \Rib (n ): R_1 , R_2 , R_3 }_{ \pairs^- }  $ and, on the other, we also have a construction of 
\bea
  { C( R_1 , R_2 , R_3 )  ( C ( R_1 , R_2 , R_3 )  +1 ) \over 2 } 
 \eea 
 as the dimension of $ V^{ \Rib (n ): R_1 , R_2 , R_3 }_{ S =1  }   $. The difference of these is the number $ C (R_1 , R_2 , R_3 )$. By choosing an injection between the smaller sub-lattice and the larger sub-lattice, we can get a constructive interpretation of $ C ( R_1 , R_2 , R_3 )$. It will be interesting to investigate if there is a canonical choice of such an injection. 
 
We summarise the outcome of the above discussion as a theorem
 \begin{theorem}
 \label{addLatt}
For every triple of Young diagrams $(R_1 , R_2 , R_3 ) $ with $n$ boxes, there are three
constructible sub-lattices of $\mZ^{|\Rib(n)|}$  of respective dimensions 
${ C ( R_1 , R_2 , R_3 ) ( C ( R_1 , R_2 , R_3 )  +1) /2 } $, 
${ C ( R_1 , R_2 , R_3 ) ( C ( R_1 , R_2 , R_3 )  -1) / 2 } $, 
and $C ( R_1 , R_2 , R_3 )$. 
 \end{theorem}

As an illustration, there are two interesting cases at $ n=5$ with $C(R_1 , R_2 , R_3) =  2$, 
$(R_1 , R_2 , R_3) = ([3, 2], [3, 1, 1], $ $[3, 1, 1])$ and $
(R_1 , R_2 , R_3) = ( [3, 1, 1], [3, 1, 1],  [2, 2, 1])$,
their permutations. 
We have $( [3, 2], $ $  [3, 1, 1],    [3, 1, 1])$: 
\bea
&& 
\Dim  \left ( V^{ \Rib ( n ) : R_1 , R_2 , R_3 }  \right ) = 4 \cr 
&& 
\Dim  \left ( V^{\Rib( n ) : R_1 , R_2 , R_3 }_{ S =+1 }  \right ) = 3\cr 
&&
\Dim  \left ( V^{ \Rib ( n ) : R_1 , R_2 , R_3 }_{ S =-1 }  \right ) = 1
\eea
The same equations  hold for $([3, 1, 1], [3, 1, 1], [2, 2, 1])$.

\section{Conclusions} 
\label{ccl}

We give a summary of our main results and outline important  directions 
for future research. Section \ref{qmem} uses the link between bi-partite ribbon graphs and Belyi maps. Section \ref{Qalg} outlines quantum algorithms motivated by links between the algebra $\cK (n)$ and Kronecker coefficients. Section \ref{gens} 
describes physically motivated  generalizations of the present work based on algebras related to $\cK(n)$ which also have interesting geometric  interpretations.

\subsection{Summary}
In this paper we have developed quantum mechanics on a class of state spaces which are also algebras
(in the present case the algebras $ \cK (n)$), and have a distinguished geometrical/combinatorial  basis associated with combinatorial objects (in the case at hand bipartite ribbon graphs). 
The combinatorial objects have a description in terms of equivalence classes defined using permutations and the algebra can be realised as a subspace of a tensor product of group algebras (in this case $ \mC ( S_n) \otimes \mC ( S_n)$ - and there is a gauge equivalent formulation in terms of $ \mC(S_n)^{ \otimes 3}  $ as explained in \cite{PCAKron}.  By exploiting the algebra structure, we are able to relate the eigenvalues of Hermitian Hamiltonians on these state spaces to characters of symmetric groups, and the multiplicities to group theoretic multiplicities (in this case Kronecker coefficients). The integrality structure of the algebra when expressed in terms of the geometrical basis means that the solving the Hamiltonians  is a problem that can draw on techniques from the mathematics of integer matrices and lattices.  It follows that the square of Kronecker coefficients ($ C^2$) can be realised as the dimension of a sub-lattice in the lattice generated by ribbon graphs.  The algebra has an involution which is inherited from the inversion of permutations and commutes with the Hamiltonians considered here.  The involution is used to define sub-lattices of dimensions $ C ( C +1) /2 $ and   $ C ( C -1)/2$. Choosing an injection of the  set of   sub-lattice  basis vectors of the smaller sub-lattice, selected by the HNF construction of integer matrices we  built, into the set of sub-lattice basis vectors of the larger sub-lattice also fixed by the HNF construction, leads to a sub-lattice of dimension $C$.

\subsection{Belyi maps and quantum membrane interpretation of quantum mechanics on $ \cK ( n ) $} 
\label{qmem}

 Bipartite ribbon graphs have a rich geometrical structure related to Belyi maps and number theory. In this section, we use this connection to Belyi maps  to give an interpretation of quantum mechanical evolution in the ribbon graph quantum mechanics  in terms of membranes : the covering surfaces arising in Belyi maps appear at fixed time and can be interpreted as string worldsheets in topological strings - the quantum mechanical time is   an additional coordinate which can be viewed as part of a membrane worldvolume.

 Bipartite ribbon graphs  with $n$ edges are in 1-1 correspondence with holomorphic maps  $f$ (branched covers)  from a Riemann surface $ \Sigma_g $ to two-dimensional  Riemann sphere  with degree $n$  and  3 branch points : 
 \bea 
 f : \Sigma_g \rightarrow \mC \mP^1 
 \eea
 These branch points  can be taken to be $\{ 0 , 1, \infty \}$ (see Chapter 2 of  \cite{LandoZvonkin}). If we label the inverse images of  a generic point on the sphere as $ \{ 1, 2, \cdots , n \} $, then the branching at the three branch points are described by the three permutations $ \s_1 , \s_2 , \s_3 = ( \s_1 \s_2 )^{ -1} $. The genus $g$  of the covering surface is given by the Riemann-Hurwitz formula 
 \bea 
 ( 2g -2 ) = n - C_{ \s_1 } - C_{ \s_2 } - C_{ \s_3 } 
 \eea 
 where $ C_{ \sigma }$ is  the number of cycles of the permutation $ \sigma $.
 Branched covers  with exactly three branch points, also called Belyi maps,  have the property  that the covering 
surface $ \Sigma_g  $ as well as the covering map can be defined in terms of equations with coefficients which are algebraic numbers, complex numbers which are solutions of polynomials with integer coefficients \cite{Belyi}. Conversely any  such algebraic surface can be realised as a branched cover of  the sphere with 3 branch points.  The inverse image of the interval $ [ 0 , 1 ] $ on the Riemann sphere defines a graph embedded in the surface $ \Sigma_g$, also called a map. These maps were called Dessins d'Enfants by Grothendieck who proposed their combinatorial study as a tool to understand representations of the absolute Galois Group, an object of fundamental importance in number theory \cite{GrothendieckDessins}.  A survey of mathematical work in this area is in \cite{schneps}.  It is interesting that the conjugation operation $S$ which has played an important role in this paper (section \ref{app:conjugate}) has previously appeared in the study of ``operations on maps''. The self-conjugate graphs correspond to reflexible Belyi maps in the terminology of \cite{Jones}. The number of distinct  terms in $E_r $, when expanded as a sum in $ \mC ( S_n ) \otimes \mC ( S_n )$  is $n!$ divided by the order of the  automorphism group of the Belyi map, i.e. the number of holomorphic invertible maps $ \phi : \Sigma_g \rightarrow \Sigma_g $  obeying $ f \circ \phi = f$.

Each ribbon graph defines an element $E_r  \in \cK (n ) = \mC ( S_n ) \otimes \mC ( S_n ) $. The quantum mechanical evolution of such a state produces 
\bea 
E_r ( t ) = e^{ - i \cH t } E_r 
\eea
At generic $t$,  this is  a superposition of different basis elements $E_s \in \cK (n ) $. 
Such a superposition determines a linear combination of Belyi curves and Belyi maps. The evolution over all $t \ge 0$ determines a quantum membrane worldvolume mapping to $ S^2 \times \mR^+$ which restricts to a single Belyi map at $ t=0$ and subsequent periodic intervals, but is generically a superposition of covering surfaces mapping to $ S^2$. Belyi maps have been linked to matrix models and topological strings \cite{dMKS,Gopakumar}. The discrete spin states of a particle such as an electron (two-dimensional state space) which are the subject of the quantum mechanics of spin are being generalized in the quantum mechanics of ribbon graphs to discrete states of a two-dimensional surface. These discrete states have the rich structure of an algebra, at each $n$ the algebra $ \cK ( n )$, and they have the rich geometrical structure in terms of an  algebraic number realization as Belyi curves. The additional time direction of the quantum mechanics forms the $ 2+1$ dimensional worldvolume of a membrane generalizing the $0+1$ dimensional worldline of a particle.  It would be interesting to develop  the realisations of such quantum mechanical evolutions using worldvolume membrane actions (see for example \cite{Duff88,DHN,GKT,BTZ,Horava2008}) in a topological and non-relativistic limit.

The link to tensor models also 
leads to other quantum mechanical  systems which can be viewed as generalizations of the  systems 
discussed here. Our quantum systems have been discussed in terms of state spaces $ \cK ( n )$ for a fixed $n$. We can generalize, somewhat trivially, to the infinite dimensional state space 
\bea\label{cKinfinity}  
\cK ( \infty ) = \bigoplus_{ n = 0 } \cK ( n ) 
\eea
where $ \cK ( 0 ) $ is defined to be the one-dimensional vector space $ \mC $. We can  
use a  Hamiltonian of the form 
 $ \cH = \sum_{ n =0}^{ \infty } \cH^{ (n) } $, were $ \cH^{ (n)} $ is a Hamiltonian of the form we discussed at fixed $n$. This Hamiltonians generates time evolutions which mix ribbon graphs with a given number of edges, or Belyi maps with a given degree. Generic Hamiltonians for tensor models, when expressed in  terms of $ \cK ( \infty )$ would be expected to generate interactions which mix different values of $n$. Some recent literature solving quantum mechanical Hamiltonians for tensor models is in  \cite{KK1802,KMPT1802}.

\subsection{Quantum computing and Kronecker  coefficients} 
\label{Qalg}
In this section we describe how   quantum mechanical systems 
on ribbon graphs, hypothetically engineered  in the laboratory, can be used to 
detect the non-vanishing of Kronecker coefficients.
 It has been shown  \cite{IkMuWa-VanKron} 
that the question of deciding the positivity of  the Kronecker coefficient for a triple of Young diagrams is NP-hard. 
The question of whether a quantum computer can outperform classical computers on some chosen task of interest is the problem of quantum supremacy. For recent progress, on specific tasks of generating random number sequences, see \cite{Qsup}. A quantum mechanical system of ribbon graphs can conceivably be engineered directly  by identifying physical objects with the properties of ribbon  graphs, or perhaps more realistically for  the near future, we may consider  the idea of quantum simulation where an  experimentally  controllable quantum system, such as superconducting circuits, is used to simulate another quantum system of interest. A recent review on quantum simulators is \cite{Paraoanu}.  A physical or simulated  quantum mechanical system of ribbon graphs would allow, using the connections we have developed here between ribbon graphs and Kronecker coefficients for any specified triple of Young diagrams, the possibility of detecting non-vanishing Kronecker coefficients by observing the time evolution of ribbon graphs. 
We refer to \cite{BenGeloun:2020yau} for  further details and  scenarios
inspecting this questions.

An important role is played in our construction of integer matrices with  integer null spaces of 
ribbon graph vectors by the number $k_* (n)$ defined in section \ref{sec:CentreKnReconn}. To get precise estimates of the computational complexity of our algorithms viewed as a way to calculate Kronecker coefficients at large $n$, it is desirable to find estimates of the growth of this number with $n$ as $n $ tends to infinity. As discussed in 
\cite{KR1911,IILoss} this asymptotic behaviour is also relevant to understanding information loss in toy models of black holes arising in AdS/CFT.

\subsection{ Generalizations }\label{gens} 

Permutation equivalence classes provide a general approach to  the counting of a variety of combinatorial objects of interest in theoretical physics, e.g. Feynman diagrams \cite{FeynString,DoubCos,GLS1709,CR1804}, gauge invariants in matrix and tensor models (see a review in \cite{RamgoComb}), and frequently  these equivalence classes have an algebra  structure. The development of quantum mechanical systems where these combinatorial objects become quantum states, their associated permutation algebras are used to express quantum mechanical problems in terms of representation theoretic objects is a promising avenue for further fruitful investigations.  We expect the integrality structure of the algebras, when expressed in terms of the geometric basis, to be generic. This will allow a realisation of the representation theoretic quantities in terms of sub-lattices of the lattice generated by the combinatorial objects. An example of such combinatorial algebra studied in detail in \cite{PCAMultMat} is associated with colored necklaces having $m$ beads of one colour and $n$ beads of another color.

Another interesting direction for research  is the generalization of the
present study to real tensor invariants in particular the 
$O(N)^3$ invariants \cite{Carrozza:2015adg,Carrozza:2018ewt}. 
In this case, the ribbon graphs are not bipartite and
their counting gives the sum of Kronecker coefficients with
Young diagram restricted to even partitions \cite{Avohou:2019qrl, BenGeloun:2020lfe}.

Although our investigations of $ \cK ( n )$ were motivated by the study of correlators of tensor models in \cite{Sanjo,PCAKron} the 3-index tensor variables have not played a direct role in this paper. The space of all gauge invariants constructed from complex tensors $ \Phi_{ ijk} , \bar \Phi^{ ijk} $ is isomorphic as a vector space to $ \cK ( \infty ) $ defined in \ref{cKinfinity}. On this space the operators $T_{k}^{(i)} $ we used in this paper should be expressible in terms of differential operators. The map between permutation algebra elements analogous to $ T_k^{ (i)} $ and differential operators was given in the context of multi-matrix invariants in \cite{EHS}. The operators $T_k^{(i)} $ are also related to the cut-and-join operators considered in tensor model context in \cite{Diaz2009,Itoyama:2017wjb}.  An interesting problem is to use the connection between differential operators and the reconnection operators $ T_k^{(i)} $ to develop Hamiltonian and Lagrangian formulations of the quantum mechanical systems discussed here. The link between such Lagrangian formulations in terms of tensor variables and  possible  membrane world-volume Lagrangians connecting with the membrane interpretation based on Belyi maps (discussed in section \ref{qmem}) would be illuminating.

\begin{center} 
{\bf Acknowledgements}
\end{center} 
SR is supported by the STFC consolidated grant ST/P000754/1 `` String Theory, Gauge Theory \& Duality” and  a Visiting Professorship at the University of the Witwatersrand, funded by a Simons Foundation grant (509116)  awarded to the Mandelstam Institute for Theoretical Physics. We thank  Robert de Mello Koch, Igor Frenkel, Amihay Hanany for useful discussions  on the subject of this paper.

\section*{ Appendix}

\appendix

\renewcommand{\theequation}{\Alph{section}.\arabic{equation}}
\setcounter{equation}{0}

\section{Reconnection operators $T^{(i)}_2$ as matrices at $ n=3 $ } 
\label{app:examn3}

In this appendix, we work at $ n=3$ and give the construction of  the matrices for reconnection operators $T_2^{(i)} $, $ i = 1,2,3$. In this case there are $11$ bipartite 
ribbon graphs. The  matrix elements $ (\cM_2^{(i)})_r^s  $ of the  reconnection operators in the geometric ribbon graph basis  are $ 11 \times 11$ integer  matrices and can be used to determine the Young diagram triples with non-vanishing Kronecker coefficient. Each of these non-vanishing Kronecker coefficients is $1$ and we construct the corresponding integer vector  in the ribbon graph lattice which generates a one-dimensional sub-lattice.

\subsection{$T_2^{(i)}$ matrices}
\label{appsub:reconn}  

Note that we have provided a code to produce all entries $ (\cM^{(i)}_k )^{ s }_{ r }$, 
see Code2 in appendix \ref{app:gap}. 
For $ T_2^{(1)} , T_2^{(2) }$ and $T_2^{(3)} $,  we have the following non-negative  integer  matrices,
respectively, 
\bea 
\cM^{(1)}_2  = \left (
  \begin {array} {ccccccccccc} 0 & 0 & 0 & {\cred  1} & 0 & 0 & 0 & 0 & 0 & 0 \
& 0 \\ 0 & 0 & 0 & 0 &{\cred  1 }&{\cred  1} & 0 & 0 & 0 & 0 & 0 \\ 0 & 0 & 0 & 0 & 0 \
& 0 & {\cred  1} & 0 & 0 & 0 & 0 \\ {\cred 3 }& 0 & 0 & 0 & 0 & 0 & 0 & {\cred 3 }& 0 & 0 & 0 \
\\ 0 & {\cred 1 }& 0 & 0 & 0 & 0 & 0 & 0 & {\cred 1} & 0 & 0 \\ 0 & {\cred 2} & 0 & 0 & 0 & 0 \
& 0 & 0 &{\cred  2} & 0 & 0 \\ 0 & 0 & {\cred 3} & 0 & 0 & 0 & 0 & 0 & 0 & {\cred 3} & {\cred 3} \\ 0 \
& 0 & 0 & {\cred 2} & 0 & 0 & 0 & 0 & 0 & 0 & 0 \\ 0 & 0 & 0 & 0 & {\cred  2} & {\cred 2} & 0 \
& 0 & 0 & 0 & 0 \\ 0 & 0 & 0 & 0 & 0 & 0 & {\cred 1} & 0 & 0 & 0 & 0 \\ 0 & 0 \
& 0 & 0 & 0 & 0 & {\cred  1} & 0 & 0 & 0 & 0 \\\end {array} \right)
\eea 
\bea 
\cM^{(2)}_2= 
\left (
  \begin {array} {ccccccccccc} 0 &{\cred   1} & 0 & 0 & 0 & 0 & 0 & 0 & 0 & 0 \
& 0 \\ {\cred  3} & 0 & {\cred 3} & 0 & 0 & 0 & 0 & 0 & 0 & 0 & 0 \\ 0 & {\cred  2} & 0 & 0 & 0 \
& 0 & 0 & 0 & 0 & 0 & 0 \\ 0 & 0 & 0 & 0 & {\cred  1} & {\cred 1} & 0 & 0 & 0 & 0 & 0 \
\\ 0 & 0 & 0 & {\cred 1} & 0 & 0 &{\cred  1} & 0 & 0 & 0 & 0 \\ 0 & 0 & 0 & {\cred 2} & 0 & 0 \
&{\cred  2} & 0 & 0 & 0 & 0 \\ 0 & 0 & 0 & 0 & {\cred 2} & {\cred 2} & 0 & 0 & 0 & 0 & 0 \\ 0 \
& 0 & 0 & 0 & 0 & 0 & 0 & 0 & {\cred 1} & 0 & 0 \\ 0 & 0 & 0 & 0 & 0 & 0 & 0 \
&{\cred  3} & 0 & {\cred 3} &{\cred  3} \\ 0 & 0 & 0 & 0 & 0 & 0 & 0 & 0 & {\cred 1} & 0 & 0 \\ 0 & 0 \
& 0 & 0 & 0 & 0 & 0 & 0 &{\cred  1}  & 0 & 0 \\\end {array} \right)
\eea
\bea 
\cM^{(3)}_2= 
\left (
  \begin {array} {ccccccccccc} 0 & 0 & 0 & 0 &  {\cred 1} & 0 & 0 & 0 & 0 & 0 \
& 0 \\ 0 & 0 & 0 &  {\cred 1} & 0 & 0 &  {\cred 1}  & 0 & 0 & 0 & 0 \\ 0 & 0 & 0 & 0 & 0 \
&   {\cred 1} & 0 & 0 & 0 & 0 & 0 \\ 0 &   {\cred 1} & 0 & 0 & 0 & 0 & 0 & 0 &  {\cred 1} & 0 & 0 \
\\ {\cred  3}  & 0 & 0 & 0 & 0 & 0 & 0 & 0 & 0 & {\cred  3}   & 0 \\ 0 & 0 & {\cred  3}  & 0 & 0 & 0 \
& 0 & {\cred  3}  & 0 & 0 & {\cred  3}   \\ 0 &  {\cred 2} & 0 & 0 & 0 & 0 & 0 & 0 &  {\cred 2} & 0 & 0 \\ 0 \
& 0 & 0 & 0 & 0 &  {\cred 1} & 0 & 0 & 0 & 0 & 0 \\ 0 & 0 & 0 &  {\cred 2} & 0 & 0 &  {\cred 2} \
& 0 & 0 & 0 & 0 \\ 0 & 0 & 0 & 0 &  {\cred 2} & 0 & 0 & 0 & 0 & 0 & 0 \\ 0 & 0 \
& 0 & 0 & 0 &  {\cred 1}  & 0 & 0 & 0 & 0 & 0 \\\end {array} \right)
\eea

\

\subsection{Nullspace at $n=3$}
\label{appsub:nullspn3} 

In this section we give,  for $n=3$, the common nullspace of  all operators $(T^{(i)}_2-\chi_{R_i}(T_2)/d(R_i))$, for all $i=1,2,3$ at fixed  $R_i\vdash n$. This is the rectangular construction of section \ref{sec:reconn}.

The vector $v$ generating the null space for each 
 triple of Young diagram $(R_1,R_2,R_3)$ with non-vanishing Kronecker coefficient is : 
\begin{center}
\begin{tabular}{ c|c|c|} 
& $(R_1,R_2,R_3)$ & $v$ \\ 
\hline 
1 & 
( [1,1,1],[1,1,1], [3] ) & ( 1, -3, 2, -3, 3, 6, -6, 2, -6, 2, 2 )  \\
2 &
 ( [1,1,1],[2,1],[2,1] )   &    (-2, 0, 2, 6, 0, 0, -6, -4, 0, 2, 2 )  \\
3 &
( [1,1,1],[3],[1,1,1] )  &   ( 1, 3, 2, -3, -3, -6, -6, 2, 6, 2, 2 )  \\
4&
 ( [2,1],[1,1,1],[2,1] )   &   ( -2, 6, -4, 0, 0, 0, 0, 2, -6, 2, 2 ) \\
5 & 
 ( [2,1],[2,1],[1,1,1] )  & ( -2, 0, 2, 0, 6, -6, 0, 2, 0, -4, 2 )   \\
6 & 
( [2,1],[2,1 ],[2,1] )   & ( 1, 0, -1, 0, 0, 0, 0, -1, 0, -1, 2 ) \\
7 &
 ( [2,1],[2,1 ],[3] )  &   ( -2, 0, 2, 0, -6, 6, 0, 2, 0, -4, 2 ) \\
8 &
( [2,1],[3],[2,1] ) &  ( -2, -6, -4, 0, 0, 0, 0, 2, 6, 2, 2 )  \\
9 &
 ( [3],[1,1,1 ],[1,1,1] ) &  ( 1, -3, 2, 3, -3, -6, 6, 2, -6, 2, 2 )  \\
10 & 
( [3],[2,1],[2,1] )   &  ( -2, 0, 2, -6, 0, 0, 6, -4, 0, 2, 2 )  \\
11 & 
 ( [3],[3],[3] )  &   ( 1, 3, 2, 3, 3, 6, 6, 2, 6, 2, 2 ) \\ 
\hline
\end{tabular}
\end{center}
One recognizes that  the last vector $ (  1, 3, 2, 3, 3, 6, 6, 2, 6, 2, 2)  $ 
is the vector $\sum_{ r }  |\Orb( r) | E_r $, where $ |\Orb(r) |$
is  the orbit size of the $r$'th  ribbon graph equivalence class.   
All $C(R_1,R_2,R_3)=1$ for the triples $(R_1,R_2,R_3)$ given above.  For a each 
triple, we actually see that the null space is one dimensional  at 
$n=3$. For $n >  4$, there are $C(R_1,R_2,R_3)>1$ and therefore the nullspace
become of dimension higher than 1 as expected.

\

\section{Geometric and Fourier basis} 
\label{app:geombasis}

This appendix elaborates  on the transformation between the Fourier basis and the geometric basis of $\cK(n)$. 
We check  that the Fourier base $\{Q^{ R_1 , R_2 , R_3 }_{\tau,\tau'}\}$ expands in terms of ribbon graph base $\{E_r\}$ and vice-versa. 
We also give a proof of Proposition \ref{propTaQo}. 

\subsection{Change of basis}

\noindent{\bf Ribbon graph expansion of $Q^{ R_1 , R_2 , R_3 }_{\tau,\tau'}$.} 
Start with the base $Q^{ R_1 , R_2 , R_3 }_{\tau,\tau'}$ \eqref{qbasis}
that we re-expand using the orbit decomposition in the same way as
in \eqref{Err} as: 
 \bea
 &&
  Q^{R_1,R_2,R_3}_{\tau_1,\tau_2}  = \frac{1}{n!}
\kappa_{R,S} \crcr
&&  \times 
\sum_{r}
\sum_{a \in \Orb(r)}
\sum_{i_1,i_2,i_3, j_1,j_2}
C^{R_1,R_2 ; R_3 , \tau_1  }_{ i_1 , i_2 ; i_3 } C^{R_1,R_2 ; R_3 , \tau_2  }_{ j_1 , j_2 ; i_3 } 
 D^{ R_1 }_{ i_1 j_1} ( \sigma^{(r)}_1(a)  ) D^{R_2}_{ i_2 j_2 } ( \sigma^{(r)}_2(a) ) \, 
  \sigma^{(r)}_1 (a)\otimes  \sigma^{(r)}_2(a)
 \crcr
 && 
\eea
where $(  \sigma^{(r)}_1 (a),  \sigma^{(r)}_2(a))$ is a representative
pair in the orbit $\Orb(r)$ that defines the ribbon graph $r$. 

Therefore, we can reorganize the sum and collect for each ribbon graph base element, its
coefficient in the above expansion
\bea
Q^{R_1,R_2,R_3}_{\tau_1,\tau_2} = \kappa_{R,S} 
   \sum_{ r} 
      \Big[
\sum_{i_1,i_2,i_3, j_1,j_2}
C^{R_1,R_2 ; R_3 , \tau_1  }_{ i_1 , i_2 ; i_3 } C^{R_1,R_2 ; R_3 , \tau_2  }_{ j_1 , j_2 ; i_3 }
 D^{R_1}_{ i_1j_1} (\s^{(r)}_1)  D^{R_2}_{ i_2j_2} (\s^{(r)}_2)  \Big]  |\Orb(r)| E_r 
 \crcr
 && 
\eea
where  $D^{ R_1 }_{ i_1 j_1} ( \sigma^{(r)}_1(a)  ) D^{R_2}_{ i_2 j_2 } ( \sigma^{(r)}_2(a) )$ has been replaced with  $ D^{ R_1 }_{ i_1 j_1} ( \sigma^{(r)}_1  ) D^{R_2}_{ i_2 j_2 } ( \sigma^{(r)}_2 )$ where $( \sigma^{(r)}_1 , \sigma^{(r)}_2 )$ is any 
representative pair in $\Orb(r)$. This can be done because the coefficient in square brackets is invariant under simultaneous conjugation of $ \sigma_1 , \sigma_2$ by a permutation  $ \gamma$. 

\

\noindent{\bf Fourier expansion of $E_r$.} 
Consider the following expansion of some $E_r$ \eqref{classr}
in terms of the basis $Q^{R_1, R_2, R_3}_{ \tau_1, \tau_2} $: 
\bea\label{ErtoQ}
 E_r= 
\sum_{R_i, \tau_i} {\bf C}_{R_1,R_2, R_3} (\s_1^{(r)} , \s_2^{(r)} ) \, 
Q^{R_1, R_2, R_3}_{ \tau_1, \tau_2}  \,, 
\eea
where $ = ( \s_1^{(r)} , \s_2^{(r)})  $ form a permutation pair in the orbit $r$, and  the coefficient ${\bf C}_{R_1,R_2, R_3}  (\s_1^{(r)} , \s_2^{(r)} )$
is to  be determined. 
Use the orthogonality relation \eqref{orhoqq0} and evaluate ${\bf C}_{R_1,R_2, R_3} (\s_1^{(r)} , \s_2^{(r)} ) $: 
\bea
\bdel_2(E_r  ,Q^{R_1, R_2, R_3}_{ \tau_1, \tau_2}  )
&=& \sum_{R'_i, \tau'_i} {\bf C}_{R'_1,R'_2, R'_3} (\s_1^{(r)} , \s_2^{(r)} ) \, 
\bdel_2( Q^{R'_1, R'_2, R'_3}_{ \tau'_1, \tau'_2} , Q^{R_1, R_2, R_3}_{ \tau_1, \tau_2} ) \crcr
& =&  {\bf C}_{R_1,R_2, R_3} (\s_1^{(r)} , \s_2^{(r)} ) \kappa_{R_1,R_2}d(R_3)
\eea
On the other hand, using \eqref{qbasis}, we also compute 
\bea\label{qbasis2}
&&
\bdel_2( E_r  , Q^{R_1, R_2, R_3}_{ \tau_1, \tau_2}) = \crcr
&&
\kappa_{R_1,R_2}
\sum_{\g_1, \g_2 \in S_n}
\sum_{i_1,i_2,i_3, j_1,j_2}
C^{R_1 , R_2 ; R_3 , \tau_1  }_{ i_1 , i_2 ; i_3 } C^{R_1 , R_2 ; R_3 , \tau_2  }_{ j_1 , j_2 ; i_3 } 
 D^{ R_1 }_{ i_1 j_1} ( \g_1  ) D^{R_2 }_{ i_2 j_2 } ( \g_2 )  \,  \bdel_2( \s_1^{(r)}   \otimes \s_2^{(r)}  , \g_1 \otimes \g_2 ) \crcr
 && = 
\kappa_{R_1,R_2}
\sum_{i_1,i_2,i_3, j_1,j_2}
C^{R_1 , R_2 ; R_3 , \tau_1  }_{ i_1 , i_2 ; i_3 } C^{R_1 , R_2 ; R_3 , \tau_2  }_{ j_1 , j_2 ; i_3 } 
 D^{ R_1 }_{ i_1 j_1} ( \s_1^{(r)}   ) D^{R_2 }_{ i_2 j_2 } ( \s_2^{(r)}  )  \,
\eea
from which we conclude 
\bea
   {\bf C}_{R_1,R_2, R_3} (\s_1^{(r)} , \s_2^{(r)} )
    = \frac{1}{d(R_3)}
 \sum_{i_l,j_l} D^{R_1}_{i_1j_1}(\s_1^{(r)} )D^{R_2}_{i_2j_2}(\s_2^{(r)} )
  C^{R_1,R_2 ; R_3 , \tau_1  }_{ i_1 , i_2 ; i_3 } 
 C^{R_1,R_2 ; R_3 , \tau_2  }_{ j_1 , j_2 ; i_3 }  \,. 
\eea
One checks that $   {\bf C}_{R_1,R_2, R_3} (\s_1, \s_2)$
is invariant under diagonal conjugation 
$$  {\bf C}_{R_1,R_2, R_3} (\g \s_1 \g^{-1}, \g \s_2 \g^{-1})
 =   {\bf C}_{R_1,R_2, R_3} (\s_1, \s_2)$$ (this can be shown using the so-called 
$DDC=CD$ relation and the orthogonality of representation matrices, 
see appendix A.1 and A.2 of \cite{PCAKron}). 

An immediate consequence of these formulae is that  
\bea
Q^{R_1,R_2,R_3}_{\tau_1,\tau_2} = \kappa_{R,S} d(R_3)
   \sum_{ r}  {\bf C}_{R_1,R_2, R_3} (\s_1^{(r)}, \s_2^{(r)}) 
      |\Orb(r)| E_r \; . 
\eea
Thus the orthogonal Fourier basis elements $Q^{R_1,R_2,R_3}_{\tau_1,\tau_2}$ 
and expressible  in terms of $E_r$ and vice-versa.

\subsection{   Fourier basis as eigenvectors of   reconnection operators $T_k^{(i)}$.} 
\label{sec:QasEigTki} 

 To prove Proposition \ref{propTaQo}, we start   with some preliminary  observations about the group algebra of the symmetric group.

Let $\mC(S_n)$ the group algebra of $S_n$. An inner product on the group 
algebra is defined by specifying on basis elements $ \sigma , \tau \in S_n$, $ \delta ( \sigma ; \tau ) = \delta ( \sigma \tau^{-1} ) $. Consider the Fourier basis set $ Q^R_{ ij}    \in \mC ( S_n)$ 
\be\label{qbbasis}
Q^R_{ ij}   =   { \kappa_R \over n! } \sum_{ \sigma \in S_n } D^R_{ ij} ( \sigma ) \sigma  \;, 
\qquad 
\kappa_R^2 = n! d(R)
\ee
where $D^R_{ij} ( \sigma )$ are matrix elements of $\sigma $ in an orthonormal basis for 
the irreducible representation $R$. $\kappa_R$ is fixed such that  $\delta(Q^{ R}_{ i j };Q^{ R'}_{ i' j' }) = \delta_{RR'}\delta_{ii'}\delta_{jj'}$ making $\{Q^R_{ ij}  \}$  an orthonormal basis of $ \mC ( S_n )$.  Furthermore, these elements also obey
\be \label{tauq}
\tau \,  Q^R_{ ij} = 
\sum_{l} D^R_{ li} ( \tau ) \,  Q^R_{ l j} \,, \qquad  
 Q^R_{ ij} \, \tau  =    \sum_{l} Q^R_{ i  l} ~ D^R_{ j l } ( \tau )  
\ee
The following statement holds: 
\begin{lemma}
\label{lemTQ}
\be
T_k  Q^R_{ ij} = { \chi_R(T_k) \over d(R)  }   \, Q^R_{ i j}
\ee
\end{lemma}
\begin{proof}   We let act $T_k$ on $ Q^R_{ ij} $ 
and using \eqref{tauq} we write: 
\bea\label{TQ}
T_k  Q^R_{ ij} = \sum_{\s \in \cC_k} \s  Q^R_{ ij}   = 
\sum_{l} \big( \sum_{\s \in \cC_k}  D^R_{ li} ( \s ) \big) \,  Q^R_{ l j}
\eea
The sum $\sum_{\s \in \cC_k}  D^R_{ li} ( \s ) $
may be also written $ D^R_{ li} ( T_k)$. As $T_k$ is central and commute with 
all elements, $D^R_{ li} ( T_k) = \alpha \delta_{li}$ a constant diagonal matrix by Schur lemma. We also have 
$\sum_{i }D^R_{ ii} ( T_k) =  \chi_R(T_k) =  \alpha d(R)$, that yields
$\alpha =   \chi_R(T_k)/d(R)$.  Then back to our previous expression \eqref{TQ}
\bea
T_k  Q^R_{ ij} = 
\sum_{l} \big( { \chi_R(T_k) \over d(R)  } \delta_{ li} \big) \,  Q^R_{ l j}
 ={ \chi_R(T_k) \over d(R)  }   \, Q^R_{ i j}
\eea
which proves the lemma. 
 \end{proof}

Thus $ Q^R_{ ij} $ is an eigenvector of $T_k  $ with eigenvalue $ \chi_R(T_k)/d(R)$. 
 We get back to our main concern, namely the action of $T_k^{(i)}$ on 
 $Q^{R_1, R_2, R_3}_{\tau_1 , \tau_2}$. 
 
 \

\noindent{\bf Proof of Proposition \ref{propTaQo}.}
We want to prove that 
 $ Q^{R_1, R_2, R_3}_{\tau_1 , \tau_2}$ define  eigenvectors of $T_k^{(i)} $. 
 $Q^{R_1, R_2, R_3}_{ \tau_1 , \tau_2 }  $  \eqref{qbasis} can be written in terms of 
 the Fourier basis set for $ \mC ( S_n)$ as  
 \be
Q^{R_1, R_2, R_3}_{ \tau_1 , \tau_2 }    = 
\kappa'_{R_1,R_2}
\sum_{i_l, j_l}
C^{R_1, R_2; R_3 , \tau_1  }_{ i_1 , i_2 ; i_3 } C^{R_1, R_2; R_3 , \tau_2  }_{ j_1 , j_2 ; i_3 } 
 Q^{ R_1 }_{ i_1 j_1} \otimes  Q^{ R_2 }_{ i_2 j_2} \,, 
 \ee
 where $\kappa'_{R,S}$  is a normalization factor
 \bea 
 \kappa'_{ R , S } = \sqrt { d(R) d(S) \over n!  } \, . 
 \eea 
 Then, it becomes obvious that, by Lemma \ref{lemTQ},  \eqref{t1Qo}
 and  \eqref{t2Qo} hold as $T_k^{(1)} $ and $T_k^{(2)}$ are 
 defined by the actions on $T_k$ of the left or right factors of $ \mC ( S_n) \otimes \mC ( S_n)$.  The third relation requires a bit more work. 
 Use \eqref{tauq} to rewrite: 
 \bea
  && T_k^{(3)} Q^{R_1, R_2, R_3}_{ \tau_1 , \tau_2 }   =  
  \kappa'_{R_1, R_2}
     \sum_{\s \in \cC_k } 
\sum_{i_l, j_l}
C^{R_1, R_2;  R_3, \tau_1  }_{ i_1 , i_2 ; i_3 } C^{R_1, R_2;  R_3, \tau_2  }_{ j_1 , j_2 ; i_3 } 
 \s   Q^{ R_1 }_{ i_1 j_1} \otimes  \s  Q^{ R_2 }_{ i_2 j_2}  \crcr
 & & =  
  \kappa'_{R_1, R_2}
     \sum_{\s \in \cC_k } 
\sum_{\s_1, \s_2 \in S_n}
\sum_{i_l, j_l}
C^{R_1, R_2;  R_3, \tau_1  }_{ i_1 , i_2 ; i_3 } C^{R_1, R_2;  R_3,\tau_2  }_{ j_1 , j_2 ; i_3 } 
 \sum_{m_1,m_2 } D^{R_1}_{m_1i_1}(\s)   Q^{ R_1 }_{ m_1 j_1} \otimes D^{R_2}_{m_2i_2}(\s)   Q^{ R_2 }_{ m_2 j_2}  \, . 
 \crcr
 &&
   \eea
We use the identity  
\be
\sum_{j_1,j_2}
D^{R_1}_{i_1 j_1}(\g)D^{R_2}_{i_2 j_2}(\g)C^{R_1,R_2; \, R_3,\,\tau}_{\, j_1,j_2; \, j_3}
 =\sum_{i_3 }  C^{R_1,R_2;\, R_3,\,\tau}_{\, i_1,i_2; \, i_3} \, D^{R_3}_{i_3 j_3}(\g)  
\label{ddc=cd} 
\ee
which holds for any  $\g \in S_n$ and expresses the fact that the Clebsch-Gordan coefficients intertwine the action of $ \g $ in $ R_1 \otimes R_2$ with the action in $ R_3$. 
We re-express the above as 
 \bea
 T_k^{(3)}   Q^{R_1, R_2, R_3}_ {\tau_1, \tau_2} &=& 
  \kappa'_{R_1, R_2}
     \sum_{\s \in \cC_k } 
\sum_{\s_1, \s_2 \in S_n}
\sum_{m_l, i_3,  j_l}
( \sum_{i_1,i_2 } 
D^{R_1}_{m_1i_1}(\s)D^{R_2}_{m_2i_2}(\s) 
C^{R_1, R_2; R_3 ,\tau_1  }_{ i_1 , i_2 ; i_3 }  ) 
\crcr
& \times & 
C^{R_1, R_2; R_3 , \tau_2  }_{ j_1 , j_2 ; i_3 } 
   Q^{ R_1 }_{ m_1 j_1} \otimes   Q^{ R_2 }_{ m_2 j_2} 
   \crcr
 &=& 
  \kappa'_{R_1, R_2}
     \sum_{\s \in \cC_k } 
\sum_{\s_1, \s_2 \in S_n}
\sum_{m_l, i_3,  j_l}
( \sum_{m_3} 
C^{R_1, R_2; R_3 , \tau_1  }_{ m_1 , m_2 ; m_3 }  
D^{R_3}_{m_3 i_3}(\s)  ) 
\crcr
& \times & 
C^{R_1, R_2; R_3 ,\tau_2  }_{ j_1 , j_2 ; i_3 } 
   Q^{ R_1 }_{ m_1 j_1} \otimes   Q^{ R_2 }_{ m_2 j_2} \crcr
& = & 
  \kappa'_{R_1, R_2}
\sum_{m_l, i_3,  j_l}
( \sum_{m_3} 
C^{R_1, R_2; R_3 , \tau_1  }_{ m_1 , m_2 ; m_3 }  
D^{R_3}_{m_3 i_3}( T_k )  ) 
\crcr
& \times & 
C^{R_1, R_2; R_3 , \tau_2  }_{ j_1 , j_2 ; i_3 } 
   Q^{ R_1 }_{ m_1 j_1} \otimes   Q^{ R_2 }_{ m_2 j_2} \crcr
& = & 
  \kappa'_{R_1, R_2}
\sum_{\s_1, \s_2 \in S_n}
\sum_{m_l, i_3,  j_l}
( \sum_{m_3} 
C^{R_1, R_2; R_3 ,\tau_1  }_{ m_1 , m_2 ; m_3 }  
{\chi_{R_2}( T_k ) \over d(R_3) } \delta_{m_3 i_3}  ) 
\crcr
& \times & 
C^{R_1, R_2; R_3 ,\tau_2  }_{ j_1 , j_2 ; i_3 } 
   Q^{ R_1 }_{ m_1 j_1} \otimes   Q^{ R_2 }_{ m_2 j_2}  \crcr
   &=& 
   {\chi_{R_3}( T_k ) \over d(R_3) }  
   \Big[
  \kappa'_{R_1, R_2}
\sum_{m_l, i_3,  j_l}
C^{R_1, R_2; R_3 , \tau_1  }_{ m_1 , m_2 ; i_3 }  
C^{R_1, R_2; R_3 ,\tau_2  }_{ j_1 , j_2 ; i_3 } 
   Q^{ R_1 }_{ m_1 j_1} \otimes   Q^{ R_2 }_{ m_2 j_2}\Big] 
 \eea
 where, at an intermediate step, we use again $D^{R_3}_{m_3 i_3}( T_k )
 = \alpha \delta_ {m_3 i_3}$ (with  $\alpha$ worked
 out in Lemma \ref{TQ}). This ends the proof the proposition. 
 
\qed
 
\section{GAP codes}
\label{app:gap}

In this appendix, we give several GAP functions that lead to the calculation of common nullspace
of the operators $T_2$  and $T_3$ at  rank 3  and for arbitrary $n$.  This code is computed with Sage
calling the GAP package. Hence the \verb+%%gap+
appearing at the beginning of each function.  Such command could be replaced by a single \verb+%gap+ 
depending on the environment. 
The comments, or  lines starting by \verb+#+, inside the code are self-explanatory and help to understand of each part of the 
current function. 

For $n \le 14$, $kmax \le 3$, here is the overall strategy of the calculation: 
\begin{enumerate}
\item generate the set of ribbon graphs (denoted as $ \Rib(n)$ in the bulk of the paper)  as  list 
with \verb+RibbSetFunction(n)+;  their number is $\ell$ (this is $|\Rib(n)| $  in the bulk of paper); 

\item construct  $T_2$ and $T_3$, and their different action
$T_2^{(i)}$ and  $T_3^{(i)}$, $i=1,2,3$ acting on different slots 
of pairs of permutation representing ribbon graphs ; 

\item calculate the number of time that a ribbon graph $b$ appears in the expansion of $T_2^{(i)} a $ or  $T_3^{(i)}a $, for all ribbon graph $a$; 

\item generate the $\ell \times \ell $-matrix $(L^{(i)}_l)_{ab}$, $l=2,3$, $i=1,2,3$, of all $T_2^{(i)}$ and  $T_3^{(i)}$; 

\item introduce the list of normalized characters $\wchi_R(T_p) := \chi_R(T_p)/d(R)$ using the formulae from \cite{Lasalle}   for $p:=2,3$; 

\item then solve the nullspace of the transpose of the stack of matrices $M^{(i)}_\ell (R) = L^{(i)}_\ell - \wchi_R(T_p) * Id_\ell$, $Id_\ell$ being the  $\ell \times \ell $-identity matrix. 
The dimension of this space is $C(R,S,T)^2$.

\item Alternatively, we generate the sequence of prime numbers
$a_{i,k}$  \eqref{aikgen} used in the construction of the Hamiltonian 
$\cH := \sum_{k=2}^{kmax} \sum_{i=1}^3 a_{i,k}T_k^{(i)}$ as
a matrix sum; and solve the nullspace of the sum of matrices
\be  
M^{(1)}_2(R) + a_{1} M^{(1)}_3(R)
+ a_{2}(M^{(2)}_2(S) +a_{1} M^{(2)}_3(S)) 
+ a_{3}(M^{(3)}_2(T) +a_{1}  M^{(3)}_3 (T))\,. 
\ee
\end{enumerate}

We compute the stack of matrices  $M^{(i)}_\ell (R) $ in Gap and solve for its  null space. 
There is a corresponding equation for $\cH$ that we also put in comments.

\

\noindent{\bf Code1: Generating all ribbons with $n$-edges.} The function
\verb+RibbSetFunction(n)+  returns the list of
bipartite ribbon graphs made $n$ ribbon edges and at most $n$ black vertices, 
 and at most $n$ white edges. 
We use the PCA formulation where 
each ribbon graph is represented by its equivalent class,  namely an orbit under diagonal 
$S_n$ group action on a pair of $S_n$ permutations $(\s_1,\s_2)$. 

{\small
\begin{verbatim} 
 %%gap
#-------------------------------------------------------------------------
# Function returning an ordered list of ribbon graphs - each ribbon 
# graph represented as a set of permutations within
# an Sn orbit by diagonal conjugation of pairs of permutations

RibbSetFunction := function( n )
    local G, Pairs, Ribb, RibbSets, a, b;
    G := SymmetricGroup(n);
    Pairs :=[];

    for a in G do
        for b in G do
            Add (Pairs, [a, b]);
        od;
    od;
    
    ## In Gap, group G action on list and groups is always by conjugation 
    # Ribb list of ribbons as G orbits on pairs (tau_1, tau_2)
    # OnPairs option of function Orbit means G acts diagonally on pairs
    # (g tau_1 g^(-1), g tau_2 g^(-1))
    
    Ribb := Orbits (G, Pairs, OnPairs);

    # RibbSets is now list of sets of pairs within an orbit, the order within
    # the set  does not matter
    RibbSets := [] ;
    for a in [ 1 .. Length ( Ribb ) ] do
        Add ( RibbSets , Set ( Ribb[a] ) ) ;
    od ;
    return RibbSets;
end;
\end{verbatim}
} 

\noindent{\bf Code2: Constructing $T$-operators.}
We   contruct the operators $T^{(i)}_l$. 

{\small
\begin{verbatim} 
%%gap
#--------------------------------------------------------------------------------
# Given a number n, and kmax 
# returns the (kmax -1)x3 arrays T( i )_{ a, b } - of size |RibSetLoc| * |RibSetLoc|.  
# at fixed i and p fixed this is the matrix of the action of T^(p)_(i) ; 
#  |RibSetLoc | is the size of the set of ribbons generated by RibbSetFunction(n); 
# for any a,  a ribbon graph orbit RibbSetLoc[a] is represented by RibSetLoc[a][1]. 
# Depending on  i,  we operate on the first permutation, the second or bot
# in RibSetLoc[a][1] and return b, namely the position in RibSetLoc of the outcome; 
# RibbSetLoc[a][1] is the representative perm of the ribbon RibbSetLoc[a]; 
# RibbSetLoc[a][1][1] is first projection of the representative RibbSetLoc[a][1];
# RibbSetLoc[a][1][2] is second projection of the representative RibbSetLoc[a][1];
# Position ( list , obj ) returns the position of the first occurrence of obj in list.
# Positions ( list , obj ) returns the number of occurrences of obj in list; 
#---------------------------------------------------------------------------------------------------------------------------------

ArrayTi := function (n, kmax)
    
    local G, a, b, i, j, l, s, Cic, cc, RibSetLoc,  L, L1, L2, L3, pos1, pos2, pos3 ;
    
    G := SymmetricGroup(n); 
    L := []; 
    Cic :=[];    

    # Construction of ribbon graphs
    RibSetLoc := RibbSetFunction(n);
    l := Length ( RibSetLoc );

       
    for i in [ 2 .. kmax] do 
        # Construction of Ti 
        Add ( Cic, Orbit (G, CycleFromList([1 .. i]), OnPoints) );
        
        # Construction of the k-1 lists of listsLpt; 
        # Lpt[i-1] is a list of 3 matrices Lpt[i-1][p] initialized at 0; 
        # Each Lpt[i-1][p] defines the operator Ti^p with action of Ti in the slot p
        Add(L, []); 
        for j in [1 .. 3] do
            Add( L[i-1], NullMat ( l , l ) ) ;; 
        od; 
    od; 
    
    for a in [ 1 .. l ] do 
        for i in [2 .. kmax] do
            # empty the Li
            L1 := []; 
            L2 := []; 
            L3 := []; 
            cc := Cic[i-1];
            for s in cc do
                pos1 := Position ( RibSetLoc , Set ( Orbit ( G , [ s *RibSetLoc [a][1][1] , 
                					RibSetLoc[a ][1][2] ] , OnPairs ) ) ) ;
                pos2 := Position ( RibSetLoc , Set ( Orbit ( G , [ RibSetLoc [a][1][1] , 
                					s *RibSetLoc[a ][1][2] ] , OnPairs ) ) ) ; 
                pos3 := Position ( RibSetLoc , Set ( Orbit ( G , [ s *RibSetLoc [a][1][1] , 
                					s *RibSetLoc[a ][1][2] ] , OnPairs ) ) ) ; 
                Add( L1 , pos1);
                Add( L2 , pos2);
                Add( L3 , pos3);
            od;
            for b in [ 1 .. l ] do 
                L[i-1][1][ b, a ] := Length ( Positions ( L1 , b ) ) ; 
                L[i-1][2][ b, a ] := Length ( Positions ( L2 , b ) ) ; 
                L[i-1][3][ b, a ] := Length ( Positions ( L3 , b ) ) ; 
            od;
        od; 
    od;
    return L ;
end ;
\end{verbatim}
}

\noindent{\bf Code3: Characters.}
The following function returns the table list of characters according to Lassalle formulae
\cite{Lasalle}. 

{\small
\begin{verbatim} 
%%gap
# Returns the list of characters of T2 (on the left) and T3 (on the right) 
# via Lassalle formulae. These characters are the eigenvalues of T's. 

CharactersEigenvalues_of_Top := function( m )
    local i, j, k, p, L, L2, L3, sum2, sum3;
    p := Partitions( m ) ;
    sum2 := 0;
    sum3 := 0;
    L2 := [];
    L3 := [];
    L := [];
    for i in [ 1..Length(p) ] do
        sum2 := 0;
        sum3 := 0;
        for j in [1..Length(p[i])] do
            for k in [1..p[i][j]] do
                sum2 := sum2 - j +  k ;
                sum3 := sum3 +(k-j)*(k-j);
            od;
        od;
        sum3:= sum3 - Factorial( m )/( Factorial( m-2 )*2 ) ;
        Add( L2, sum2 );
        Add( L3, sum3 );
        Add( L, [ L2[i], L3[i] ] );
    od;
    return L;
end;
\end{verbatim}
}

\noindent{\bf Code4: Primes.}
The two functions returns either a couple $(a_1,a_2)$ 
or a triple $(a_1,a_2, a_3)$ of prime numbers  are used in the construction
of the total Hamiltonian and insure that it cannot vanishes
outside of the required values. This follows the sufficient conditions explained in 
section \ref{primes}.

{\small
\begin{verbatim} 
%%gap 
# The code produces a couple that makes the QM Hamiltonian not vanishing 
# unless the triples (R1,R2,R3) = (R1',R2',R3')

CouplePrime := function(m) 
    local p1, L; 
    p1 := NextPrimeInt( m*(m-1) );
    L := []; 
    Add(L, p1); 
    Add(L, NextPrimeInt( (p1+1)*m*(m-1))); 
    return  L ;
end; 

# The code produces a tripe that makes the QM Hamiltonian not vanishing 
# unless the triples (R1,R2,R3) = (R1',R2',R3')

CouplePrime2 := function(n) 
    local p1, p2, p3, L; 
    p1 := NextPrimeInt( n*(n-1) );
    p2 := NextPrimeInt( Int( (n/3)*(n-1)*(3 + 2*p1*(n-2))) )  ; 
    p3 := NextPrimeInt( Int( (n/3)*(n-1)*(1+p2)*( 3 + 2*p1*(n-2) )) ) ; 
    L := []; 
    Add(L, p1); 
    Add(L, p2 ); 
    Add(L, p3); 
    return  L ;
end; 

\end{verbatim}
}

\noindent{\bf Code5: Matrices and null spaces.}
We are now in position to address the nullspace  of the multiple actions of  $T_2$ and $T_3$. 
The following code constructs the stack of matrices (note that, in comments, we also
give instructions to construct the total Hamiltonian calling prime numbers)  made of 3 
(resp. 6) matrices determined by three Young Diagram $R,S,T$, given $kmax = 2$ (resp. $kmax = 3$). 
After the construct, it returns the null space of the resulting matrix. 

{\small
\begin{verbatim} 
%%gap
MatForNullVectors := function(m,  kmax, R, S,T)
    local a, l, chi , M1, M2, M3, M5, M4, M6, Arr, Id;
    
    # for the total Hamiltonian version uncomment the following
    #local a, l, chi , M1, M2, M3, M5, M4, M6, Arr, Id, c,  c2, Ham;
    
    l:=Length(RibbSetFunction(m)); # cardinality of the ribbon set
    chi := CharactersEigenvalues_of_Top(m);
    
    # for the Hamiltonian uncomment the following
    #c := CouplePrime(m);
    #c2 := CouplePrime2(m); 
    
    #Constructing the matrices
    Arr := ArrayTi(m , kmax  ); 
    
    if kmax > 3 then
        Print ("kmax >3"); 
        return 0; 
    fi; 
    
    if kmax = 2 then 
        M1 := Arr[1][1] -  chi[R][1]  * IdentityMat ( l );
        M3 := Arr[1][2] - chi[S][1]  * IdentityMat ( l );
        Append(M1, M3);
        
        # For the Hamiltonian uncomment the following 
        # Ham := M1 + c[1]*M3 ; 
        
        M5 :=  Arr[1][3] - chi[T][1]  * IdentityMat ( l ); 
        Append(M1, M5);
        
        # For the Hamiltonian uncomment the following 
        #Ham := Ham + c[2]*M5 ;
    fi; 
    
    
    if kmax = 3 then 
        M1 := Arr[1][1] -  chi[R][1]  * IdentityMat ( l );
        M2 := Arr[2][1]  - chi[R][2] * IdentityMat ( l );
        Append(M1, M2);
        
        #  For the Hamiltonian uncomment the following 
        # Ham := M1 + c2[1]*M2 ;
        
        M3 :=  Arr[1][2] - chi[S][1]  * IdentityMat ( l );
        Append(M1, M3);
        M4 :=  Arr[2][2] - chi[S][2]  * IdentityMat ( l );
        Append(M1, M4);
        
        #  For the  Hamiltonian uncomment the following 
        # Ham := Ham +  c2[2]*(M3 + c2[1]*M4) ;
         
        M5 :=  Arr[1][3] - chi[T][1]  * IdentityMat ( l );
        Append(M1, M5);
        M6 :=  Arr[2][3] - chi[T][2]  * IdentityMat ( l );
        Append(M1, M6);
        
        #  For the  Hamiltonian uncomment the following 
        # Ham := Ham + c2[3]*(M5 + c2[1]*M6) ; 
    fi; 
    
    
    return NullspaceIntMat(TransposedMat (M1) );
    #  For the total Hamiltonian uncomment the following 
    #return NullspaceMat(TransposedMat ( Ham) ) ;
end;
\end{verbatim}
}

\end{document}